\pdfoutput=1
\documentclass[11pt]{article}
\usepackage{amsfonts}
\usepackage{mathrsfs}
\usepackage{amssymb,amsmath,amsthm}
\usepackage{graphicx}
\usepackage[dvips, paper=letterpaper, top=0.95in, bottom=.6in, left=0.9in, right=0.95in, nohead, includefoot, footskip=.4in]{geometry}
\usepackage[authoryear,round]{natbib}
\usepackage{environ}
\usepackage{color}
\usepackage{epsfig}
\usepackage{caption}
\usepackage{verbatim}
\usepackage{subcaption}
\usepackage{float}
\usepackage{setspace}
\usepackage{enumitem}
\setlist[enumerate]{itemsep=0mm}
\usepackage{calc}
\usepackage[pdfpagemode=UseOutlines,hyperfootnotes=false, plainpages=false]{hyperref}
\usepackage{mathtools}
\usepackage{dsfont}
\usepackage[english]{babel}
\usepackage{amsthm}
\usepackage{bm}
\usepackage{breqn}
\usepackage{bigints}
\usepackage{tikz}
\usepackage{nicefrac}
\usepackage{xurl}
\usepackage{algorithmic}
\usepackage[ruled,vlined]{algorithm2e}

\usepackage[font=normalsize,labelfont=bf]{caption}
\usepackage[font=footnotesize,labelfont=bf]{subcaption}
\usepackage{sectsty}
\allsectionsfont{\normalfont\sffamily\bfseries}
\sffamily

\newtheoremstyle{theoremsansserif} 
    {\topsep}                    
    {\topsep}                    
    {\itshape}                   
    {}                           
    {\sffamily\bfseries }        
    {.}                          
    {.5em}                       
    {}  

\theoremstyle{theoremsansserif}

\newtheorem{lemma}{Lemma}
\newtheorem{corollary}{Corollary}
\newtheorem{proposition}{Proposition}
\newtheorem{remark}{Remark}
\newtheorem{definition}{Definition}

\newtheorem*{example*}{Example}

\newtheorem{theorem}{Theorem}

\newcommand{\R}{\mathbb{R}}
\newcommand{\D}{\tilde{D}}
\newcommand{\PP}{\mathcal{P}}
\renewcommand{\S}{\mathcal{S}}

\newcommand{\X}{\mathcal{X}}

\newcommand{\E}{\mathbb{E}}
\newcommand{\tA}{\zeta_A}

\newcommand{\T}{\mathscr{T}}

\DeclareMathOperator{\argmin}{argmin}
\DeclareMathOperator{\argmax}{argmax}

\DeclareMathOperator{\opt}{OPT}
\begin{document}
\title{\sf\textbf{Online Resource Allocation under Horizon Uncertainty}}
\author{
\sf Santiago R. Balseiro\\
\sf Columbia University \\
\small\texttt{srb2155@columbia.edu}
\and
\sf Christian Kroer \\
\sf Columbia University\\
\small\texttt{christian.kroer@columbia.edu}
\and
\sf Rachitesh Kumar\\
\sf Columbia University \\
\small\texttt{rk3068@columbia.edu}}

\date{\vspace{1em}
\sf This version: \today
}

\maketitle

\allowdisplaybreaks

\begin{abstract}
	We study stochastic online resource allocation: a decision maker needs to allocate limited resources to stochastically-generated sequentially-arriving requests in order to maximize reward. At each time step, requests are drawn independently from a distribution that is unknown to the decision maker. Online resource allocation and its special cases have been studied extensively in the past, but prior results crucially and universally rely on the strong assumption that the total number of requests (the horizon) is known to the decision maker in advance. In many applications, such as revenue management and online advertising, the number of requests can vary widely because of fluctuations in demand or user traffic intensity. In this work, we develop online algorithms that are robust to horizon uncertainty. In sharp contrast to the known-horizon setting, no algorithm can achieve even a constant asymptotic competitive ratio that is independent of the horizon uncertainty. We introduce a novel generalization of dual mirror descent which allows the decision maker to specify a schedule of time-varying target consumption rates, and prove corresponding performance guarantees. We go on to give a fast algorithm for computing a schedule of target consumption rates that leads to near-optimal performance in the unknown-horizon setting.
	In particular, our competitive ratio attains the optimal rate of growth (up to logarithmic factors) as the horizon uncertainty grows large. Finally, we also provide a way to incorporate machine-learned predictions about the horizon which interpolates between the known and unknown horizon settings.
\end{abstract}


\setstretch{1.5}

\section{Introduction}

Online resource allocation is a general framework that includes as special cases various fundamental problems like network revenue management \citep{talluri2004theory}, online advertising \citep{mehta2013online}, online linear/convex programming \citep{agrawal2014fast,agrawal2014dynamic, devanur2011near, kesselheim2014primal}, bidding in repeated auctions \citep{balseiro2019learning}, and assortment optimization under inventory constraints~\citep{golrezaei2014real}. It captures any setting in which a decision-maker endowed with a limited amount of resources faces sequentially-arriving requests, each of which consumes a certain amount of resources and generates a reward. At each time step, the decision maker observes the request and then takes an action, with the overarching aim of maximizing cumulative reward subject to resource constraints. In this work, we focus on the setting in which requests are generated by some stationary distribution unknown to the decision maker. We impose very mild assumptions on this distribution, and in particular do not require the requests to satisfy common convexity assumptions.


To the best of our knowledge, all of the previous works on stochastic online resource allocation (with the exception of the concurrent work of \citealt{brubach2019online}, \citealt{bai2022fluid}, and \citealt{aouad2022nonparametric}) assume that the total number of requests (horizon) is known to the decision maker. Importantly, this assumption is vital for previous algorithms and performance guarantees because it allows them to compute a \emph{per-period resource budget} (total amount of resources divided by total number of requests) and use it as the target consumption in each time step. However, if the total number of requests is not known to the decision maker, one can no longer compute this quantity and these previous works fail to offer any guidance.

Juxtapose this with a world in which viral trends are becoming ever more common, causing online advertising platforms, retailers and service providers to routinely experience traffic spikes. These spikes inject uncertainty into the system and make it difficult to accurately predict the total number of requests that will arrive. In fact, these spikes present lucrative opportunities for the advertiser/retailer, which makes addressing the uncertainty even more pertinent~\citep{esfandiari2015online}. Moreover, it is usually difficult to predict these spikes, e.g. a news story breaks about COVID-travel bans being lifted, which results in a sudden and large uptick in the number of advertising opportunities for an airline. In fact, search-traffic spikes might be so large that they cause websites to crash\footnote{\url{https://developers.google.com/search/blog/2012/02/preparing-your-site-for-traffic-spike}}. This motivates us to relax the previously-ubiquitous known-horizon assumption and address this omission in the literature by developing algorithms which are robust to horizon uncertainty. We use as our benchmark the hindsight optimal allocation that can be computed with full knowledge of all requests and the time horizon.

\subsection{Main Contributions}

\textbf{Impossibility Results.} If no assumptions are made about the horizon, it has been shown in the context of matching \citep{brubach2019online} and prophet inequalities \citep{alijani2020predict} that no algorithm can guarantee a positive constant fraction of the hindsight optimal reward. Contrast this with the online convex optimization literature, where one can easily obtain a low-regret algorithm (i.e., asymptotic competitive ratio of one) for the unknown horizon setting  by applying the doubling trick to a low-regret algorithm for the known-horizon setting \citep{shalev2012online, hazan2016introduction}\footnote{Specifically, the doubling trick involves repeatedly running the low-regret algorithm for the known-horizon setting with increasing lengths of horizons. This is possible because the online convex optimization problem decouples across time into independent subproblems. Such decomposition is not possible in our problem because of the resource constraints: resources that have been consumed in the past restrict future actions of the algorithm.}. In this paper, we prove a new impossibility result for the setting when the horizon $T$ is constrained to lie in an uncertainty window $[\tau_1,\tau_2]$ that is known to the decision maker, which is the mildest possible assumption that renders the problem interesting. In the uncertainty-window setting, we show that no online algorithm can achieve a greater than $\tilde O\left(\ln(\tau_2/\tau_1)^{-1} \right)$ fraction of the hindsight optimal reward (Theorem~\ref{thm:upper-bound}). This upper bound holds even when (i) there is only 1 type of resource, (ii) the decision maker receives the same request at each time step, (iii) this request is known to the decision maker ahead of time, (iv) the request has a smooth concave reward function and linear resource consumption, (v) $\tau_1$ is arbitrarily large, and (vi) the initial resource endowment $B = \Theta(\tau_1)$ scales with the horizon. In particular, unlike the known-horizon setting, vanishing regret is impossible to achieve under horizon uncertainty, leading us to focus on developing algorithms with a good asymptotic competitive ratio (fraction of the hindsight optimal reward).

\begin{figure}[t!]
    \centering
    \includegraphics[width = 0.4\textwidth]{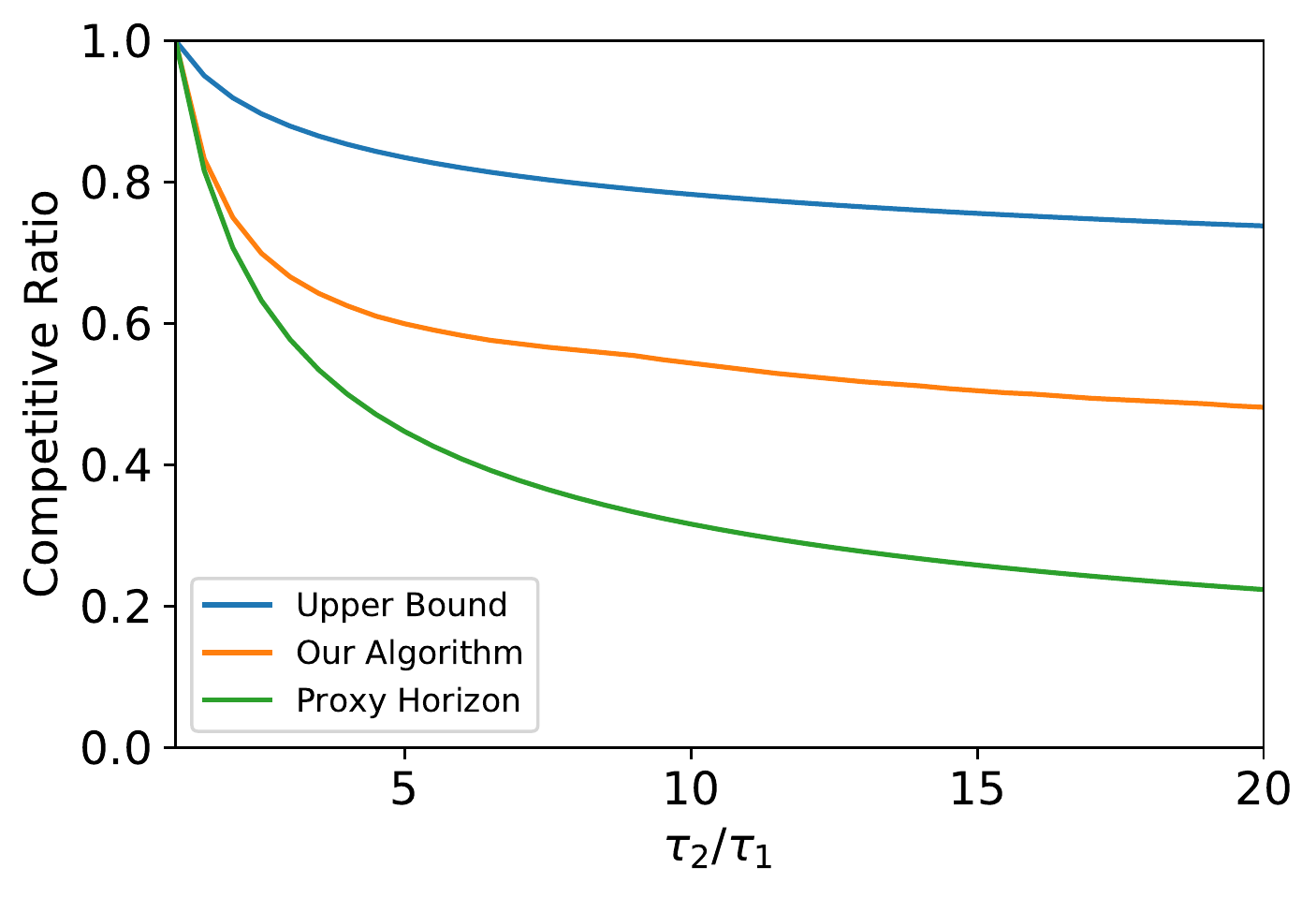}
    \caption{Plot with (i) our upper bound on the best-possible competitive ratio (Theorem~\ref{thm:upper-bound}) which scales as $\tilde{O}\left(\ln(\tau_2/\tau_1)^{-1} \right)$, (ii) the (asymptotic) competitive ratio of our algorithm which scales as $\Omega\left(\ln(\tau_2/\tau_1)^{-1} \right)$ (Algorithm~\ref{alg:dual-descent} with target sequence from Algorithm~\ref{alg:opt-target-sequence}), and (iii) an upper bound on the competitive ratio of algorithms that are optimal for the known-horizon setting when used with some proxy horizon $T^* \in [\tau_1,\tau_2]$ (Section~\ref{sec:previous-alg}), which scales as $\left(\sqrt{\tau_2/\tau_1} \right)^{-1}$. Even for small values of $\tau_2/\tau_1$, our algorithm significantly outperforms previous ones.}
    \label{fig:us-vs-previous}
\end{figure}

\textbf{Variable Target Dual Mirror Descent.} Dual mirror descent is a natural algorithm for the known-horizon setting introduced by \citet{balseiro2020best}, who build on a long line of primal-dual algorithms for online allocation problems \citep{agrawal2014fast, devanur2011near, gupta2016experts}. It maintains a price (i.e., dual variable) for each resource and then dynamically updates them with the goal of consuming the per-period resource budget at each step---if the resource is being over-consumed, increase its price; and vice-versa. As stated earlier, this approach fails if the horizon is not known because the per-period budget cannot be computed ahead of time. A natural approach to handle horizon uncertainty is to use dual mirror descent with some proxy horizon $T^* \in [\tau_1,\tau_2]$ in the hopes of getting good performance for all $T \in [\tau_1,\tau_2]$. Unfortunately, as we show in Section~\ref{sec:previous-alg}, this approach can be extremely suboptimal, not just for dual mirror descent but for any algorithm which is optimal for the known-horizon setting. Thus, the unknown-horizon setting calls for new algorithms. Our main insight is that, even though one cannot compute the per-period resource budget and target its consumption, it is possible to compute a time-varying sequence of target consumptions which, if consumed at those rates, perform well no matter what the horizon turns out to be. To achieve this, we develop \emph{Variable Target Dual Mirror Descent} (Algorithm~\ref{alg:dual-descent}), which takes a sequence of target consumptions as input and dynamically updates the prices to hit those targets. One of our primary technical contributions is generalizing the analysis of dual mirror descent to develop a fundamental bound that allows for general target consumption sequences. We leverage this bound to show that there exists a simple time-varying target consumption sequence which can be described in closed form and achieves a near-optimal $\Omega \left(\ln(\tau_2/\tau_1)^{-1} \right)$ asymptotic competitive ratio when deployed with Algorithm~\ref{alg:dual-descent}, matching the upper bound up to logarithmic factors.

\textbf{Optimizing the Target Sequence and Incorporating Predictions.} Variable Target Dual Mirror Descent reduces the complex problem of finding an algorithm which maximizes the competitive ratio to the much simpler problem of finding the optimal target consumption sequence. We develop an algorithm to solve the latter efficiently (Algorithm~\ref{alg:opt-target-sequence}), leading to substantial gains over previous algorithms even for small values of $\tau_2/\tau_1$ (see Figure~\ref{fig:us-vs-previous}). Importantly, Algorithm~\ref{alg:opt-target-sequence} does not require one to solve computationally-expensive linear programs (LPs), which can be desirable in time-sensitive applications.  We then use the Algorithms-with-Predictions framework~\citep{mitzenmacher2020algorithms} to study incorporating (potentially inaccurate) predictions about the horizon with the goal of performing well if the prediction comes true, while also ensuring a good competitive ratio no matter what the horizon turns out to be. We show that the problem of computing the optimal target consumption sequence for the goal of optimally incorporating predictions can also be solved efficiently using Algorithm~\ref{alg:opt-target-sequence}. Our algorithm allows the decision maker to account for the level of confidence she has in the predictions, and smoothly interpolate between the known-horizon and uncertainty-window settings.

\subsection{Additional Related Work}\label{related-work}

Online resource allocation is an extremely general framework that captures a wide range of problems, which together have received significant attention in the past. Here, we focus on works that present results which hold at the level of generality of online linear packing or higher, and refer the reader to \citet{mehta2007adwords} and \citet{balseiro2021survey} for a discussion of other special sub-problems like Adwords, network revenue management, repeated bidding in auctions, online assortment optimization etc.

Online linear packing is a special case of online resource allocation in which the reward and consumption functions are linear. It has been studied in the random permutation model, which assumes that the requests arrive in a random order and is slightly more general than the i.i.d. stochastic arrival assumption. Most of the results on online linear packing described here hold for this more general model. \citet{devanur2009adwords} and \citet{feldman2010online} present an algorithm that uses the initial requests to learn the dual variables and then subsequently uses these to make accept/reject decisions. It can be shown to have a regret of $O(T^{2/3})$. \citet{agrawal2014dynamic} improved this to $O(\sqrt{T})$ regret by repeatedly re-learning the dual variables by solving LPs at intervals of geometrically-increasing lengths. \citet{devanur2011near} and \citet{kesselheim2014primal} provide algorithms with $O(\sqrt{T})$ regret but better dependence of the constants on the number of resources. \citet{li2020simple} give a dual-stochastic-gradient-based algorithm which achieves $O(\sqrt{T})$ regret without solving computationally-expensive LPs, albeit with a slightly worse dependence on the number of resources. For the special case of network revenue management with finitely many types where the arrival probability of each type and the per-period budget are apriori known to the decision maker, \citet{freund2019good} provided a primal resolving algorithm which attains constant regret when the exact horizon $T$ is revealed $T^{0.5 + \epsilon}$ time periods before the final request. Importantly, if only the total budget is known, then their algorithm needs to know the horizon to compute the per-period budget.

\cite{agrawal2014fast} study online resource allocation with concave rewards and a convex constraint set. They describe a dual-based algorithm which achieves $O(\sqrt{T})$ regret. The work that is most closely related to ours is \citet{balseiro2020best}. It studies the fully-general online resource allocation under a variety of assumptions on the request-generation process including stochastic i.i.d. They present a dual mirror descent algorithm and show that it is optimal for a variety of assumptions on the requests. In particular, their algorithm achieves a $O(\sqrt{T})$ regret when requests are stochastic. In light of the lower bound of $\Omega(\sqrt{T})$ proven by \citet{arlotto2019uniformly} for the multi-secretary problem, which is a special case of online resource allocation, their result is tight in its dependence on $\sqrt{T}$. \citet{jiang2020online} study a more general model that allows for time-varying request distributions. They propose a dual gradient descent algorithm which works with target consumptions computed using prior information about future request distributions. When the prior information is close to being accurate (measured using Wasserstein distance), their algorithm attains $O(\sqrt{T})$ regret against the hindsight optimal. Their algorithm reduces to that of \citet{balseiro2020best} with the Euclidean regularizer when the request distributions and priors are identical. Our algorithm (Algorithm~\ref{alg:dual-descent}) generalizes their algorithm to allow for arbitrary regularizers. Importantly, their guarantees for dual gradient descent with time-varying target consumptions only hold when the prior is close to being accurate, and not for general target consumption sequences like the ones we develop (Theorem~\ref{thm:regret}), which lie at the heart of our results. Crucially, all of the aforementioned results rely on knowledge of the horizon (total number of requests) $T$, and no longer hold in its absence (see Theorem~\ref{thm:upper-bound}). Moreover, naive attempts to extend them to the unknown horizon setting lead to sub-optimal algorithms (see Section~\ref{sec:previous-alg}).


There is also a line of work studying online allocation problems when requests are adversarially chosen. Naturally, the fully-adversarial model subsumes our input model, in which requests are drawn i.i.d.~from an unknown distribution and the horizon is uncertain. Therefore, guarantees for adversarial algorithms carry over to our setting. We remark, however, that it is not possible to obtain bounded competitive ratio for the general online allocation problems we study in this paper (see, e.g., \citealt{feldman2009free}). Notable exceptions are online matching~\citep{karp1990optimal}, the AdWords problem~\citep{mehta2007adwords}, or personalized assortment optimization~\citep{golrezaei2014real}, which are linear problems in which rewards are proportional to resource consumption. When rewards are not proportional to resource consumption, there is a stream of literature studying algorithms with parametric competitive ratios. These competitive ratios either depend on the range of rewards (see, e.g., \citealt{ball2009toward,ma2020algorithms}) or the ratio of budget to resource consumption (see, e.g., \citealt{balseiro2020best}). Our work leverages the fact that requests are drawn i.i.d.~from an unknown distribution to derive stronger competitive ratios that only depend on horizon uncertainty (and that are independent of all other problem parameters). In particular, our competitive ratios are bounded, while the parametric competitive ratios of the aforementioned papers can be made arbitrarily large if an adversary can choose the distribution of requests.

Finally, a few very recent papers warrant attention, all of which assume that the distribution of the horizon is known in advance. \citet{brubach2019online} study a generalization of online bipartite matching which accounts for ranked preferences over the offline vertices under a variety of input models. They show that a constant competitive ratio cannot be attained under stationary stochastic input when the horizon is completely unknown and use it to justify the known-horizon assumption. Our impossibility result (Theorem~\ref{thm:upper-bound}) establishes a parametrized upper bound on the competitive ratio in terms of the uncertainty $\tau_2/\tau_1$ and implies their result as a special case when $\tau_2/\tau_1 \to \infty$. \citet{alijani2020predict} study the multi-unit prophet-inequality problem in which the resource is perishable, with each unit of the resource exiting the system independently at some time whose distribution is known to the decision maker. When there is one unit of the resource, their model captures horizon uncertainty in the prophet-inequality problem, which is a special case of online resource allocation. Importantly, when there is more than one unit, our models are incomparable. For the single-unit special case, they prove a parameterized upper bound of $\tilde O\left( \ln(\tau_2/\tau_1)^{-1}\right)$ on the competitive ratio. In contrast, our upper bound (Theorem~\ref{thm:upper-bound}) holds for the more general regime where the initial resource endowment (number of units of the resource) scales linearly with the horizon and the action space is continuous. This is crucial since the performance guarantees of algorithms for online resource allocation with known-horizon often only hold in this regime \citep{balseiro2020best, mehta2007adwords, talluri2004theory}, thereby making the single-unit upper bound inapplicable.

\citet{bai2022fluid} develop a fluid approximation to the dynamic-programming solution for network revenue management when both the distribution of the request and the horizon are completely known. They show that the asymptotically-tight fluid approximation should attempt to respect the resource constraint for all possible horizon values and not in expectation over the horizon. \citet{aouad2022nonparametric} consider a model for network revenue management in which the distribution of the horizon is known and each type of request follows an adversarial or random-order arrival pattern. They also show that the fluid LP relaxation based on the expected value of the horizon can be arbitrarily bad and develop tighter LP relaxations. We do not assume that the type of requests, the distribution of requests or the distribution of the horizon are known ahead of time, and use the hindsight optimal allocation as the benchmark, making our results incomparable even for the special case of network revenue management.

\section{Model}\label{sec:model}

\paragraph{Notation} We use $\R_+$ to denote the set of  non-negative real numbers and $\R_{++}$ to denote the set of positive real numbers. We use $a^+$ to denote $\max\{a,0\}$. For a vector $v \in \R^m$ and a scalar $a \in \R \setminus \{0\}$, $v/a$ denotes the scalar multiplication of $v$ by $1/a$. For vectors $u,v \in \R^m$ and a relation $R \in \{\geq , >, \leq, <\}$, we write $u\ R\ v$ whenever the relation holds component wise: $u_i\ R\ v_i$ for all $i \in [m]$.

We consider a general online resource allocation problem with $m$ resources, in which requests arrive sequentially. At time $t$, a request $\gamma_t = (f_t, b_t, \X_t)$ arrives, which is composed of a reward function $f_t: \X_t \to \R_+$, a resource consumption function $b_t: \X_t \to \R^m_{+}$ and a compact action set $\X_t \subset \R_+^d$. We assume that $0 \in \X_t$ and $b_t(0) = 0$ for all $t$. This ensures that the decision maker has the option to not spend any resources at each time step. We make no assumptions about the convexity/concavity of either $f_t$, $b_t$ or $\X_t$.

Let $\S$ represent the set of all possible requests. We make standard regularity assumptions: there exist constants $\bar f, \bar b \in \R_+$ such that, for every request $\gamma = (f, b, \X) \in \S$, we have $|f(x)| \leq \bar f$ and $\|b(x)\|_\infty \leq \bar b$ for all $x \in \X$. Furthermore, we assume that the requests are drawn i.i.d. from some distribution $\PP$ over $\S$, both of which are not assumed to be known to the decision maker. The decision maker has a known initial resource endowment (or budget) of $B = (B_1, \dots, B_m) \in \R_{++}^m$, where $B_i$ denotes the initial amount of resource $i$ available to the decision maker. We will assume that $B_i \geq 2 \bar b$ for all $i \in [m]$.

Let $T$ denote the total number of requests that will arrive (also called the horizon). We will use $\rho_T = B/T$ to denote the per-period resource endowment that is available to the decision maker. In contrast to previous work, we do not assume that $T$ (or its distribution) is known to the decision maker. Looking ahead, this uncertainty is what makes our problem much harder than vanilla online resource allocation where the horizon is known, as evidenced by the fact that no algorithm can attain even a positive competitive ratio when nothing is known about the horizon $T$ (see Theorem~\ref{thm:upper-bound}), which is a far-cry from the no-regret property exhibited by algorithms for the known-horizon setting.

At each time step $t \leq T$, the following sequence of events take place: (i) A request $\gamma_t = (f_t, b_t, \X_t)$ arrives and is observed by the decision maker; (ii) The decision maker selects an action $x_t \in \X_t$ from the action set based on the information seen so far; (iii) The decision maker receives a reward of $f_t(x_t)$ and the request consumes $b_t(x_t)$ resources. The goal of the decision maker is to take actions that maximize her total reward while keeping the total consumption of resources below the initial resource endowment. More concretely, an online algorithm (for the decision maker) chooses an action $x_t \in \X_t$ at each time step $t \leq T$ based on the current request $\gamma_t = (f_t, b_t, \X_t)$ and the history observed so far $\{\gamma_s, x_s \}_{s=1}^t$ such that the resource constraints $\sum_{t=1}^T b_t(x_t) \leq B$ are satisfied almost surely w.r.t. $\vec{\gamma} \sim \PP^T$. Our results continue to hold even if the actions $\{x_t\}_t$ are randomized, but we work with deterministic actions for ease of exposition. Since we assume that $T$ is not known to the decision maker, the actions of the online algorithm cannot depend on $T$. The total reward of algorithm $A$ on a sequence of requests $\vec{\gamma} = (\gamma_1, \dots, \gamma_T)$ is denoted by $R(A|T, \vec{\gamma}) = \sum_{t=1}^T f_t(x_t)$.

We measure the performance of an algorithm for the decision maker by comparing it to the hindsight optimal solution computed with access to all the requests and the value of $T$.  More concretely, for a horizon $T$ and a sequence of requests $\vec{\gamma} = (\gamma_1, \dots, \gamma_T)$, the hindsight-optimal reward $\opt(T, \vec{\gamma})$ is defined as the optimal value of the following hindsight optimization problem:
\begin{align}\label{eq:opt}
  \opt(T, \vec{\gamma}) \coloneqq\max_{x \in \prod_t \X_t}\ \sum_{t=1}^T f_t(x_t) \quad \text{subject to} \quad \sum_{t=1}^T b_t(x_t) \leq B\,.
\end{align}
We define the performance ratio of an algorithm $A$ for horizon $T$ and request distribution $\PP$ as
\begin{align*}
  c(A|T,\PP) \coloneqq \frac{\E_{\vec\gamma \sim \PP^T}\left[R(A|T, \vec \gamma) \right]}{\E_{\vec\gamma \sim \PP^T}\left[\opt(T, \vec\gamma)\right]} \,.
\end{align*}

Throughout this paper, we will assume that the horizon $T$ belongs to an uncertainty window $[\tau_1, \tau_2]$ which is known to the decision maker. This assumption is necessary because it is impossible to attain non-trivial performance guarantees in the absence of an upper bound on the horizon (see Theorem~\ref{thm:upper-bound}). Moreover, we will assume that there exists a constant $\kappa > 0$ such that $\E[\opt(T, \vec\gamma)] \geq \kappa \cdot T$ for all $T \in [\tau_1, \tau_2]$. The assumption that $\E[\opt(T, \vec\gamma)] = \Omega(T)$ is common in the literature on online resource allocation with bandit feedback (see \citealt{slivkins2019introduction} for a survey). A mild sufficient condition for this assumption to hold is the existence of some mapping from request to actions which achieves positive expected reward: $\exists$ $x:\S \to \X$ such that $\E_{\gamma \sim \PP}[f(x(\gamma))] > 0$.\footnote{To see how this, define $\psi \coloneqq \E_{\gamma \sim \PP}[f(x(\gamma))] > 0$ and set $x_t = x(\gamma_t)$ starting from $t=1$ till some resource runs out. Since $\|b(x)\|_\infty \leq \bar b$, resource $j$ will last at least $\lfloor B_j/\bar b \rfloor$ time steps, which in combination with $B \geq 2 \bar b$ implies
\begin{align*}
  \E_{\vec\gamma} [\opt(T,\vec\gamma)] \geq \min_{j \in [m]} \lfloor B_j/\bar b \rfloor \cdot \E_{\gamma \sim \PP}[f(x(\gamma))] \geq \min_{j \in [m]} \left( \frac{B_j}{\bar b} - 1 \right) \psi \geq \min_{j \in [m]} \frac{B_j \psi}{2\bar b} =  \min_{j \in [m]} \frac{\rho_{T,j} \psi}{2\bar b} \cdot T \,.
\end{align*}
Since $\rho_T \geq \rho_{\tau_2}$ for all $T \leq \tau_2$, setting $\kappa = \min_j \rho_{\tau_2,j} \psi/2\bar b$ yields $\E[\opt(T, \vec\gamma)] \geq \kappa \cdot T$ for all $T \in [\tau_1,\tau_2]$.
}

Next, we describe the models of horizon uncertainty we consider in this paper.

\textbf{Uncertainty Window.} Here, we assume that the decision maker is not aware of the exact value of $T$ and it can lie anywhere in the known uncertainty window $[\tau_1, \tau_2]$. This approach is motivated by the literature on robust optimization, where it is often assumed that the exact value of the parameter is unknown but it is constrained to belong to some known uncertainty set~\citep{ben2002robust}. Our goal here is to capture settings with large horizon uncertainty where it is difficult to predict the total number of requests with high accuracy. In such settings, it is often easier to generate confidence intervals than precise estimates. For this model of horizon uncertainty, we measure the performance of an algorithm $A$ by its competitive ratio $c(A)$, which we define as
\begin{align*}
  c(A) \coloneqq \inf_{\PP} \min_{T \in [\tau_1,\tau_2]} c(A|T, \PP)  \,.
\end{align*}
We also say that an algorithm $A$ is $c(A)$-competitive if it has a competitive ratio of $c(A)$. The competitive ratio of our algorithm degrades at a near-optimal logarithmic rate as $\tau_2/\tau_1$ grows large, and consequently yields good performance even for conservative estimates of the uncertainty window.


\textbf{Algorithms with Predictions.} We also consider a model of horizon uncertainty, inspired by the Algorithms-with-Predictions framework, which interpolates between the previously studied known-horizon model and the uncertainty-window model described above. This framework assumes that the decision maker has access to a prediction $T_P \in [\tau_1, \tau_2]$ about the horizon but no assumptions are made about the accuracy of this prediction. In particular, the goal is to develop algorithms that perform well when the prediction is accurate (consistency) while maintaining worst-case guarantees (competitiveness). For this setting, our algorithm allows the decision maker to smoothly trade-off consistency and competitiveness depending on her preferences.


\subsection{Why do we need a new algorithm?}\label{sec:previous-alg}

As discussed earlier, online resource allocation and its special cases have been extensively studied in the literature. Perhaps one of the algorithms from the literature continues to perform well under horizon uncertainty? We show below that previously-studied algorithms can be exponentially worse than our algorithm. Consider an uncertainty window $[\tau_1, \tau_2]$, where $\tau_1, \tau_2 \in \mathbb{Z}_+$. Consider an online algorithm $A$ which takes as input the horizon and is optimal (defined precisely later) for stochastic online resource allocation when the horizon is known. Suppose we pick some horizon $T^* \in [\tau_1, \tau_2]$ before the first request arrives and run algorithm $A$ with $T^*$ in the hope of getting good performance for all horizons $T \in [\tau_1, \tau_2]$. As we show next, this approach performs much worse than our algorithm even when there is only one resource ($m=1$), the same request arrives at all time steps, and the decision-maker knows this to be the case.

Let  $B$ be the initial resource endowment, $\X_t = [0,\max\{1,B/\tau_1\}]$ be the action set for all $t \in [\tau_2]$, and $\PP_r$ be the deterministic distribution that always serves the request $(f, b)$ where $f(x) = x^r$ for a fixed $r \in (0,1)$ and $b(x) = x$. Observe that all the requests are the same, the decision-maker knows this fact, and she takes her first action after observing the request. In particular, the decision maker completely knows the deterministic request after the first request arrives and before she takes her first action. Moreover, if she employs an algorithm for the known-horizon setting with horizon $T^*$ as the input and we have $T = T^*$, then the algorithm has as much information (about the request and the horizon) available before making its first decision as it would in hindsight. This motivates us to call an algorithm \textbf{optimal for the known-horizon setting} if it takes the same actions as the hindsight optimal $\opt(T^*, \vec\gamma)$ on this instance when $T = T^*$ and it is given horizon $T^*$ as the input. The dual-descent algorithm of \citet{balseiro2020best} (when appropriately initialized), and all of the primal methods based on solving the fluid approximation (e.g., \citealt{jasin2012re, agrawal2014dynamic} and \citealt{balseiro2021survey}) satisfy this definition of optimality. Let $A$ be such an optimal algorithm.

Consequently, if $\{x_t\}_t$ are the actions of algorithm $A$, then $\{x_t\}_t$ is an optimal solution to the hindsight-optimization problem $\opt(T^*, \vec\gamma)$ (see equation~\eqref{eq:opt}). In Lemma~\ref{lemma:r-opt}, we will show that the concavity of $f$ implies that $x^*_t = B/T^*$ for all $t \leq T^*$ and $x^*_t = 0$ for $t > T^*$ is the unique hindsight optimal solution of $\opt(T^*, \vec\gamma)$, which implies that $x_t = x_t^*$ for all $t \geq 1$. Now, recall that algorithm $A$ does not know the horizon $T$ and is non-anticipating. Consequently, it will take the actions $\{x_t\}_t$ no matter what $T$ turns out to be. This is because, if $T \leq T^*$, then it does not know that $T$ is different from $T^*$, and if $T > T^*$, then it has run out of budget by time step $T^*$.

The performance ratio of algorithm $A$ for $T = \tau_1$ is given by
\begin{align*}
  c(A|\tau_1, \PP_r) = \frac{\E_{\vec\gamma \sim \PP_r^{\tau_1}}\left[R(A|\tau_1, \vec \gamma) \right]}{\E_{\vec\gamma \sim \PP_r^{\tau_1}}\left[\opt(\tau_1, \vec\gamma)\right]} = \frac{(B/T^*)^r \cdot \tau_1}{(B/\tau_1)^r \cdot \tau_1} = \left( \frac{\tau_1}{T^*} \right)^r
\end{align*}
and for $T = \tau_2$ is given by
\begin{align*}
  c(A|\tau_2, \PP_r) = \frac{\E_{\vec\gamma \sim \PP_r^{\tau_2}}\left[R(A|\tau_2, \vec \gamma) \right]}{\E_{\vec\gamma \sim \PP_r^{\tau_2}}\left[\opt(\tau_2, \vec\gamma)\right]} = \frac{(B/T^*)^r \cdot T^*}{(B/\tau_2)^r \cdot \tau_2} = \left( \frac{T^*}{\tau_2} \right)^{1-r} \,.
\end{align*}
Finally, observe that:
\begin{itemize}
  \item If $T^* > \sqrt{\tau_1\tau_2}$, then $\inf_{r \in (0,1)} c(A|\tau_1, \PP_r) = \lim_{r \to 1} c(A|\tau_1,\PP_r) = \tau_1/T^* \leq \left(\sqrt{\tau_2/\tau_1} \right)^{-1}$.
  \item If $T^* \leq \sqrt{\tau_1\tau_2}$, then $\inf_{r \in (0,1)} c(A|\tau_2, \PP_r) = \lim_{r \to 0} c(A|\tau_2,\PP_r) = T^*/\tau_2 \leq \left( \sqrt{\tau_2/\tau_1} \right)^{-1}$.
\end{itemize}

Therefore, we get that the competitive ratio of algorithm $A$ is bounded above by $(\sqrt{\tau_2/\tau_1})^{-1}$ for all values of $\{\tau_1, \tau_2\}$. In stark contrast, if $\tau_1$ is large and $B = \Theta(\tau_1)$, we will show that our online algorithm achieves a competitive ratio greater than $(1 + \ln(\tau_2/\tau_1))^{-1}$, which is exponentially better than algorithm $A$. Even for small values of $\tau_2/\tau_1$, our algorithm significantly outperforms previous algorithms (see Figure~\ref{fig:us-vs-previous}).

We have shown that a proxy horizon does not allow us to use algorithms which are optimal for the known-horizon setting to obtain good performance in the face of horizon uncertainty. Perhaps one can use the Doubling Trick instead? The Doubling Trick involves running an algorithm designed for the known-horizon setting repeatedly on time-intervals of increasing lengths. More precisely, given an optimal (or low-regret) algorithm $A$ for the known horizon setting, run it separately on the intervals $[1, T^*], [T^*, 2T^*], \dots, [2^k T^*, 2^{k+1} T^*]$ for some $T^* \geq 1$. Unfortunately, as we alluded to in the Introduction, the Doubling Trick does not work for online resource allocation. This is because, unlike online convex optimization \citep{hazan2016introduction,shalev2012online} where the problem decouples and the regret from the different intervals is simply added together to get total regret, the online resource allocation problem has global resource constraints and does not decouple.

In particular, if we were to run an algorithm $A$ with low-regret in the known-horizon setting on the interval $[1,T^*]$, it will attempt to deplete all of the available resources by time $T^*$ (because unused resources have no value to $A$ after $T^*$), which in turn implies that we will not have sufficient resource capacity to even run algorithm $A$ on latter intervals $[T^*, 2T^*] , \dots, [2^k T^*, 2^{k+1} T^*]$. The crux of the problem is that the Doubling Trick does not take the resources capacities into account: since we only have a finite amount of resources, one cannot repeatedly run algorithm $A$ because it will consume the entire resource capacity on every run (if possible). Additionally, note that the benchmark in online resource allocation is the optimal solution in hindsight considering all requests till time $T$, which is very different from the sum of the benchmark optimal solutions in the intervals $[1, T^*], [T^*, 2T^*], \dots, [2^k T^*, 2^{k+1} T^*]$. One can potentially come up with sophisticated versions of the Doubling Trick that allocate the resource endowment between the intervals in non-trivial ways. But the aforementioned lack of decomposability of the benchmark across intervals means that analyzing such heuristics would be far from straightforward. In fact, one of our primary contributions is a general performance guarantee for dual mirror descent with arbitrary allocation of the resource endowment across time steps (Theorem~\ref{thm:regret}), which allows one to analyze such heuristics. Finally, in online convex optimization, the Doubling Trick allows one to convert an algorithm for the known-horizon setting into one for the unknown-horizon setting while maintaining the same asymptotic competitive ratio of 1. However, as we will show in Theorem~\ref{thm:upper-bound}, it is impossible to achieve the same competitive ratio in the known-horizon and unknown-horizon settings for the online resource allocation problem. Thus, the simple Doubling Trick cannot be applied to online resource allocation, necessitating the need for novel techniques beyond the ones developed for online convex optimization.

\section{The Algorithm}\label{sec:lower-bound}

In this section, we describe our dual-mirror-descent-based master algorithm. 
As the name suggests, this algorithm maintains and updates dual variables, using them to compute the action $x_t$ at time $t$.
Moreover, the algorithm is parameterized by a \emph{target consumption sequence}. 


\begin{definition}
	We call a sequence $\vec\lambda = (\lambda_1, \lambda_2, \dots, \lambda_{t}, \dots, \lambda_{\tau_2} )$ a target consumption sequence if $\lambda_t \in \R_+^m$ for all $t \leq \tau_2$, $\lambda_1 > 0$ and $\sum_{t=1}^{\tau_2} \lambda_t \leq B$.
\end{definition}
Here, $\lambda_t \in \R_+^m$ denotes the target amount of resources that one wants to consume in the $t$-th time period. $\sum_{t=1}^{\tau_2} \lambda_t \leq B$ ensures that the budget never runs out if one is able to hit these target consumptions. Given a target consumption sequence $\vec\lambda$, we use $\bar{\lambda} = \max_{t, j} \lambda_{t,j}$ to denote the largest target consumption of any resource at any time step.  

We will be showing performance guarantees for our algorithms in terms of the target consumption sequence, and then provide algorithms for computing the optimal target sequence in later sections. 

\subsection{The Dual Problem}

The Lagrangian dual problem to the hindsight optimization problem~\eqref{eq:opt} is obtained by moving the resource constraints to the objective using multipliers $\mu = (\mu_1, \dots, \mu_m) \in \R_+^m$. Intuitively, the dual variable $\mu_j$ acts as the price of resource $j$ and accounts for the opportunity cost of consuming resource $j$. This allows us to define the objective function of the dual optimization problem:
\begin{align*}
	D(\mu | T, \vec\gamma) &\coloneqq \sup_{x \in \prod_t \X_t}\ \left\{ \sum_{t=1}^T f_t(x_t)  + \mu^\top \left( B - \sum_{t=1}^T b_t(x_t) \right) \right\}\\	 &=  \sum_{t=1}^T \sup_{x_t \in \X_t}\left\{ f_t(x_t) + \mu^\top(\rho_T - b_t(x_t)) \right\}\\
	 &= \sum_{t=1}^T \left( f^*_t(\mu) + \mu^\top \rho_T \right)\,,
\end{align*}
where the second equation follows because the objective is separable and $\rho_T = B/T$, and the last by defining the opportunity-cost-adjusted reward to be $f^*_t(\mu) \coloneqq \sup_{x \in \X_t} \left\{ f_t(x) - \mu^\top b_t(x) \right\}$. The dual problem is simply $\min_{\mu \in \R_+} D(\mu|T, \vec\gamma)$. Importantly, we get weak duality: $\opt(T, \vec\gamma) \leq D(\mu|T, \vec\gamma)$ for all dual solutions $\mu$ (we provide a proof in Proposition~\ref{prop:weak-duality} of Appendix~\ref{appendix:lower-bound}).

Recall that, in our definition of competitive ratio \eqref{eq:comp-ratio}, we are interested in the expectation of $\opt(T,\vec\gamma)$ when $\vec\gamma \sim \PP^T$. Weak duality allows us to bound this quantity from above as
\begin{align}\label{eq:dual-decomp}
	\E_{\vec\gamma \sim \PP^T} [\opt(T, \vec\gamma)] \leq\E_{\vec\gamma \sim \PP^T} [D(\mu|T, \vec\gamma)] =  
	\sum_{t=1}^T  \left( \E_{\gamma \sim \PP}\left[f^*_t(\mu) \right] + \rho_T^\top \mu\, \right).
\end{align}
This motivates us to define the following single-period dual function with  target consumption $\lambda \in \R_+^m$  as $\D(\mu| \lambda, \PP) \coloneqq	 \E_{\gamma \sim \PP}\left[f^*(\mu) \right] + \lambda^\top \mu$.
The following lemma notes some important properties of the single-period dual objective.

\begin{lemma}\label{lemma:dual-prop}
	$\D(\mu| \lambda, \PP)$ is convex in $\mu \in \R_+^m$ for every $\lambda \in \R_+^m$. Moreover, for every $\mu \in \R_+$ and $T \geq 1$, the following properties hold:
	\begin{itemize}
		\item[(a)] Separability: $\E_{\vec\gamma \sim \PP^T} [D(\mu| T, \vec\gamma)] = T \cdot \D(\mu | \rho_T, \PP)$
		\item[(b)] Sub-homogeneity: For $a \in [0,1]$, $a \cdot\D(\mu|\lambda, \PP) \leq \D(\mu | a \cdot \lambda, \PP)$.
		\item[(c)] Monotonicity: If $\lambda \leq \kappa$, then $\D(\mu|\lambda, \PP) \leq \D(\mu | \kappa, \PP)$.
	\end{itemize}
\end{lemma}

\subsection{Variable Target Dual Mirror Descent}

\begin{algorithm}[t!]
	\SetAlgoLined
	{\bf Input:} Initial dual solution $\mu_1$, initial resource endowment $B_1 = B$, target consumption sequence $\vec\lambda$, reference function $h: \R^m \rightarrow \R$, and step-size $\eta$. \\
	\For{$t=1,\ldots,T$}{
		Receive request $(f_t, b_t, \X_t) \sim \PP$.\\
		Make the primal decision $x_t$ and update the remaining resources $B_t$:
		\begin{align}
		&\tilde{x}_{t} \in \argmax_{x\in\X_t}\left\{f_{t}(x)-\mu_{t}^{\top} b_{t}(x)\right\} \ ,\label{eq:primal_decision} \\
		&x_{t}=\begin{cases}
			\tilde{x}_t    & \text{  if } b_t(\tilde{x}_t)\le B_t  \\
		0
		& \text{ otherwise}
		\end{cases}\ , \nonumber \\
		&B_{t+1} = B_t - b_t(x_t) . \nonumber
		\end{align}\\

		Obtain a sub-gradient of the dual function: $g_t = \lambda_t -b_t(\tilde x_t)$.

		Update the dual variable by mirror descent:
           $\mu_{t+1} = \arg\min_{\mu\in\R_+^m} g_t^\top \mu + \frac{1}{\eta} V_h(\mu, \mu_t)$,
       where $V_h(x,y)=h(x)-h(y)-\nabla h(y)^\top (x-y)$ is the Bregman divergence.
	}
	\caption{Variable Target Dual Mirror Descent Algorithm}
	\label{alg:dual-descent}
\end{algorithm}

Algorithm~\ref{alg:dual-descent} is a highly-flexible stochastic dual descent algorithm that allows the decision maker to specify the target consumption sequence $\vec\lambda$, in addition to the initial dual variable $\mu_1$, the reference function $h(\cdot)$ and the step-size $\eta$ needed to specify the mirror-descent procedure. This flexibility allows us to seamlessly analyze a variety of different algorithms using the same framework. As is standard in the literature on mirror descent \citep{shalev2012online, hazan2016introduction}, we require the reference function $h(\cdot)$ to be either differentiable or essentially smooth \citep{bauschke2001essential}, and be $\sigma$-strongly convex in the $\|\cdot \|_1$ norm. Moreover, Algorithm~\ref{alg:dual-descent} is efficient when an optimal solution for the per-period optimization problem in equation \eqref{eq:primal_decision} can be computed efficiently. This is possible for many applications, see \cite{balseiro2020best} for details.

The algorithm maintains a dual variable $\mu_t$ at each time step, which acts as the price of the resources and accounts for the opportunity cost of spending them at time $t$. Then, for a request $\gamma_t = (f_t, b_t, \X_t)$ at time $t$, it chooses the action $x_t$ that maximizes the opportunity-cost-adjusted reward $x_t \in \argmax_{x \in \X_t} \left\{ f_t(x) - \mu_t^\top b_t(x) \right\}$. As our goal here is to build intuition, we ignore the minor difference between $\tilde{x}_t$ and $x_t$ which ensures that we never overspend resources. The dual variable is updated using mirror descent with reference function $h(\cdot)$, step-size $\eta$, and using $g_t = \lambda_t - b_t(x_t)$ as a subgradient. Intuitively, mirror descent seeks to make the subgradients as small as possible, which in our settings translates to making the expected resource consumption in period $t$ as close as possible to the target consumption $\lambda_t$. As a result, the target consumption sequence can be interpreted as the ideal expected consumption per period, and the algorithm seeks to track these rates of consumption.

We conclude by discussing some common choices for the reference function. If the reference function is the squared-Euclidean norm, i.e., $h(\mu) = \|\mu\|_2^2$, then the update rule is $\mu_{t+1} = \max\left\{ \mu_t - \eta g_t, 0 \right \}$ and the algorithm implements subgradient descent. If the reference function is the negative entropy, i.e., $h(\mu) = \sum_{j=1}^m \mu_j \log(\mu_j)$, then the update rule is $\mu_{t+1,j} = \mu_{t,j} \exp\left( - \eta g_{t,j} \right)$ and the algorithm implements multiplicative weights.

\subsection{Performance Guarantees}

In this section, we provide worst-case performance guarantees of our algorithm for arbitrary target consumption sequences. Before stating our result, we provide further intuition about our algorithm. Consider the single-period dual function with target consumption $\lambda$, given by
\begin{align*}
	\D(\mu| \lambda, \PP) = \E_{\gamma \sim \PP}\left[f^*(\mu) \right] + \lambda^\top \mu = \E_{\gamma \sim \PP}\left[\sup_{s \in \X}\ \left\{f(x) - \mu^\top b(x)\right\} \right] + \lambda^\top \mu \,.
\end{align*}
Then, by Danskin's Theorem, its subgradient is given by $\E_{\gamma \sim \PP}[\lambda - b(x_\gamma(\mu))] \in \partial_\mu \D(\mu| \lambda, \PP)$ where $x_{\gamma}(\mu) \in \argmax_{x \in \X} \left\{ f(x) - \mu^\top b(x)\right\}$ is an optimal decision when the request is $\gamma = (f,b,\X)$ and the dual variable is $\mu$. Therefore, $g_t = \lambda_t - b_t(x_t)$ is a random unbiased sample of the subgradient of $\D(\mu| \lambda, \PP)$. Now, if we wanted to minimize the dual objective $\E[D(\mu|T, \vec\gamma)] = \sum_{t=1}^T \D(\mu | \rho_T, \PP)$ for some known $T$, we can run mirror descent on the function $\D(\mu | \rho_T, \PP)$ by setting $\lambda_t = \rho_T$ for all $t \leq T$. This is exactly the approach taken by \cite{balseiro2020best}. Unfortunately, this method does not extend to our setting because the horizon $T$ is unknown. 

Observe that, since mirror descent guarantees vanishing regret even against adversarially generated losses, it continues to give vanishing regret in the dual space even when the single-period duals $\D(\mu| \lambda_t, \PP)$ vary with time due to the changing target consumptions. However, when $\lambda_t \neq \rho_T$ for some $t \leq T$, it is no longer the case that $\sum_{t=1}^T \D(\mu | \lambda_t, \PP)$ provides an upper bound on the hindsight optimization problem. The crux of the following result involves overcoming this difficulty and comparing the time-varying single-period duals with the hindsight optimal solution for all $T$ simultaneously, leading to a performance guarantee for Algorithm~\ref{alg:dual-descent}.

\begin{theorem} \label{thm:regret} Consider Algorithm \ref{alg:dual-descent} with initial dual solution $\mu_1$, initial resource endowment $B_1 = B$, a target consumption sequence $\vec\lambda$, reference function $h(\cdot): \R^m \rightarrow \R$, and step-size $\eta$. For any $T \geq 1$, if we set
	\begin{align}\label{eq:comp-ratio}
		c(\vec\lambda, T) \coloneqq  \frac{1}{T}\sum_{t=1}^T \min \left\{ \min_{1 \leq j \leq m} \frac{\lambda_{t,j}}{\rho_{T,j}}, 1 \right\} \,,
	\end{align}
	then it holds that
		\begin{align}\label{eq:master}
			\E_{\vec\gamma \sim \PP^T} \left[ c(\vec\lambda, T) \cdot \opt(T, \vec\gamma) - R(A|T, \vec\gamma)\right] \le C^{(T)}_1 + C_2 \eta T + \frac {C^{(T)}_3}{\eta} \,.
		\end{align}
		where $C^{(T)}_1 = \bar f \bar b  / \underline \rho_T$, $C_2 = (\bar b + \bar \lambda)^2/2\sigma$, $C^{(T)}_3= \max \left\{  V_h(\mu, \mu_1) : \mu \in \{0, (\bar f/\underline \rho_T) e_1,\ldots, (\bar f/\underline \rho_T) e_m\}  \right\}$. Here $e_j\in\R^m$ is the $j$-th unit vector and $\underline \rho_T = \min_j \rho_{T,j}$.
\end{theorem}

The proof proceeds in multiple steps. First, we write the rewards collected by Algorithm~\ref{alg:dual-descent} as a sum of per-period duals and complementary-slackness terms. Next, we use weak duality to upper bound the expected value of the hindsight optimal reward $\E[\opt(T,\vec\gamma)]$ in terms of the expected hindsight dual. These two steps are common to all primal-dual analyses, but past techniques offer no guidance beyond this point. The core difficulty is that the expected hindsight dual is equal to the sum of per-period duals with target consumption $\rho_T$, whereas the lower bound on the performance of our algorithm is in terms of per-period duals with target consumptions $\lambda_t$. Importantly, this difficulty does not arise in past works because the target consumptions $\lambda_t = \rho_T$, which makes the two terms directly comparable. Our main technical insight lies in using Lemma~\ref{lemma:dual-prop} to manipulate the per-period duals and then carefully choosing the right dual solution in order to compare the two terms. Moreover, one also needs to take into account the fact that the resources may run out before the horizon $T$ is reached, and Algorithm~\ref{alg:dual-descent} does not accumulate rewards after this point. Since the point at which the budget runs out depends on the target consumption sequence, we also establish a bound on the loss from depleting the resources too early which applies to variable target sequences. We believe that Algorithm~\ref{alg:dual-descent} and our proof techniques distill the core tradeoffs of the problem and can be used more broadly. The full proof is in Appendix~\ref{appendix:lower-bound}.

Theorem~\ref{thm:regret} is the bedrock of our positive results. It allows us to drastically simplify the design of algorithms: instead of searching for the optimal algorithm, we can focus on the much simpler problem of selecting the optimal target consumption sequence. The following result provides a key step in this direction by showing that $c(\vec\lambda, T)$ is the asymptotic performance ratio of Algorithm~\ref{alg:dual-descent} with target sequence $\vec\lambda$.

\begin{proposition}\label{prop:master}
	Let $A$ be Algorithm~\ref{alg:dual-descent} with initial dual solution $\mu_1$, initial resource endowment $B_1 = B$, a target consumption sequence $\vec\lambda$, reference function $h(\cdot): \R^m \rightarrow \R$, and step-size $\eta$. Set $C'_1 = \max_{T \in [\tau_1, \tau_2]}C^{(T )}_1$ and $C'_3 = \max_{T \in [\tau_1, \tau_2]}C^{(T )}_3$. 
	Then, with step size $\eta = \sqrt{C'_3/\{C_2 \tau_2\}}$, the following statement holds for all $T \in [\tau_1, \tau_2]$:
	\begin{align*}
		\inf_\PP c(A|T, \PP) = \inf_\PP \frac{\E_{\vec\gamma \sim \PP^T}\left[R(A|T, \vec \gamma) \right]}{\E_{\vec\gamma \sim \PP^T}\left[\opt(T, \vec\gamma)\right]} \geq c(\vec\lambda, T) - \epsilon \,,
	\end{align*}
	where
	\begin{align*}
		\epsilon = \frac{C'_1}{\kappa \tau_1} + 2 \cdot \frac{\sqrt{(\tau_2/\tau_1)C_2 C'_3 }}{\kappa \sqrt{\tau_1}} \,.
	\end{align*}
\end{proposition}

\begin{remark}
	To convert the guarantee in Proposition~\ref{prop:master} to an asymptotic guarantee, one needs to consider the regime where the initial resources scale with the horizon as $B = \Omega(\tau_2)$, which ensures that $\rho_{\tau_2} = \Omega(1)$ and the constants $C_1'$ and $C_2'$ remain bounded. In which case, if we let $\tau_1$ grow large while ensuring $\tau_2 = O(\tau_1)$, we can make $\epsilon$ arbitrarily small. In particular, $\epsilon = O(\tau_1^{-1/2})$. The assumption that the initial resources scales linearly with the number of requests is pervasive in the literature and well-motivated in applications such as internet advertising~\citep{mehta2013online} and revenue management~\citep{talluri2004theory}. Moreover, an error of $\epsilon = \Omega(\tau_1^{-1/2})$ is unavoidable even for the case when the horizon is known, i.e., $\tau_1 = \tau_2$ (see \citealt{arlotto2019uniformly}).
\end{remark}

\begin{remark}
	In applications where it might be difficult to estimate the constants $C_2$ and $C_3'$, one can use the step size $\eta = 1/\sqrt{\tau_2}$ to get
	\begin{align*}
		\epsilon = \frac{C'_1}{\kappa \tau_1} + \frac{\sqrt{(\tau_2/\tau_1)} \cdot (C_2 + C'_3)}{\kappa \sqrt{\tau_1}}\,,
	\end{align*}
	which yields similar asymptotic rates.
\end{remark}

Having characterized the performance of Algorithm~\ref{alg:dual-descent} in terms of the target sequence, we next optimize it for the models of horizon uncertainty discussed in Section~\ref{sec:model}. Although we will only discuss two models of uncertainty, we would like to note that Theorem~\ref{thm:regret} and Proposition~\ref{prop:master} are very general tools that can be applied more broadly. In particular, observe that
\begin{align*}
	c(\vec\lambda, T) =  \frac{1}{T}\sum_{t=1}^T \min \left\{ \min_{1 \leq j \leq m} \frac{\lambda_{t,j}}{\rho_{T,j}}, 1 \right\}
\end{align*}
is a concave function of $\vec\lambda$ for all $T \geq 1$. This is because each term in the sum is a minimum of a finite collection of linear functions of $\vec\lambda$. Consequently, any performance measure of Algorithm~\ref{alg:dual-descent} which is a concave non-decreasing function of the performance ratios $\{c(A|T,\PP)\}_T$ is a concave function of the target sequence $\vec\lambda$ in light of Proposition~\ref{prop:master}. We pause to emphasize this important transition we have made in this section: we reduced the extremely complex problem of designing an algorithm for online resource allocation under horizon uncertainty to a concave optimization problem with the power to handle a variety of constraints and objectives. In the next section, we show that this reduction is without much loss in the uncertainty-window setting---picking the optimal target consumption sequence leads to a near-optimal competitive ratio in the uncertainty-window model.

\section{Uncertainty Window}\label{sec:uncertainty-window}

Motivated by robust optimization, in this section,  we take the uncertainty-set approach to modeling horizon uncertainty. In particular, we assume that the decision maker is not aware of the exact value of $T$ but knows it can lie anywhere in the known uncertainty window $[\tau_1, \tau_2]$. Recall that we measure the performance of an algorithm $A$ in this model by its competitive ratio $c(A)$, which is defined as
\begin{align*}
  c(A) \coloneqq \min_{\PP} \min_{T \in [\tau_1, \tau_2]} c(A|T, \PP) = \inf_{T \in [\tau_1, \tau_2]} \frac{\E_{\vec\gamma \sim \PP^T}\left[R(A|T, \vec \gamma) \right]}{\E_{\vec\gamma \sim \PP^T}\left[\opt(T, \vec\gamma)\right]} \,.
\end{align*}

\subsection{Upper Bound on Competitive Ratio}\label{sec:upper-bound}

We begin by showing that no online algorithm can attain a competitive ratio of $c(A) = 1$ whenever $\tau_2/ \tau_1 > 1$ and, moreover, when $\tau_2/\tau_1$ is large the competitive ratio degrades at a rate of $e\cdot \ln\ln(\tau_2/\tau_1)/\ln(\tau_2/\tau_1)$. In other words, the competitive ratio of every algorithm degrades to 0 as $\tau_2/\tau_1$ grows large. In fact, we prove that this upper bound on the best-possible competitive ratio holds even when (i) there is only 1 resource, (ii) the decision maker receives the same request at each time step, (iii) this request is known to the decision maker ahead of time, and (iv) the request has a smooth concave reward function and linear resource consumption.

For the purposes of this subsection, set the number of resources to $m = 1$. Consider an arbitrary initial resource endowment $B \in \R_{++}^m$.
For every $r \in (0,1)$, define the singleton request space $\S_r = \{(f_r, I, \X)\}$, where $\X = [0, \max\{1, B/\tau_1\}]$, and $f_r(x) = x^r$, $I(x) = x$ for all $x \in \X$. Note that $f_r$ is concave for all $r \in (0,1)$. Let $\PP_r$ be the canonical distribution on $\S_r$ that serves the request $(f_r, I, \X)$ with probability one. Since all requests are identical, we abuse notation and use $\opt(T, r)$ (similarly $R(A|T, r)$) to denote the hindsight-optimal reward $\opt(T, \vec\gamma)$ (total reward $R(A|T, \vec\gamma)$ of algorithm $A$) when $\vec \gamma \sim \PP_r^T$, i.e., $\gamma_t = (f_r, I, \X)$ for all $t \leq T$.

Before stating the upper bound, we would like to note that randomization only makes the performance of any online algorithm worse. To see this consider any non-deterministic online algorithm $A$ and let $x(A)_t$ denote the random variable which captures the action taken by $A$ at time $t$. Define $A'$ to be the online algorithm which takes the action $x(A')_t = \E[x(A)_t]$ at time $t$. Then, due to the strict concavity of $f_r$, we have $f_r(x(A')_t) > \E[f_r(x(A)_t)]$, and from the linearity of expectation, we have $\sum_{t=1}^{\tau_2} x(A')_t = \E[\sum_{t=1}^{\tau_2} x(A)_t] \leq B$. Therefore, the deterministic algorithm $A'$ attains strictly greater reward. Consequently, we will focus only on deterministic online algorithms for the remainder of this subsection. We are now ready to state the main result of this section.

\begin{theorem}\label{thm:upper-bound}
    For all $r \in (0,1)$ and $1 \leq \tau_1 \leq \tau_2$, every online algorithm $A$ satisfies
    \begin{align*}
        \min_{T \in [\tau_1, \tau_2]} \frac{R(A|T,r)}{\opt(T,r)} \leq \frac{1}{\left(1 + (1 - r)^{1/r} \cdot \ln(\tau_2/\tau_1) + \ln\left( \frac{\tau_1}{\tau_1 + 1} \right) \right)^r}\,.
    \end{align*}
  In particular, when $r = 1 - \{1/\ln\ln(\tau_2/\tau_1)\}$ and $\tau_2/\tau_1 > e^e$, every online algorithm $A$ satisfies
  \begin{align*}
    \min_{T \in [\tau_1, \tau_2]} \frac{R(A|T,r)}{\opt(T,r)} \leq  \frac{e\cdot \ln\ln(\tau_2/\tau_1)}{\ln(\tau_2/\tau_1)}\,.
  \end{align*}
  The above bounds hold even for online algorithms that have prior knowledge of $\PP_r$ before time $t=1$.
\end{theorem}

\begin{remark}
  Note that the upper bound on the competitive ratio established in Theorem~\ref{thm:upper-bound} degrades to zero as $\tau_2/\tau_1$ grows large. In particular, a positive competitive ratio cannot be obtained if no upper-bound on the horizon $T$ is known, thereby necessitating the need for a known uncertainty window.
\end{remark}


Figure~\ref{fig:us-vs-previous} plots the value of the upper bound on the competitive ratio as a function of $\tau_2/\tau_1$ for $\tau_2/\tau_1 \in [1,20]$.

We now discuss the main ideas behind the proof of Theorem~\ref{thm:upper-bound}. It suffices to prove the stronger statement in the theorem that holds for online algorithms with prior knowledge of $(r, \PP_r)$ before time $t=1$. Consequently, we assume that online algorithms have this prior knowledge in the remainder of this section. Any algorithm without this knowledge can only do worse. We begin by utilizing the concavity of $f_r$ to evaluate the optimal reward, which we note in the following lemma.

\begin{lemma}\label{lemma:r-opt}
  For $r \in (0,1)$ and $T \in [\tau_1, \tau_2]$, we have $\opt(T, r) = T \cdot (B/T)^r = B^r \cdot T^{1-r}$. Moreover, $x_t = B/T$ is the unique hindsight optimal solution.
\end{lemma}

Because the reward function $f_r$ is concave, it is optimal to spread resources uniformly over time and the optimal action with the benefit of hindsight is $x_t = B/T$. Next, we provide an alternative characterization of the competitive ratio that is more tractable.

\begin{lemma}\label{lemma:alter-char}
  For $r \in (0,1)$and $1 \leq \tau_1 \leq \tau_2$, we have
  \begin{align*}
    \sup_A\min_{T \in [\tau_1, \tau_2]} \frac{R(A|T,r)}{\opt(T,r)} =
    \max \left\{ c \in [0,1]\ \biggr\lvert\ \tau_1 \cdot f_r^{-1} \left(c \cdot \frac{\opt(\tau_1, r)}{\tau_1} \right) + \sum_{t= \tau_1 + 1}^{\tau_2} f_r^{-1} \left( c \cdot \Delta \opt(t, r) \right) \leq B  \right\}\,,
  \end{align*}
  where $\Delta\opt(t,r) = \opt(t,r) - \opt(t-1, r)$ and the $\sup$ is taken over all online algorithms.
\end{lemma}

We present a proof sketch of Lemma~\ref{lemma:alter-char} here. The main step in the proof involves showing that, for a given competitive ratio $c$, the minimum amount of resources that any online algorithm $A$ needs to be spend in order to satisfy $\min_{T \in [\tau_1, \tau_2]} R(A|T,r)/\opt(T,r) \geq c$ is given by
\begin{align*}
  \tau_1 \cdot f_r^{-1} \left(c \cdot \frac{\opt(\tau_1, r)}{\tau_1} \right) + \sum_{t= \tau_1 + 1}^{\tau_2} f_r^{-1} \left( c \cdot \Delta \opt(t, r) \right) \,.
\end{align*}
This is because $f_r$ is concave for all $r \in (0,1)$ and the resource consumption function $I$ is linear, which together imply that the marginal bang-per-buck $f'_r(x)$ (amount of reward per marginal unit of resource spent) decreases with $x$. As a consequence, an online algorithm that does not have any knowledge of $T$ (other than $T \in [\tau_1, \tau_2]$) and needs to satisfy $R(A|T,r) \geq c \cdot \opt(T,r)$ for all $T \in [\tau_1, \tau_2]$ would spend the minimum amount of resources in doing so if (i) it attains a reward of $c \cdot \opt(\tau_1, r)$ by evenly spending resources in the first $\tau_1$ steps, and (ii) it spends just enough resources to attain a reward of $c \cdot \Delta \opt(t,r)$ at each time step $t \geq \tau_1 +1$. Proving (ii) requires showing that $\Delta \opt(t,r)$ decreases with an increase in $t$, which follows from Lemma~\ref{lemma:r-opt}.
In particular, this ensures that obtaining all of $\Delta \opt(t,r)$ at time $t$ is cheaper than obtaining some of that reward at an earlier time $t' < t$. However, the proof requires a sophisticated water-filling argument to show that the aforementioned greedy strategy of minimizing the amount of resources at each time step leads to globally-minimal spending. Finally, combining Lemma~\ref{lemma:alter-char} and Lemma~\ref{lemma:r-opt} yields
\begin{align*}
  \tau_1 \cdot \left(c^* \cdot \frac{B^r \tau_1^{1-r}}{\tau_1} \right)^{1/r} + \sum_{t= \tau_1 + 1}^{\tau_2} \left( c^* \cdot [B^r t^{1-r} - B^r (t-1)^{1-r}] \right)^{1/r} \leq B
\end{align*}
for $c^* = \sup_A\min_{T \in [\tau_1, \tau_2]} R(A|T,r)/\opt(T,r)$. The above equation specifies an upper bound on $c^*$, which upon simplification leads to Theorem~\ref{thm:upper-bound}.

We conclude by noting that the upper bound of Theorem~\ref{thm:upper-bound} can be extended to the popular setting of online resource allocation with random linear rewards and consumptions (see Appendix~\ref{appendix:randomized-linear-upper-bound} for details). Moreover, the upper bound of Theorem~\ref{thm:upper-bound} holds even when the horizon $T$ is drawn from a distribution $\T$ supported on $[\tau_1, \tau_2]$ and this distribution is known to the decision maker. A proof based on strong duality can be found in Appendix~\ref{appendix:dist-upper-bound}.

\subsection{Optimizing the Target Sequence}

Having shown that no algorithm can attain a competitive ratio better than $\tilde{O}(1/\ln(\tau_2/\tau_1))$, we now show that Algorithm~\ref{alg:dual-descent} with an appropriately chosen target consumption sequence $\vec \lambda$ can achieve a competitive ratio of $\Omega(1/\ln(\tau_2/\tau_1))$ for sufficiently large $\tau_1$ and $B$. In light of Proposition~\ref{prop:master}, we can optimize the competitive ratio of Algorithm~\ref{alg:dual-descent} by finding the target consumption sequence which maximizes $\min_{T \in [\tau_1, \tau_2]} c(\vec\lambda, T)$, i.e., we need to solve the following maximin problem:
\begin{align*}
  \max_{\vec\lambda} \min_{T \in [\tau_1, \tau_2]} c(\vec\lambda, T) = \max_{\vec\lambda} \min_{T \in [\tau_1, \tau_2]} \frac{1}{T}\sum_{t=1}^T \min \left\{ \min_{1 \leq j \leq m} \frac{\lambda_{t,j}}{\rho_{T,j}}, 1 \right\} \,.
\end{align*}

The following proposition restates the above maximin problem as an LP.
\begin{proposition}\label{prop:uncertainty-window-LP}
  For budget $B$ and uncertainty window $[\tau_1, \tau_2]$, we have
  \begin{align*}
    \max_{\vec\lambda} \min_{T \in [\tau_1, \tau_2]} c(\vec\lambda, T) \quad  = \quad  \max_{z, y, \lambda} \quad &z\\
    \text{s.t.} \quad &z \leq \frac{1}{T} \sum_{t=1}^T y_{T,t} &\forall T \in [\tau_1, \tau_2]\\
    &y_{T,t} \leq \frac{\lambda_{t,j}}{\rho_{T,j}} &\forall j \in [m], T \in [\tau_1, \tau_2], t \in [T]\\
    &y_{T,t} \leq 1 &\forall T \in [\tau_1, \tau_2], t \in [T]\\
    &\sum_{t=1}^{\tau_2} \lambda_t \leq B\\
    &\lambda \geq 0
  \end{align*}
\end{proposition}

\begin{figure}[t!]
  \begin{subfigure}{0.33\textwidth}
    \includegraphics[width = \linewidth]{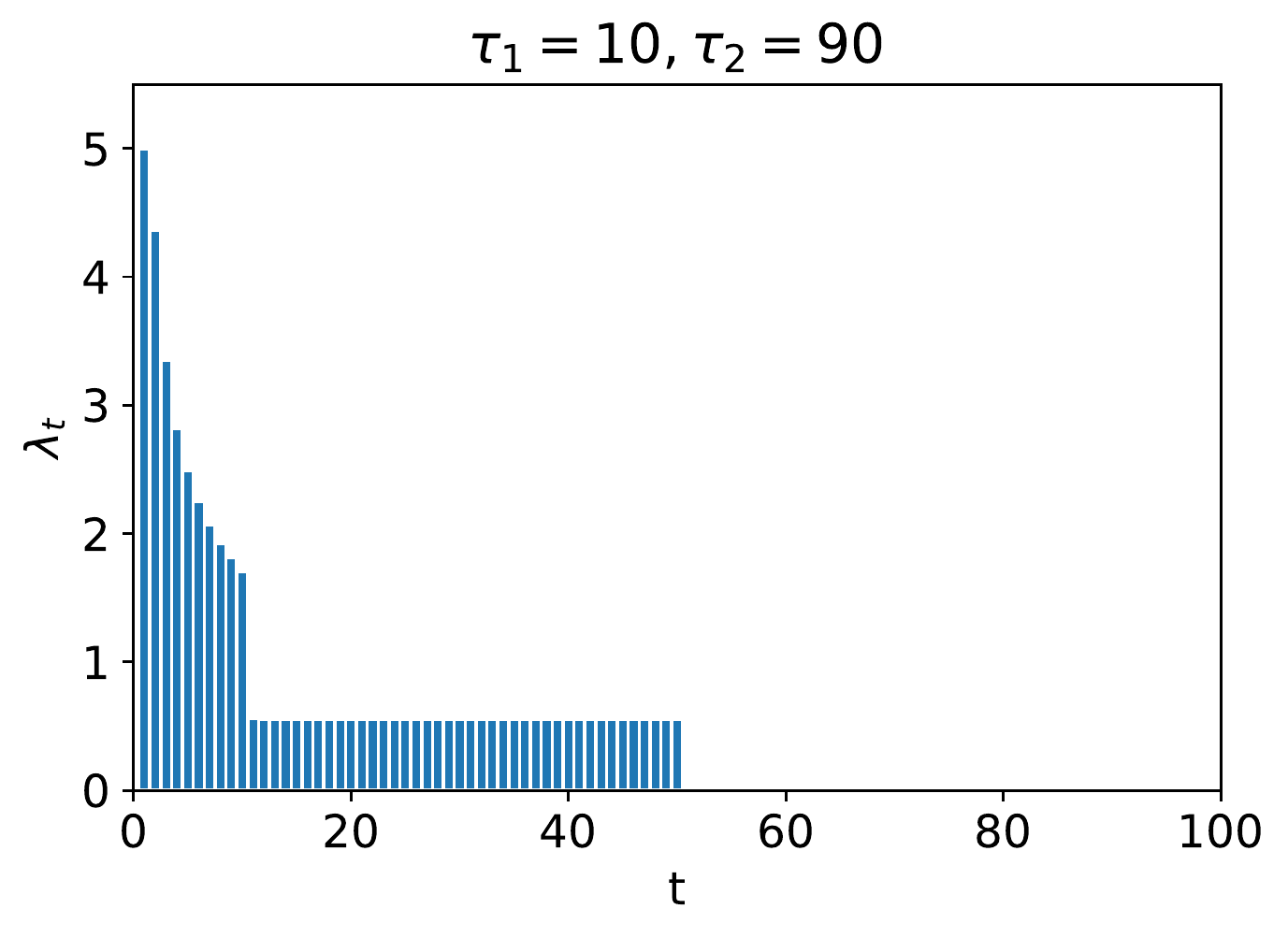}
  \end{subfigure}%
  \begin{subfigure}{0.33\textwidth}
    \includegraphics[width = \linewidth]{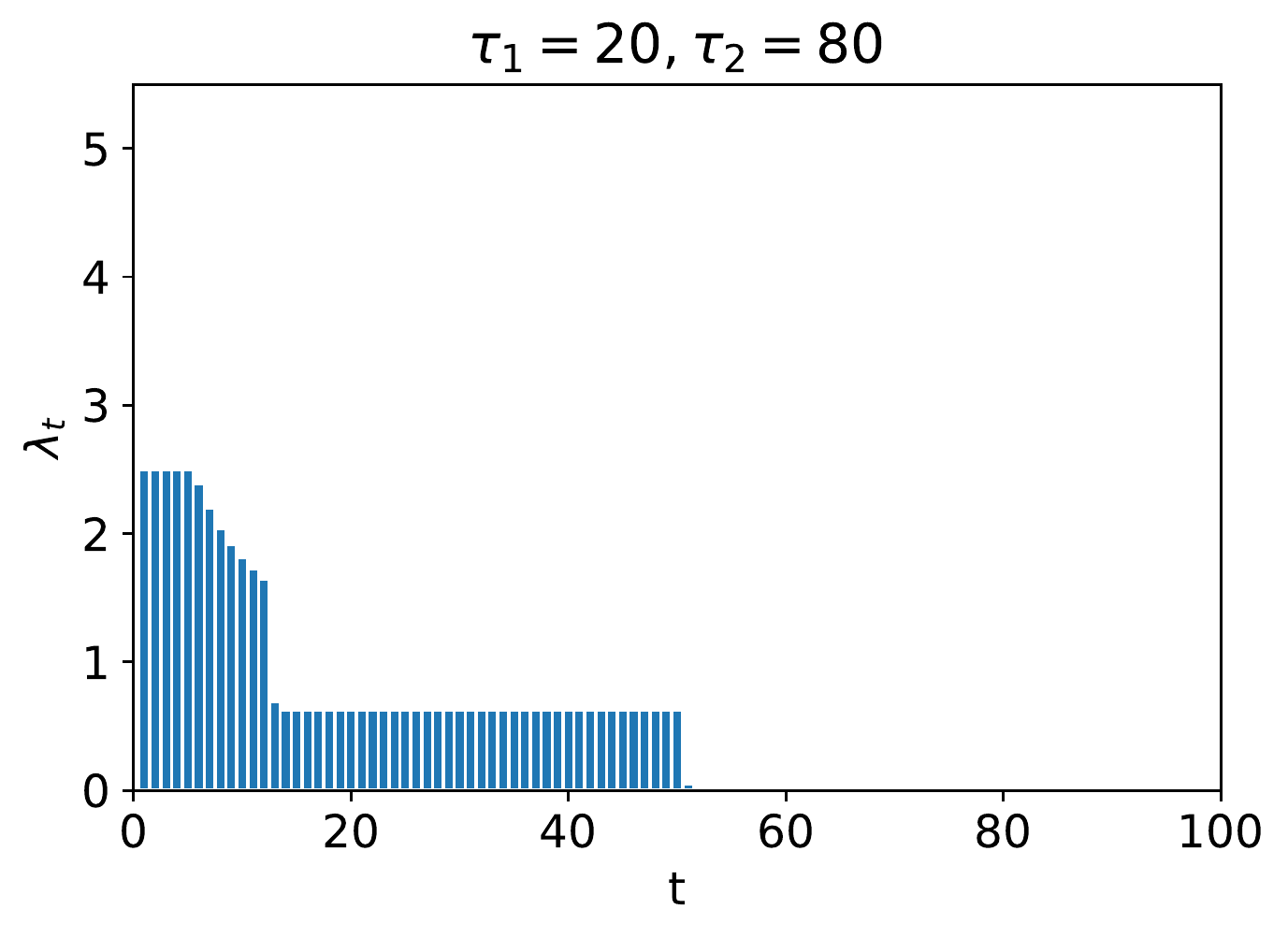}
  \end{subfigure}%
  \begin{subfigure}{0.33\textwidth}
    \includegraphics[width = \linewidth]{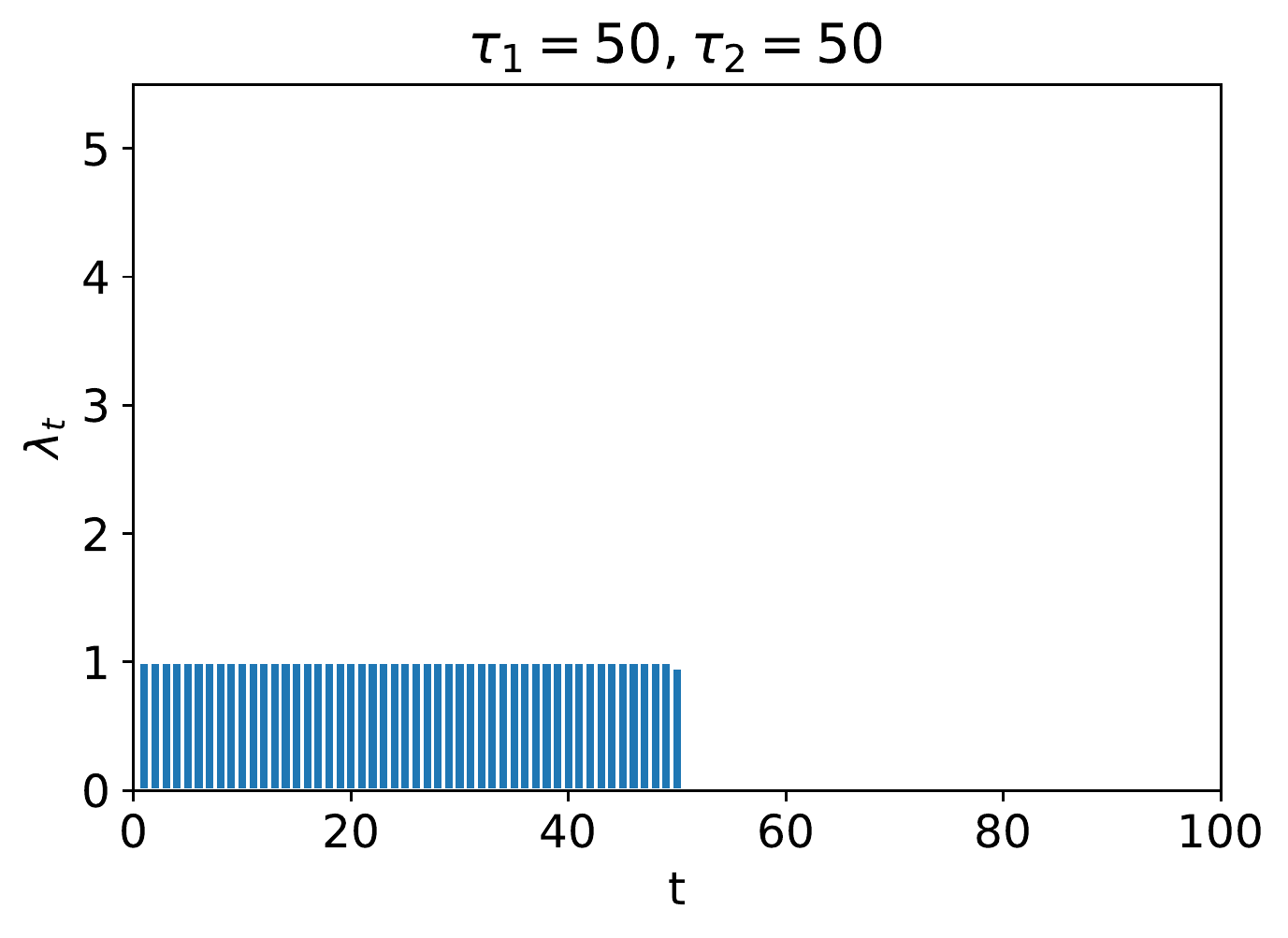}
  \end{subfigure}
  \vspace{-1em}
  \caption{The optimal target consumption sequence for various possible uncertainty windows centered on $T = 50$. Here, number of resources $m=1$ and initial resource endowment $B = 50$.}
  \label{fig:comp-target-seq}
\end{figure}

Proposition~\ref{prop:uncertainty-window-LP} states that we can efficiently compute the optimal target consumption sequence by solving an LP. Figure~\ref{fig:comp-target-seq} plots the optimal target sequences from Proposition~\ref{prop:uncertainty-window-LP} for different uncertainty windows. The optimal target consumption sequences are decreasing as the algorithm consumes resources more aggressively early on to prevent having too many leftover resources if the horizon ends being short. Moreover, as the uncertainty window becomes more narrow, the consumption sequence becomes less variable.

\begin{figure}[t!]
    \centering
    \includegraphics[width= 0.8\textwidth]{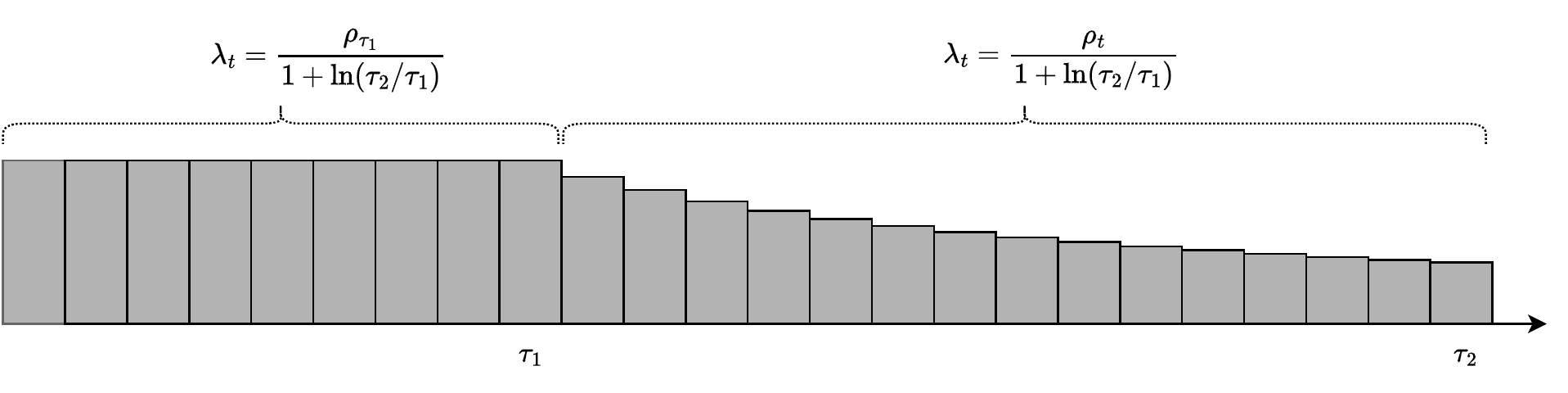}
    \caption{A simple target consumption sequence that achieves a competitive ratio of $1/ (1 + \ln(\tau_2/\tau_1))$. The height of the bars represents $\lambda_t$.}
    \label{fig:simple-log}
\end{figure}

To see that Algorithm~\ref{alg:dual-descent} with the optimal target consumption sequence from the above LP has an asymptotic competitive ratio of $\Omega(1/\ln(\tau_2/\tau_1))$, consider the following target consumption sequence (depicted in Figure~\ref{fig:simple-log}):
\begin{align}\label{eqn:simple-target-sequence}
  \lambda_t \coloneqq \begin{cases}
    \frac{1}{1 + \ln(\tau_2/\tau_1)} \cdot \frac{B}{\tau_1} = \frac{1}{1 + \ln(\tau_2/\tau_1)} \cdot \rho_{\tau_1}  &\text{if } t \leq \tau_1\,,\\
    \frac{1}{1 + \ln(\tau_2/\tau_1)} \cdot \frac{B}{t} = \frac{1}{1 + \ln(\tau_2/\tau_1)} \cdot \rho_t &\text{if } \tau_1+1 \leq t \leq \tau_2\,.
  \end{cases}
\end{align}

It is easy to see that it satisfies the budget constraint:
\begin{align*}
    \sum_{t=1}^{\tau_2} \lambda_t = \frac{B}{1 + \ln(\tau_2/\tau_1)} \cdot\left( \tau_1 \cdot \frac{1}{\tau_1} + \sum_{t=\tau_1+1}^{\tau_2} \frac{1}{t}\right) \leq \frac{B}{1 + \ln(\tau_2/\tau_1)} \cdot\left( 1 + \ln\left( \frac{\tau_2}{\tau_1} \right)\right) = B\,.
\end{align*}

Moreover, since $\rho_t \geq \rho_T$ for all $t \leq T$ and $T \in [\tau_1, \tau_2]$, we get
\begin{align*}
  \frac{1}{T}\sum_{t=1}^T \min \left\{ \min_{1 \leq j \leq m} \frac{\lambda_{t,j}}{\rho_{T,j}}, 1 \right\} \geq  \frac{1}{T}\sum_{t=1}^T \frac{1}{1 + \ln(\tau_2/\tau_1)} = \frac{1}{1 + \ln(\tau_2/\tau_1)} \,,
\end{align*}
where the inequality follows from the fact that $\rho_T \leq \rho_t$ for all $t \in [\tau_1, T]$ and the definition of $\lambda$ as given in \eqref{eqn:simple-target-sequence}.

Since $\vec\lambda$ from \eqref{eqn:simple-target-sequence} is just one possible choice of the target consumption sequence, we have
\begin{align*}
  \max_{\vec\lambda} \min_{T \in [\tau_1, \tau_2]} c(\vec\lambda, T) \geq \frac{1}{1 + \ln(\tau_2/\tau_1)} \,.
\end{align*}

\begin{figure}[t!]
    \centering
    \includegraphics[width= 0.5\textwidth]{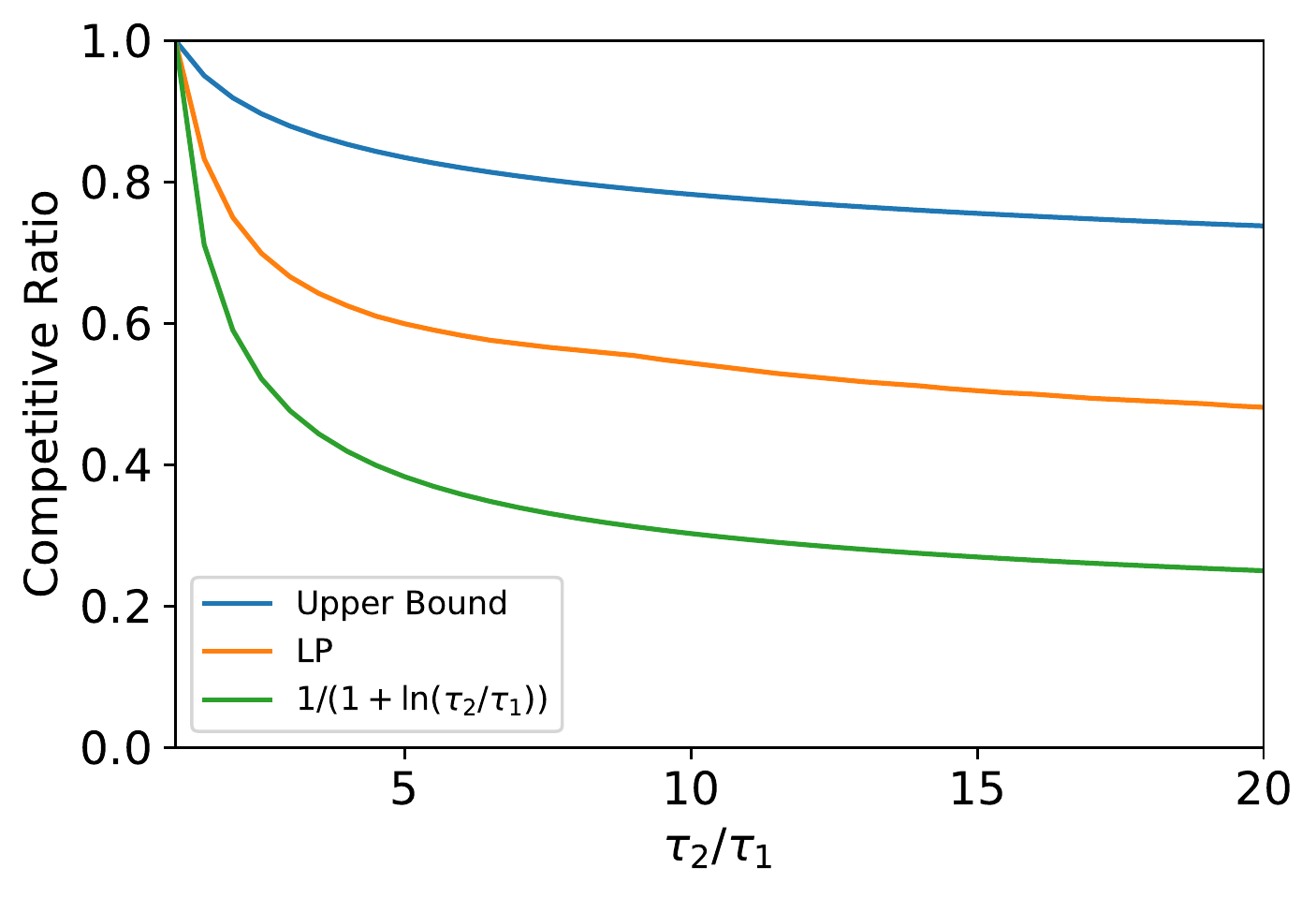}
    \caption{The competitive ratios achieved using the target consumption sequence from the LP in Proposition~\ref{prop:uncertainty-window-LP}, and the simple one defined in \eqref{eqn:simple-target-sequence} that yields a competitive ratio of $1/ (1 + \ln(\tau_2/\tau_1))$.}
    \label{fig:better-than-log}
\end{figure}

Therefore, we get that Algorithm~\ref{alg:dual-descent} in combination with the target consumption sequence returned by the LP in Proposition~\ref{prop:uncertainty-window-LP} achieves a degradation of $1/(1 + \ln(\tau_2/\tau_1))$ in the competitive ratio as a function of the multiplicative uncertainty $\tau_2/\tau_1$, which is optimal up to constants and a $\ln\ln(\tau_2/\tau_1)$ factor. In fact, as Figure~\ref{fig:better-than-log} shows, the target consumption sequence from the LP performs much better than $1/(1 + \ln{\tau_2/\tau_1})$, even for small values of $\tau_2/\tau_1$. In Section~\ref{sec:faster-alg}, we will give a faster algorithm which leverages the structure of the problem to optimize the target sequence and does not require solving an LP.

We conclude with a discussion on the structural similarity of the results of this subsection with those of \citet{besbes2014dynamic}, who studied the dynamic pricing problem (special case of online resource allocation) under demand shifts. They worked under the assumption that the request distribution is perfectly known, and showed that the optimal dynamic programming solution has a non-decreasing resource consumption sequence when the horizon is uncertain. The target consumption sequences described in this section are also non-increasing, leading to a similar structural insight for the more general online resource allocation problem with unknown request distribution.

\section{Incorporating Predictions about the Horizon}\label{sec:advice}

In the previous section, we did not assume that we had any information about the horizon $T$ other than the fact that it belonged to the uncertainty window $[\tau_1,\tau_2]$. This may be too pessimistic in settings where the environment is well behaved and machine learning algorithms can be deployed to make predictions about the horizon. In this section, we show that our Variable Target Dual Descent algorithm allows us to easily incorporate predictions by optimizing the target sequences. We formulate an LP to optimize the target sequence which allows the decision-maker to smoothly interpolate between the uncertainty-window setting and the known-horizon setting, thereby catering to different levels of confidence in the prediction.

First, we define the performance metrics we will use to measure the performance of an online algorithm capable of incorporating predictions. These metrics are pervasive in the Algorithms-with-Predictions literature (see \citealt{mitzenmacher2020algorithms} for an excellent survey) and are aimed at capturing the performance of the algorithm both when the prediction is accurate and in the worst case when the instance bears no resemblance to the prediction. To this end, in addition to the competitive-ratio metric defined in Section~\ref{sec:model}, which captures the worst-case performance, we introduce the notion of consistency to capture the performance of the algorithm when the prediction is accurate. Let $T_P \in [\tau_1,\tau_2]$ denote the predicted value of the horizon and let $A(T_P)$ denote algorithm $A$ when provided with the prediction $T_P$. We say that an algorithm $A$ is $\beta$-consistent on prediction $T_P$ and $\gamma$-competitive if it satisfies
\begin{align*}
  c(A(T_P)|T_P, \PP) = \frac{\E_{\vec\gamma \sim \PP^T}\left[R(A(T_P)|T_P, \vec \gamma) \right]}{\E_{\vec\gamma \sim \PP^T}\left[\opt(T_P, \vec\gamma)\right]} \geq \beta\,,
\end{align*}
and
\begin{align*}
  \inf_{T \in [\tau_1, \tau_2]} c(A(T_P)|T, \PP) = \inf_{T \in [\tau_1, \tau_2]} \frac{\E_{\vec\gamma \sim \PP^T}\left[R(A(T_P)|T, \vec \gamma) \right]}{\E_{\vec\gamma \sim \PP^T}\left[\opt(T, \vec\gamma)\right]} \geq \gamma\,,
\end{align*}
for all request distributions $\PP$. In other words, an algorithm which is $\beta$-consistent on prediction $T_P$ and $\gamma$-competitive is guaranteed to get a $\beta$ fraction of the hindsight optimal reward in expectation if the prediction comes true and it is guaranteed to attain a $\gamma$ fraction of the hindsight optimal reward for every horizon $T \in [\tau_1,\tau_2]$ (whether or not it conforms to the prediction). Consistency and competitiveness are conflicting objectives and different decision makers might have different preferences over them. In particular, increasing consistency usually leads to lower competitiveness. Consequently, our goal is to find an algorithm which can trade off the two quantities, allowing us to interpolate between the known-horizon and the uncertainty-window settings.

Once again, the versatility of Algorithm~\ref{alg:dual-descent} and its ability to reduce the problem of finding the optimal algorithm to that of finding the optimal target consumption sequence comes to the fore. In particular, Proposition~\ref{prop:master} implies that Algorithm~\ref{alg:dual-descent} with target consumption sequence $\vec\lambda(T_P)$ for prediction $T_P \in [\tau_1,\tau_2]$ is $\beta'$-consistent for prediction $T_P$ and $\gamma'$-competitive with
\begin{align*}
  \beta' = c(\vec\lambda(T_P), T_P) - \epsilon \quad \text{ and } \quad \gamma' = \inf_{T\in [\tau_1,\tau_2]} c(\vec\lambda(T_P), T) - \epsilon \,.
\end{align*}

Therefore, given a prediction $T_P$ and a required level of competitiveness $\gamma' = \gamma - \epsilon$, we need to solve the following optimization problem in order to maximize consistency while achieving $\gamma'$-competitiveness:
\begin{align*}
  \max_{\vec\lambda} c(\vec\lambda, T_P) \quad \text{s.t.} \quad \inf_{T \in [\tau_1,\tau_2]} c(\vec\lambda, T) \geq \gamma \,.
\end{align*}
As in the uncertainty-window setting, we can rewrite this as an LP.
\begin{proposition}\label{prop:advice-LP}
  For budget $B$, uncertainty window $[\tau_1,\tau_2]$, predicted horizon $T_P \in [\tau_1,\tau_2]$ and required level of competitiveness $\gamma' = \gamma -\epsilon$, we have
  \begin{align*}
    \max_{\vec\lambda} \quad &c(\vec\lambda, T_P) &  &= &  \max_{\lambda, y} \quad &\frac{1}{T_P} \sum_{t=1}^{T_P} y_{T_P, t}\\
    \text{s.t.} \quad &\min_{T \in [\tau_1, \tau_2]} c(\vec\lambda, T) \geq \gamma &&& \text{s.t.}\quad &\gamma \leq \frac{1}{T} \sum_{t=1}^T y_{T,t} &\forall T \in [\tau_1, \tau_2]\\
    &&&&&y_{T,t} \leq \frac{\lambda_{t,j}}{\rho_{T,j}} &\forall j \in [m], T \in [\tau_1, \tau_2], t \in [T]\\
    &&&&&y_{T,t} \leq 1 &\forall T \in [\tau_1, \tau_2], t \in [T]\\
    &&&&&\sum_{t=1}^{\tau_2} \lambda_t \leq B\\
    &&&&&\lambda \geq 0
  \end{align*}
\end{proposition}

\begin{remark}
  Our framework can also accommodate distributional predictions about the horizon, leading to a similar LP with the only difference being an additional expectation over the predicted horizon $T_P$ in the objective.
\end{remark}

\begin{figure}[t!]
  \begin{subfigure}{0.33\textwidth}
    \includegraphics[width = \linewidth]{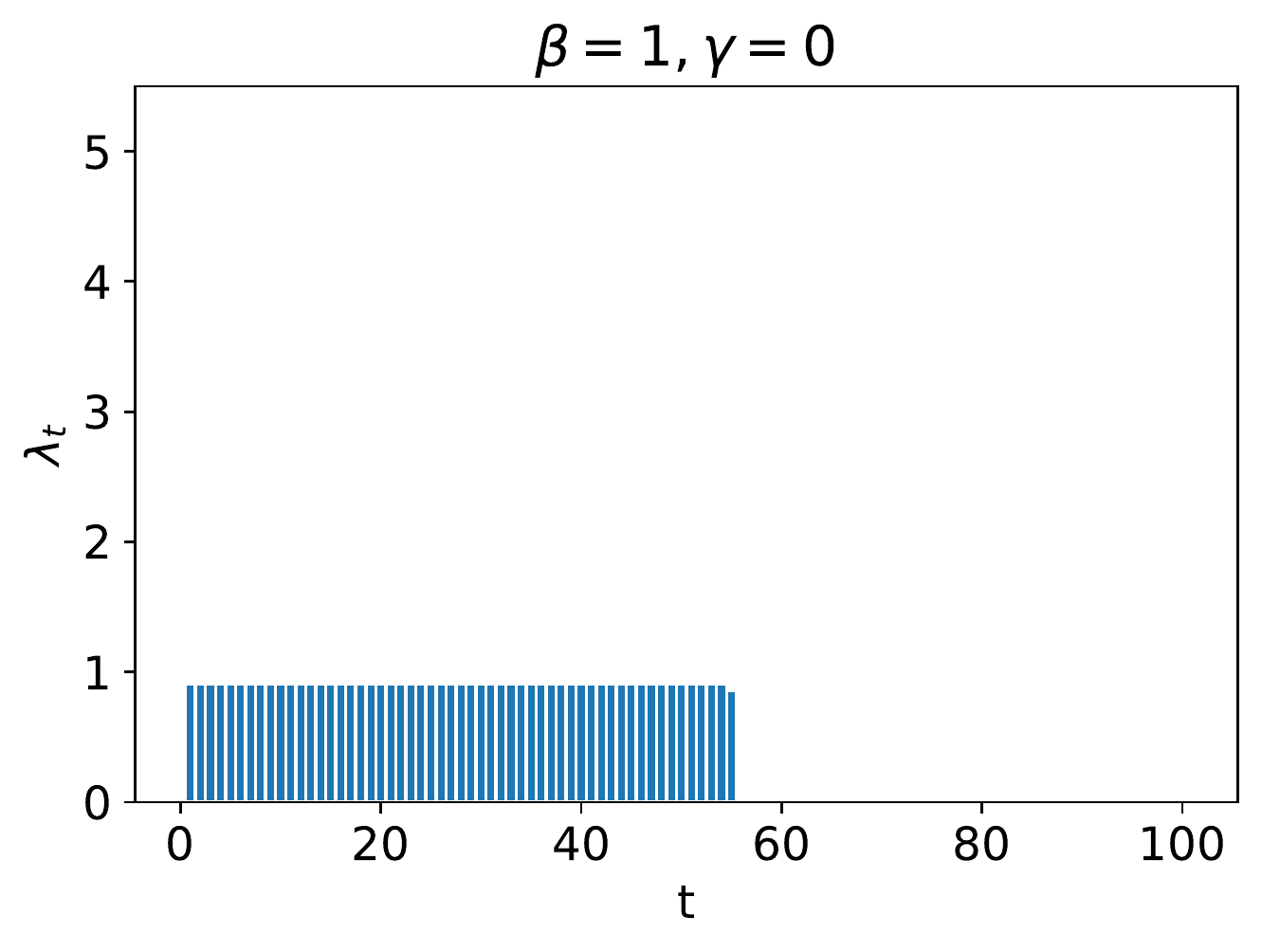}
  \end{subfigure}%
  \begin{subfigure}{0.33\textwidth}
    \includegraphics[width = \linewidth]{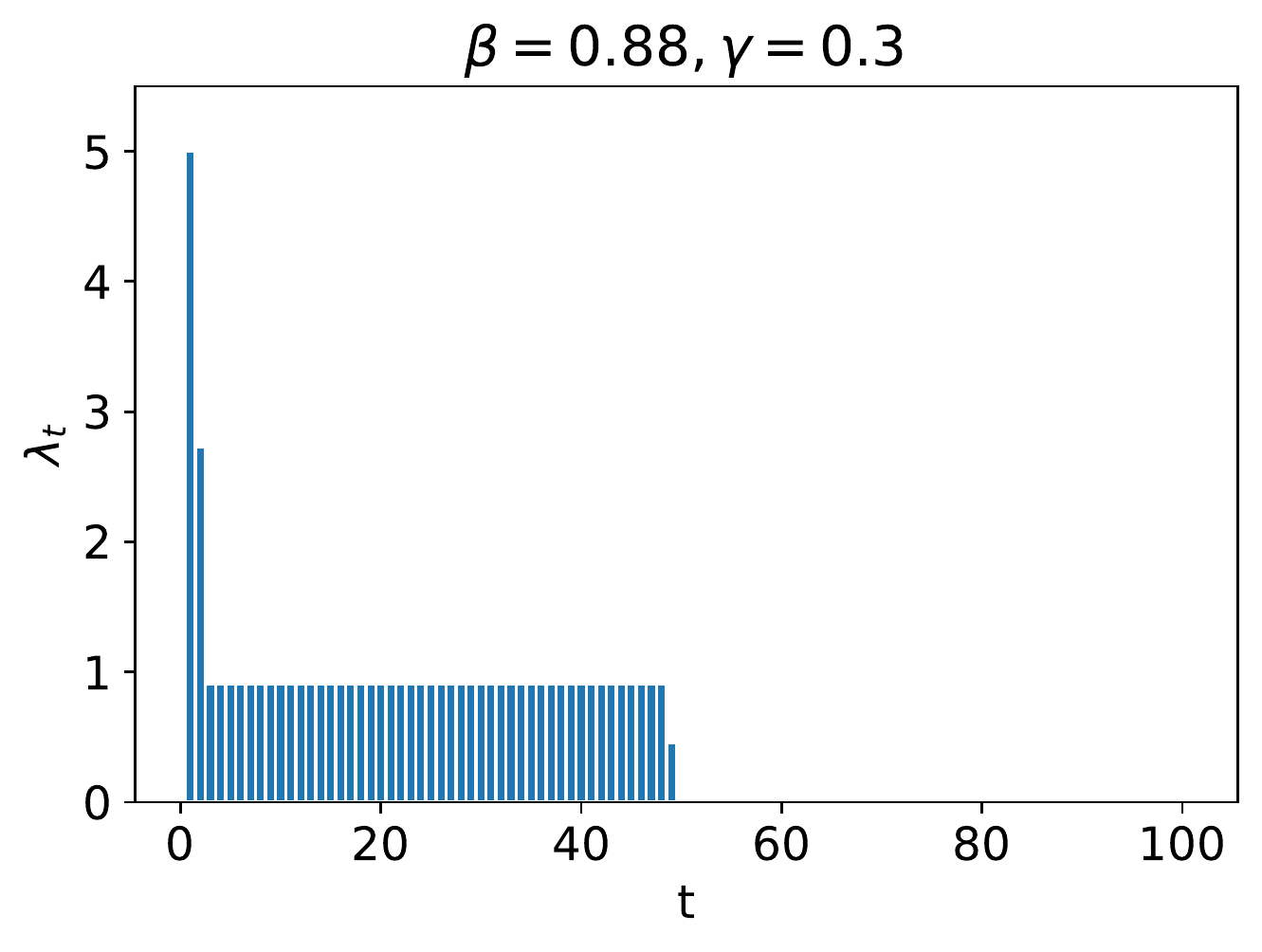}
  \end{subfigure}%
  \begin{subfigure}{0.33\textwidth}
    \includegraphics[width = \linewidth]{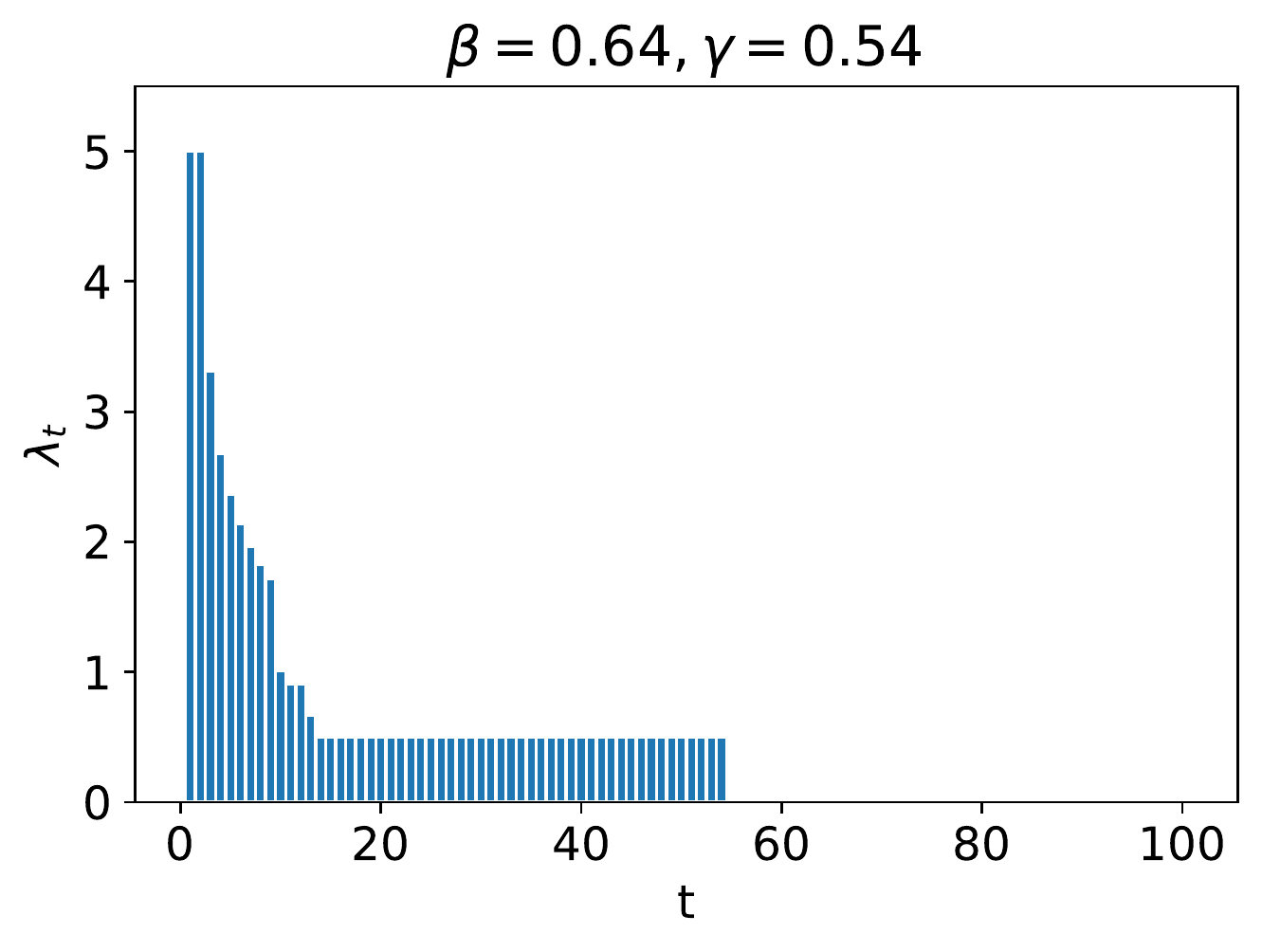}
  \end{subfigure}
  \caption{The optimal target consumption sequence for various values of required levels of competitiveness $\gamma$. Here, $m=1$, $\tau_1 = 10$, $\tau_2 = 100$, $B = 50$ and $T_P = 55$. Moreover, $0.54$ is the competitive ratio of the optimal target sequence, i.e., 0.54 is the optimal value of the LP in Proposition~\ref{prop:uncertainty-window-LP}. The sequences lose consistency and gain competitiveness from left to right.}
  \label{fig:target-seq}
\end{figure}

\begin{figure}[t!]
  \begin{subfigure}{0.33\textwidth}
    \includegraphics[width = \linewidth]{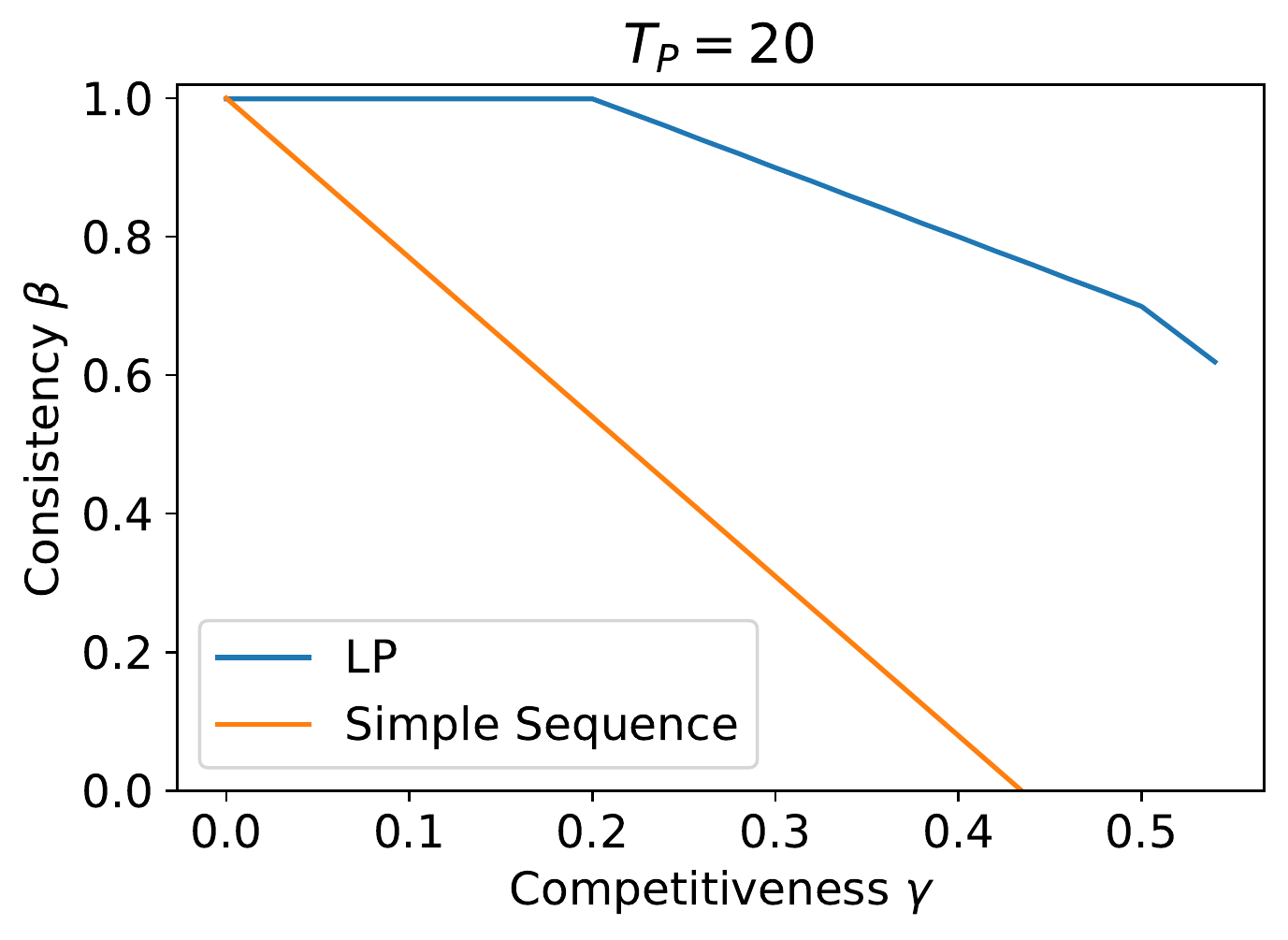}
  \end{subfigure}%
  \begin{subfigure}{0.33\textwidth}
    \includegraphics[width = \linewidth]{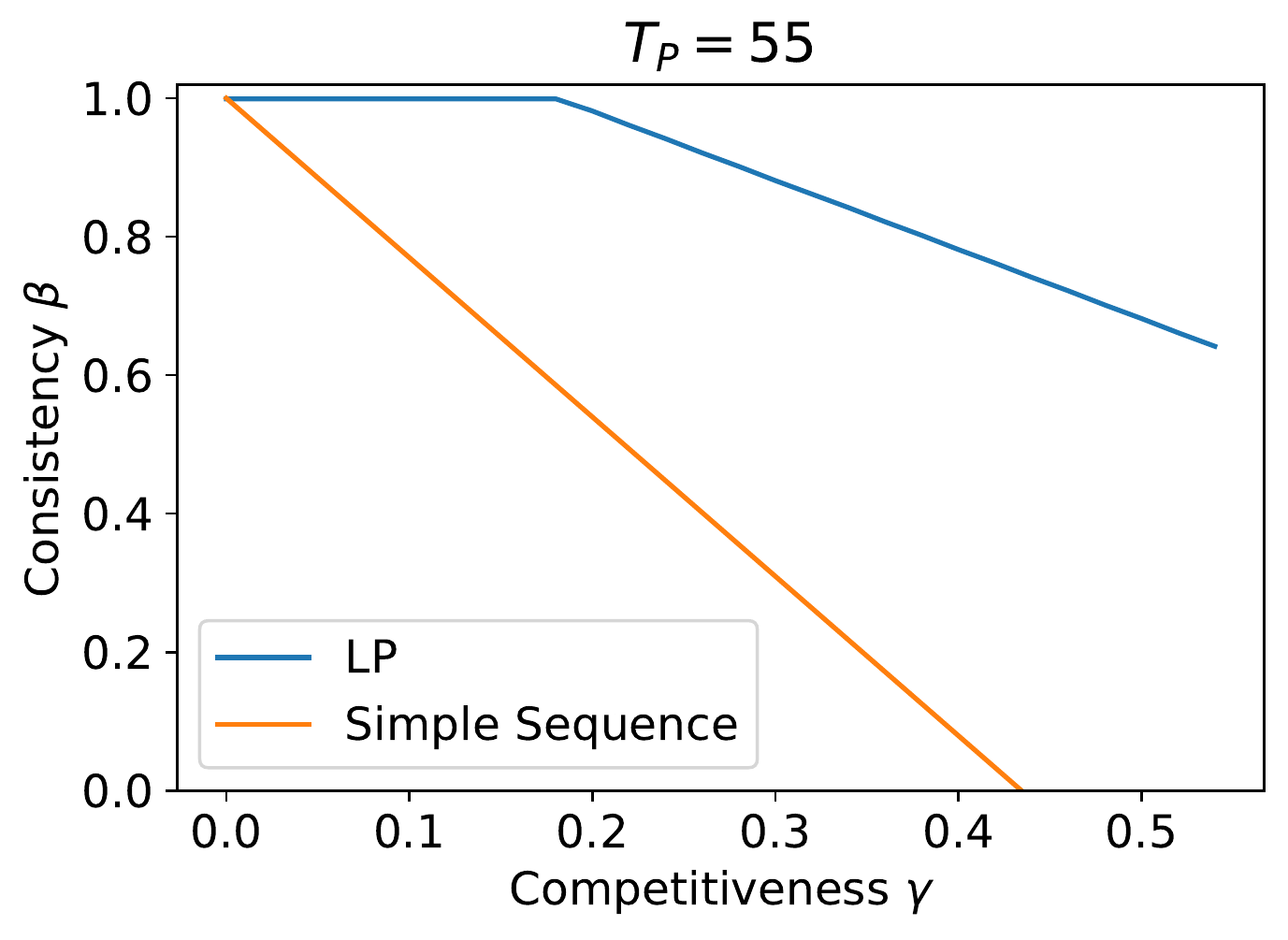}
  \end{subfigure}%
  \begin{subfigure}{0.33\textwidth}
    \includegraphics[width = \linewidth]{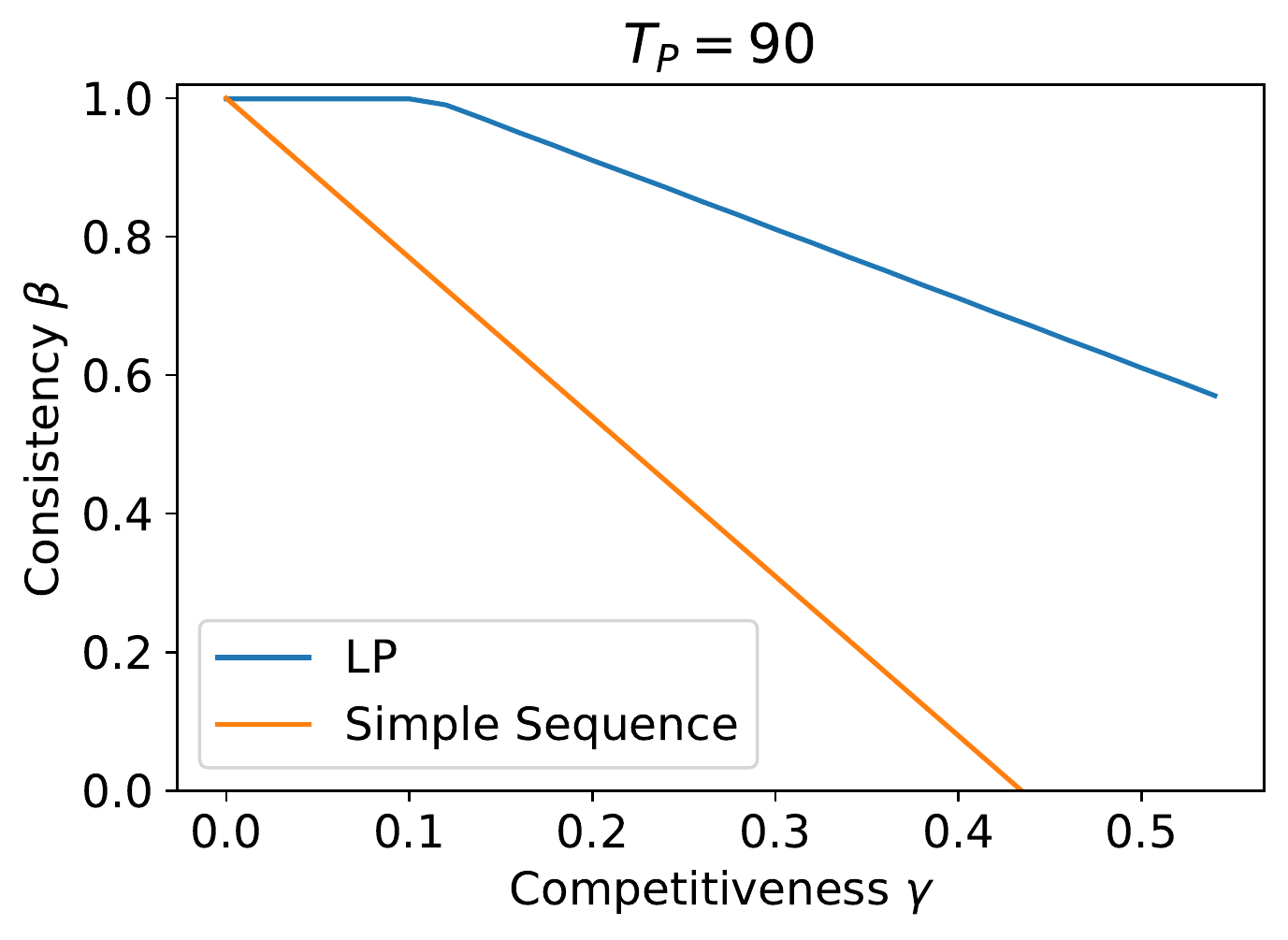}
  \end{subfigure}
  \caption{The consistency-competitiveness curves for the LP from Proposition~\ref{prop:advice-LP} and the simple target sequence from \eqref{eqnq:simple-advice-seq}, with predicted horizon $T_P \in \{20,55,90\}$. Here, $m=1$, $\tau_1 = 10$, $\tau_2 = 100$ and $B = 50$. Consistency $\beta = 1$ corresponds to the known-horizon setting and competitiveness $\gamma = 0.54$ corresponds to the largest possible competitiveness which can be obtained by optimizing the target sequence (Proposition~\ref{prop:uncertainty-window-LP}).}
  \label{fig:frontier}
\end{figure}

\begin{figure}[t!]
  \begin{subfigure}{0.33\textwidth}
    \includegraphics[width = \linewidth]{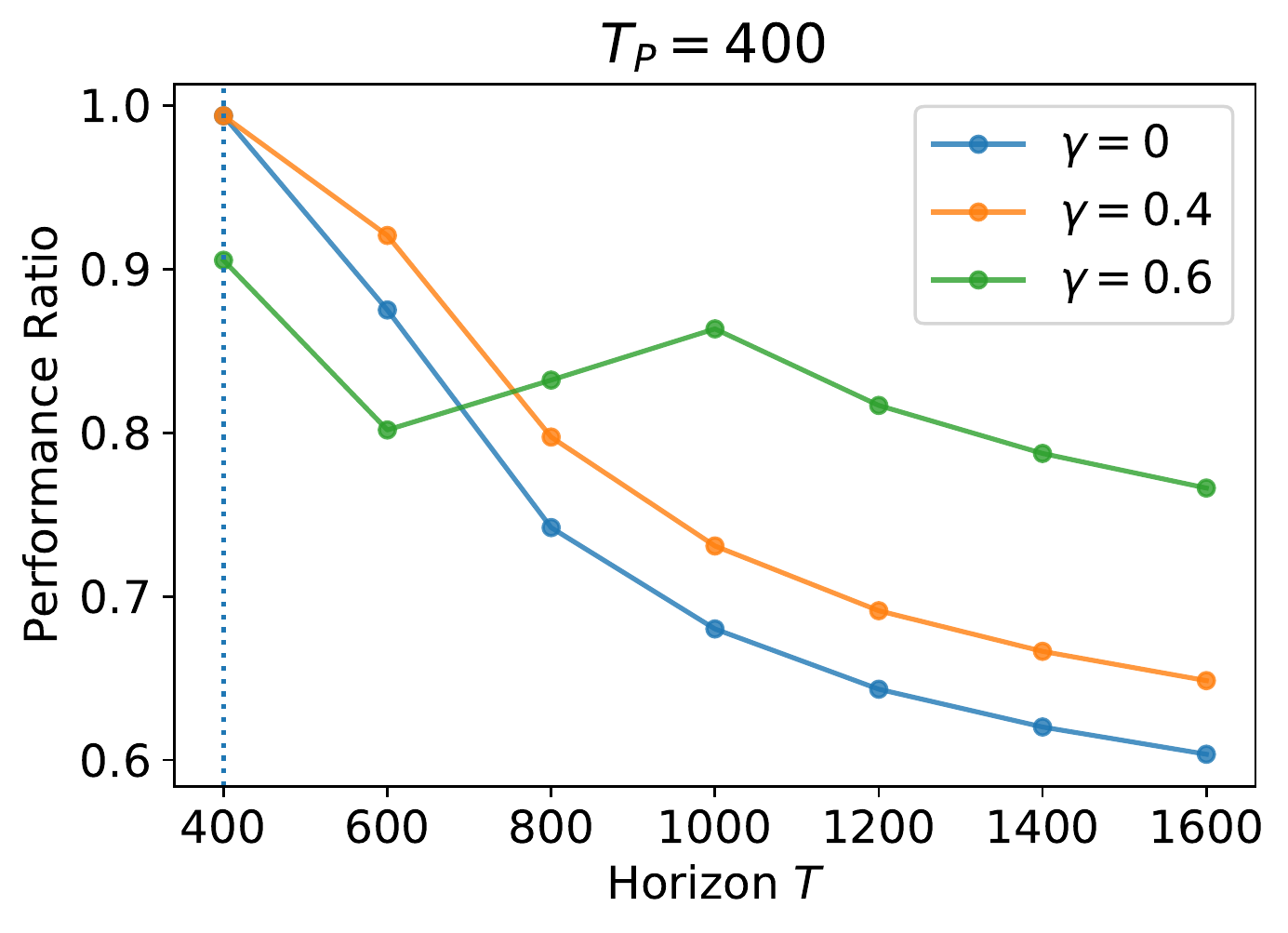}
  \end{subfigure}%
  \begin{subfigure}{0.33\textwidth}
    \includegraphics[width = \linewidth]{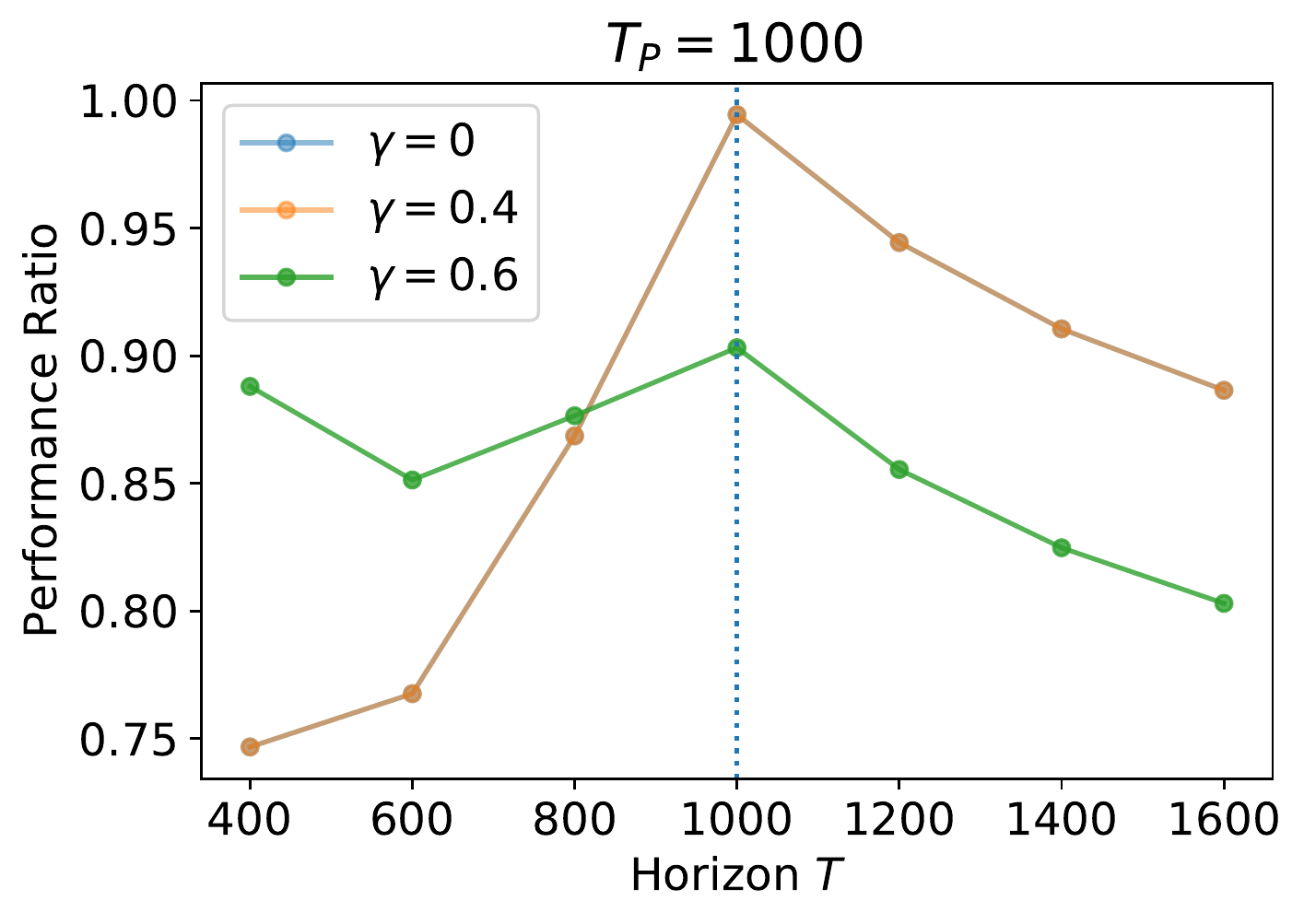}
  \end{subfigure}%
  \begin{subfigure}{0.33\textwidth}
    \includegraphics[width = \linewidth]{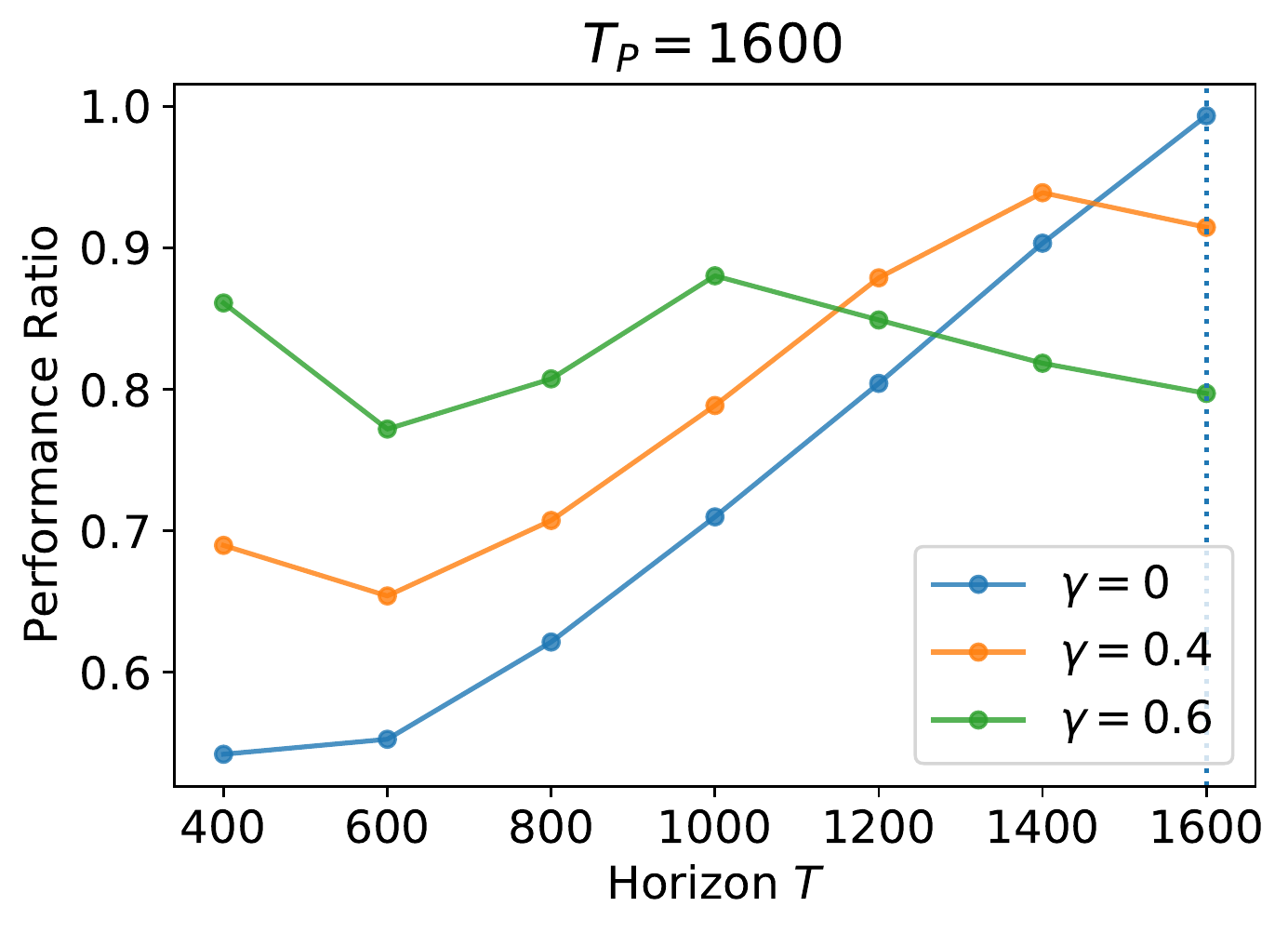}
  \end{subfigure}
  \caption{Average performance ratio (over 100 runs) of Algorithm~\ref{alg:dual-descent}, with target sequence from the LP in Proposition~\ref{prop:advice-LP} for different values of $\gamma$, on the uniform multi-secretary problem.}
  \label{fig:multisec}
\end{figure}

Observe that, when $\gamma = 0$ and the decision maker does not desire robustness, the LP in Proposition~\ref{prop:advice-LP} would output $\vec\lambda$ with $\lambda_t = \rho_{T_P}$ for $t \leq T_P$ and $\lambda_t = 0$ otherwise. Algorithm~\ref{alg:dual-descent} with this target consumption sequence is exactly the algorithm of \citet{balseiro2020best}, which yields a consistency of $\beta = 1$. On the other extreme is $\gamma$ being equal to the output of the LP in Proposition~\ref{prop:uncertainty-window-LP}, in which case the LP in Proposition~\ref{prop:advice-LP} would output a target sequence $\vec\lambda$ which maximizes the competitive ratio $\min_{T \in [\tau_1,\tau_2]} c(\vec\lambda, T)$. For values of $\gamma$ in between the two extremes, the LP in Proposition~\ref{prop:advice-LP} outputs a target consumption sequence which attempts to balance the two objectives, as can be seen in Figure~\ref{fig:target-seq}. This allows the decision maker to interpolate between the known-horizon and the uncertainty-window settings (see Figure~\ref{fig:frontier}).

Now, suppose the required level of competitiveness $\gamma' = \gamma - \epsilon$ is such that $\gamma = \alpha \cdot (1 + \ln(\tau_2/\tau_1))^{-1}$ for some $\alpha \in [0,1]$. Then, for predicted horizon $T_P \in [\tau_1, \tau_2]$, consider the following simple target consumption sequence
\begin{align}\label{eqnq:simple-advice-seq}
  \lambda_t \coloneqq \begin{cases}
    \frac{\alpha}{1 + \ln(\tau_2/\tau_1)} \cdot \frac{B}{\tau_1} + (1 - \alpha) \cdot \frac{B}{T_P} = \frac{\alpha}{1 + \ln(\tau_2/\tau_1)} \cdot \rho_{\tau_1} + (1 - \alpha) \cdot \rho_{T_P} &\text{if } t \leq \tau_1 \,,\\
    \frac{\alpha}{1 + \ln(\tau_2/\tau_1)} \cdot \frac{B}{t} + (1 - \alpha) \cdot \frac{B}{T_P} = \frac{\alpha}{1 + \ln(\tau_2/\tau_1)} \cdot \rho_t + (1 - \alpha) \cdot \rho_{T_P} &\text{if } \tau_1+1 \leq t \leq T_P\,,\\
    \frac{\alpha}{1 + \ln(\tau_2/\tau_1)} \cdot \frac{B}{t} = \frac{\alpha}{1 + \ln(\tau_2/\tau_1)} \cdot \rho_t &\text{if } T_P+1 \leq t \leq \tau_2\,.
  \end{cases}
\end{align}

The target sequence $\vec\lambda$ is simply a sum of two target sequences: (i) The first part is an $\alpha$-scaled-down version of the simple target sequence from \eqref{eqn:simple-target-sequence}, which ensures $\alpha \cdot (1 + \ln(\tau_2/\tau_1)^{-1})$ competitiveness; (ii) The second is a $(1 - \alpha)$-scaled-down version of the target sequence which spends $\rho_{T_P} = B/T_P$ evenly and is optimal when the prediction were true. $\vec\lambda$, as defined in \eqref{eqnq:simple-advice-seq}, is a feasible solution to the optimization of Proposition~\ref{prop:advice-LP}, which allows us to establish the following closed-form guarantee.

\begin{proposition}\label{prop:simple-advice-seq}
  Let $\epsilon$ be as in Proposition~\ref{prop:master}. Consider a target level of competitiveness $\gamma - \epsilon$, where $\gamma = \alpha \cdot (1 + \ln(\tau_2/\tau_1))^{-1}$ for some $\alpha \in [0,1]$. Let $\vec\lambda(T_P)$ be an optimal solution of the LP in Proposition~\ref{prop:advice-LP} and let $A(T_P)$ denote Algorithm~\ref{alg:dual-descent} with the target sequence $\vec\lambda(T_P)$. Then, for every request distribution $\PP$ and predicted horizon $T_P \in [\tau_1, \tau_2]$, we have
  \begin{align*}
    c(A(T_P)|T_P, \PP) \geq \left(1 - \alpha+ \frac{\alpha}{1 + \ln(\tau_2/\tau_1)} \right) - \epsilon \quad \text{and} \quad \inf_{T, T_P}\ c(A(T_P)|T, \PP) \geq \frac{\alpha}{1 + \ln(\tau_2/\tau_1)} -\epsilon \,.
  \end{align*}
\end{proposition}


Note that the target sequence in \eqref{eqnq:simple-advice-seq} is just one particular target sequence and the LP in Proposition~\ref{prop:advice-LP} computes the optimal target sequence, and consequently the latter always performs better. This domination in performance is depicted in Figure~\ref{fig:frontier}, where the consistency-competitiveness curve the simple sequence (in orange) lies entirely below the curve from Proposition~\ref{prop:advice-LP} (blue curve). 

\textbf{Numerical Experiment.} We evaluated our algorithm (Algorithm~\ref{alg:dual-descent} with target sequence from Proposition~\ref{prop:advice-LP}) on the multi-secretary problem with uniform rewards and the results are summarized in Figure~\ref{fig:multisec}. In this experiment, the request distribution captures the uniform multi-secretary problem: each request $\gamma = (f,b, \X)$ has reward $f(x) = r \cdot x$ for $r \sim \text{Unif}([0,1])$, consumption $b(x) = x$ and an accept/reject action space $\X = \{0,1\}$. Moreover, $\tau_1 = 400$, $\tau_2 = 1600$, $B = 500$, $\eta = 0.03$, $\mu_1 = 0.5$ and $h(\cdot) = \|\cdot\|_2$. As expected, smaller values of $\gamma$ lead to better performance when the true horizon $T$ is close to the prediction $T_P$, but this comes at the expense of lower worst-case reward (minimum competitive ratio over all possible values of the horizon $T \in [\tau_1, \tau_2]$). Recall that $\gamma = 0$ represents the algorithm of \citet{balseiro2020best} with horizon $T_P$. Our experiment demonstrates its fragility to traffic spikes: if the number of requests turns out to be 3 times the predicted traffic of $T_P = 400$, the algorithm of \citet{balseiro2020best} achieves a drastically lower performance ratio than our algorithm with $\gamma = 0.6$.

%

\section{Bypassing the LP: A Faster Algorithm}\label{sec:faster-alg}

Even though the LPs of Proposition~\ref{prop:uncertainty-window-LP} and Proposition~\ref{prop:advice-LP} compute the optimal target consumption sequence in polynomial time, they do not exploit the structure of the problem and are not desirable in large-scale domains like internet advertising where speed is of the essence. To remedy this, we next develop a faster algorithm to compute the optimal target consumption sequence; this algorithm more directly exploits the structure of the problem. The algorithm (Algorithm~\ref{alg:opt-target-sequence}) will rely on the following observation about $c(\vec\lambda, T)$:
\begin{align}\label{eqn:comp-ratio-restate}
    c(\vec\lambda, T) =  \frac{1}{T}\sum_{t=1}^T \min \left\{ \min_{1 \leq j \leq m} \frac{\lambda_{t,j}}{\rho_{T,j}}, 1 \right\} \leq \min_{1 \leq j \leq m} \frac{1}{T}\sum_{t=1}^T \min \left\{\frac{\lambda_{t,j}}{\rho_{T,j}}, 1 \right\} = \min_{1 \leq j \leq m} \frac{1}{B_j}\sum_{t=1}^T \min \left\{\lambda_{t,j}, \rho_{T,j} \right\}
\end{align}
where the last equality follows from multiplying and diving by $\rho_{T,j} = B_j/T$. Moreover, note that the above inequality is tight when $\frac{\lambda_{t,j}}{\rho_{T,j}} = \frac{\lambda_{t,k}}{\rho_{T,k}}$ for all $j,k \in [m], t \in [T] \,.$

Therefore, any target sequence $\vec\lambda$ which is $\beta$-consistent for prediction $T_P$, i.e., $c(\vec\lambda, T_P) \geq \beta$, and $\gamma$-competitive, i.e. $\min_{T \in [\tau_1, \tau_2]} c(\vec\lambda, T) \geq \gamma$, satisfies the following inequalities for all $j \in [m]$:
\begin{align*}
  \sum_{t=1}^{T_P} \min\{\lambda_{t,j}/\rho_{T_P,j}\} \geq \beta \cdot B_j \quad \text{ and } \quad \sum_{t=1}^T \min\{\lambda_{t,j}/\rho_{T,j}\} \geq \gamma \cdot B_j \quad \forall\ T \in [\tau_1, \tau_2] \,.
\end{align*}
Algorithm~\ref{alg:opt-target-sequence} minimizes $\sum_{t=1}^{\tau_2} \lambda_{t,j}$ while maintaining the above property. And as a consequence, we can show that $\beta$ consistency on $T_P$ and $\gamma$ competitiveness are attainable if and only if Algorithm~\ref{alg:opt-target-sequence} returns TRUE. Given this property, it is a straightforward exercise to use binary search in conjunction with Algorithm~\ref{alg:opt-target-sequence} to compute the optimal solution to the LPs in Proposition~\ref{prop:uncertainty-window-LP} and Proposition~\ref{prop:advice-LP} up to arbitrary precision (For completeness, we provide details in Appendix~\ref{appendix:faster-alg}).

\begin{algorithm}[t!]
  \SetAlgoLined
  {\bf Input:} Budget $B \in \R_{++}^m$, uncertainty window $[\tau_1, \tau_2]$, prediction $T_P$, required level of consistency $\beta \in [0,1]$ and required level of competitiveness $\gamma \in [0,\beta]$. \\

  \textbf{Initialize:} $\lambda_{t,j} \leftarrow 0\quad \forall t \in [\tau_2], j \in [m]$\\
  \For{$T = \tau_2$ to $\tau_1$}{
    \For{$t = 1$ to $T$}{
      \begin{equation}\label{eqn:opt-seq-update}
        \lambda_{t,j} \leftarrow \begin{cases}
          \lambda_{t,j} + \min\left\{\rho_{T,j} - \lambda_{t,j}, \beta \cdot B_j - \sum_{s=1}^T \lambda_{s,j} \right\}^+ &\text{if } T = T_P,\\
          \lambda_{t,j} + \min\left\{\rho_{T,j} - \lambda_{t,j}, \gamma \cdot B_j - \sum_{s=1}^T \lambda_{s,j} \right\}^+ &\text{if } T \neq T_P
        \end{cases}
      \end{equation}
    }
  }

  {\bf Return:} TRUE if $\sum_{t=1}^{\tau_2} \lambda_{t} \leq B_j$; else FALSE.
  \caption{Optimal Target Consumption Sequence}
  \label{alg:opt-target-sequence}
\end{algorithm}

\begin{theorem}\label{thm:faster-alg}
  Given budget $B \in \R_{++}^m$, uncertainty window $[\tau_1, \tau_2]$, prediction $T_P$, required level of consistency $\beta \in [0,1]$ and required level of competitiveness $\gamma \in [0,\beta]$ as input, let $\vec\lambda^*$ be the sequence computed by Algorithm~\ref{alg:opt-target-sequence}. Then,
  \begin{enumerate}
    \item $c(\vec\lambda^*, T_P) \geq \beta$ and $\min_{T \in [\tau_1, \tau_2]} c(\vec\lambda^*, T) \geq \gamma$
    \item $\sum_{t=1}^{\tau_2} \lambda^*_t \leq B$ if and only if there exists a target consumption sequence $\vec\lambda'$ (with $\sum_{t=1}^{\tau_2} \lambda'_t \leq B$) which satisfies $c(\vec\lambda', T_P) \geq \beta$ and $\min_{T \in [\tau_1, \tau_2]} c(\vec\lambda', T) \geq \gamma\,.$
  \end{enumerate}
\end{theorem}

 \begin{figure}
   \centering
   \includegraphics[width = 0.4\textwidth]{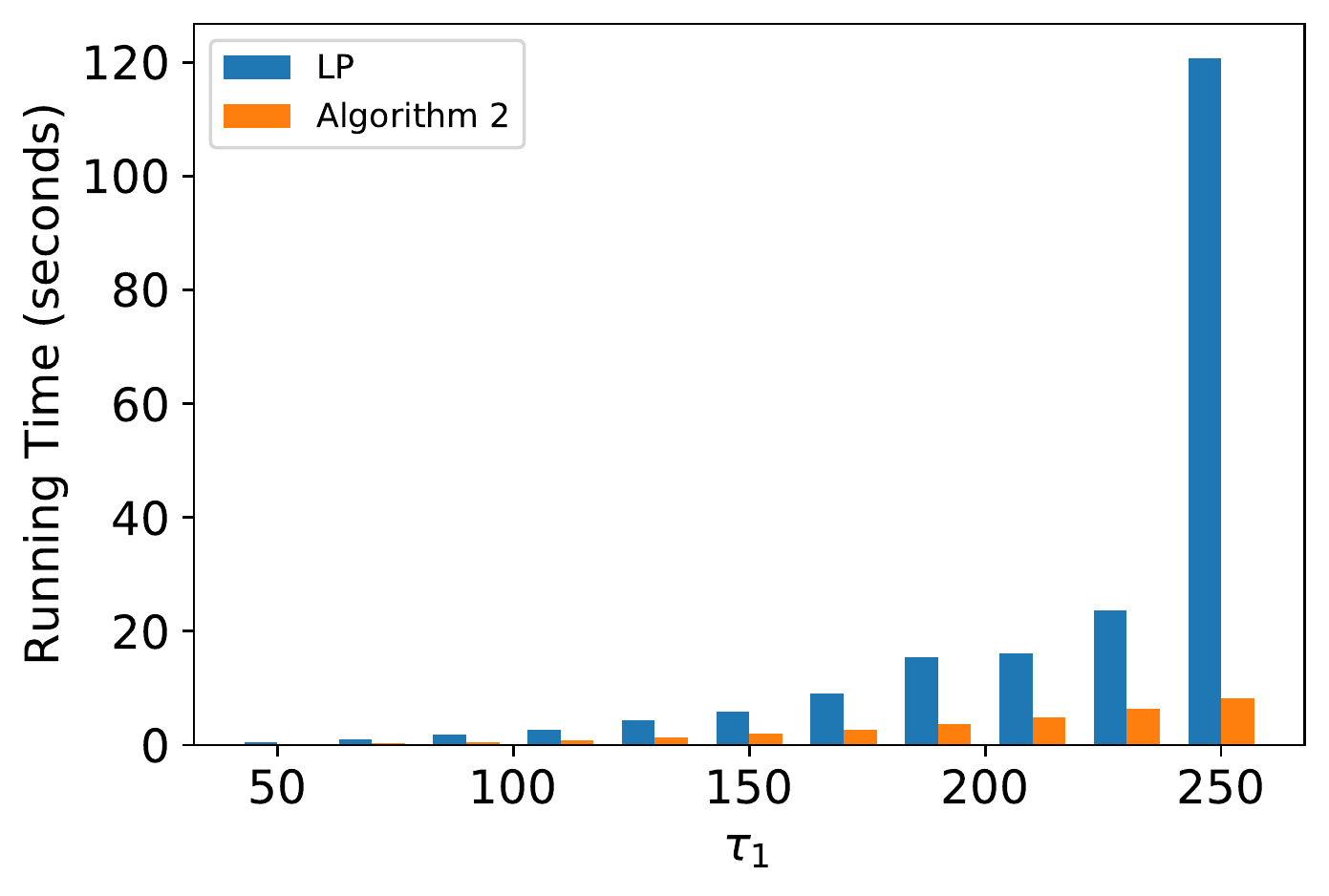}
   \caption{A comparison of the running times of the LP from Proposition~\ref{prop:uncertainty-window-LP} solved using Gurobi Optimizer version 9.1.2 build v9.1.2rc0 (mac64) and Algorithm~\ref{alg:opt-target-sequence} run on Python 3.7.6 without the use of any advanced libraries. The minimum runtime from 10 runs was selected for GUROBI and the maximum runtime from 10 runs was selected for Algortihm~\ref{alg:opt-target-sequence}. Both algorithms were limited to a single thread to ensure parity of computational resources. Here, $\tau_2 = 3 \cdot \tau_1$  and $B = 1.5 \cdot \tau_1$ for all values of $\tau_1$. For $\tau_1 \geq 300$, GUROBI did not terminate with a solution even after 10 min, while Algorithm~\ref{alg:opt-target-sequence} consistently did so under 10 seconds.}
   \label{fig:runtime}
 \end{figure}

Observe that there can be at most $\tau_2^2$ updates of the target sequence $\vec\lambda$ (as given in \eqref{eqn:opt-seq-update}) during the run of Algorithm~\ref{alg:opt-target-sequence}. One can maintain and iteratively update $\sum_{s=1}^T \lambda_{s,j}$ after the completion of each iteration of the inner and outer \textbf{For} loops to perform the update in constant time. Therefore, the runtime complexity of Algorithm~\ref{alg:opt-target-sequence} is $O(m \cdot \tau_2^2)$, which is faster than any known general-purpose LP solver applied to the LP in Proposition~\ref{prop:uncertainty-window-LP} or Proposition~\ref{prop:advice-LP}. We also empirically observed this difference in running times between the LP of Proposition~\ref{prop:uncertainty-window-LP} and Algorithm~\ref{alg:opt-target-sequence} (see Figure~\ref{fig:runtime}).

%

\section{Conclusion}

 We develop and analyze a generalized version of dual descent which can incorporate variable target consumption sequences (Algorithm~\ref{alg:dual-descent}), thereby reducing the complicated problem of finding an algorithm for online resource allocation under horizon uncertainty to the much simpler (and convex) problem of optimizing the target sequence. We then demonstrate the power of this reduction by showing that, with the optimal target sequence, Algorithm~\ref{alg:dual-descent} achieves a near-optimal competitive ratio when only upper and lower bounds on the horizon are known. We also provide a way to smoothly interpolate between the previously-studied known-horizon setting and the uncertainty-window setting through the Algorithms-with-Predictions framework, thereby providing a robust approach to online allocation which allows the decision-maker to tailor the degree to robustness to their requirements. Our algorithms have the added advantage of simplicity and speed because they do not require the decision-maker to solve any large linear programs.

We leave open the problem of closing the gap between our lower and upper bounds on the competitive ratio by accounting for the $e \cdot \ln\ln(\tau_2/\tau_1)$ discrepancy. Although this gap is not large asymptotically, closing it will likely result in a deeper understanding of the problem. It would also be interesting to explore whether algorithms which operate in the primal space can similarly benefit from employing a variable target consumption sequence. Finally, when both the distribution of requests and the distribution of the horizon are known in advance, it is worth studying if it is possible to achieve a constant/logarithmic regret against an appropriately defined benchmark (see for example \citealt{bumpensanti2020re, vera2019bayesian} for similar results when the horizon is known).


\singlespace

\bibliographystyle{plainnat}
\bibliography{refs}

\appendix

\newpage
\pagenumbering{arabic}\renewcommand{\thepage}{ec \arabic{page}}

\renewcommand{\theequation}{\thesection-\arabic{equation}}

\section{Proofs for Section~\ref{sec:lower-bound}}\label{appendix:lower-bound}

We begin by formally stating and proving weak duality

\begin{proposition}\label{prop:weak-duality}
	For every $\mu \in \R_+$, $T \geq 1$ and $\vec\gamma \in \S^T$, we have $\opt(T, \vec\gamma) \leq D(\mu|T, \vec\gamma)$.
\end{proposition}
\begin{proof}
    Consider any $x \in \prod_t \X_t$ such that $\sum_{t=1}^T b_t(x_t) \leq B$. Then, for $\mu \geq 0$, we have
    \begin{align*}
        D(\mu | T, \vec\gamma) \geq \left\{\sum_{t=1}^T f_t(x_t) \right\} + \mu^\top \left( B - \sum_{t=1}^T b_t(x_t) \right) \geq \sum_{t=1}^T f_t(x_t)
    \end{align*}
     Since $x \in \prod_t \X_t$ satisfied $\sum_{t=1}^T b_t(x_t) \leq B$ was otherwise arbitrary, we have shown $\opt(T\vec\gamma) \leq D(\mu| T, \vec\gamma)$.
\end{proof}

\subsection{Proof of Lemma~\ref{lemma:dual-prop}}
\begin{proof}[Proof of Lemma~\ref{lemma:dual-prop}]
    The convexity of $\D(\cdot| \lambda, \PP)$ follows from part (a) and  the fact that the dual objective $D(\mu|T,\vec\gamma)$ is always convex since it is a supremum of a collection of linear functions.
    \begin{itemize}
        \item[(a)] Shown in equation~\ref{eq:dual-decomp}.
        \item[(b)] For $a \in [0,1]$, we have
        \begin{align*}
            \D(\mu | a \cdot \lambda, \PP) - a \cdot\D(\mu|\lambda, \PP) = (1 - a) \E_{\gamma \sim \PP}\left[f^*(\mu) \right] \geq 0
        \end{align*}
        where we have used the fact that $f^*(\mu) \geq 0$ for all $\gamma = (f,b, \X) \in \S)$, which holds because $0 \in \X$.
        \item[(c)] For $\lambda \leq \kappa$, we have
        \begin{align*}
            \D(\mu | \kappa, \PP) - \D(\mu|\lambda, \PP) = (\kappa - \lambda)^\top \mu \geq 0
        \end{align*}
        where the inequality follows from the fact that $\kappa - \lambda,\ \mu \geq 0$.\qedhere
    \end{itemize}
\end{proof}

\subsection{Proof of Theorem~\ref{thm:regret}}

\begin{proof}[Proof of Theorem~\ref{thm:regret}]
	Fix an arbitrary $T \geq 1$. To simplify notation, define
	\begin{align*}
	    \frac{\lambda_t}{\rho_T} \coloneqq \min_{1 \leq j \leq m} \frac{\lambda_{t,j}}{\rho_{T,j}}
	\end{align*}
	We split the proof into four steps: (1) The first step involves lower bounding the performance of our algorithm in terms of single-period duals and the complementary slackness term; (2) The second step involves bounding the complementary slackness term using standard regret analysis of mirror descent; (3) The third step involves bounding the optimal from above using the single-period dual; (4) The final step puts it all it all together. Our proof significantly generalizes the proof of \citet{balseiro2020best}, who established this result for the special case of a constant target consumption rate $\lambda_t = \rho_T$ for all $t \leq T$. The main technical contribution of the current proof is establishing a general performance guarantee for dual-mirror descent for \emph{all} target consumption sequences, which will prove critical in getting an asymptotically-near-optimal competitive ratio for our model. This involves establishing a novel target-rate-dependent lower bound on the algorithm's reward (Step~1), a novel target-rate-dependent upper bound on the optimal reward (Step~3), and a new way to put these bounds together (Step~4). 
	
	\paragraph{Step 1: Lower bound on algorithm's reward.} Consider the filtration $\mathcal{F} = \{\sigma(\xi_t)\}_t$, where $\xi_t = \{\gamma_1, \dots, \gamma_t\}$ is the set of all requests seen till time $t$ and $\sigma(\xi_t)$ is the sigma algebra generated by it. Note that Algorithm~\ref{alg:dual-descent} only collects rewards when there are enough resources left. Let $\tA$ be first time less than $T$ for which there exists a resource $j$ such that $\sum_{t=1}^{\tA} b_{t,j}(x_t) + \bar{b} \ge B_j$. Here, $\tA = T$ if this inequality is never satisfied. Observe that $\tA$ is a stopping time w.r.t. $\mathcal{F}$ and it is defined so that we cannot violate the resource constraints before $\tA$. In particular, $x_t = \tilde{x}_t$ for all $t \leq \tA$. Therefore, we get
		\begin{align*}
			f_t(x_t) = f^*(\mu_t) + \mu_t^\top b_t(x_t)
		\end{align*}
		Observe that $\mu_t$ is measurable w.r.t. $\sigma(\xi_{t-1})$ and $\gamma_t$ is independent of $\sigma(\xi_{t-1})$, which allows us to take conditional expectation w.r.t. $\sigma(\xi_{t-1})$ to write
		\begin{align}\label{eq:bound-one-period}
			\mathbb E\left[ f_t(x_t)  | \sigma(\xi_{t-1}) \right] \nonumber
			&=\E_{\gamma_t \sim \PP}\left[f^*(\mu_t)\right] + \mu_t^{\top} \lambda_t + \mu_t^\top \left( \E\left[ b_t(x_t)| \sigma(\xi_{t-1}) \right] - \lambda_t \right)    \nonumber\\
			&=   \tilde D(\mu_t| \lambda_t, \PP) - \mathbb E\left[ \mu_t^\top \left( \lambda_t - b_t(x_t) \right) | \sigma(\xi_{t-1}) \right]
		\end{align}
		where the second equality follows the definition of the single-period dual function.
		
		Define $Z_t = \sum_{s=1}^t \mu_s^\top \left(\lambda_s - b_{s} (x_s)\right) - \mathbb E\left[ \mu_s^\top \left(\lambda_s - b_{s} (x_s) \right) | \sigma(\xi_{s-1}) \right]$. Then, $\{Z_t\}_t$ is a martingale w.r.t. the filtration $\mathcal{F}$ because $Z_t \in \sigma(\xi_t)$ and $\E[Z_{t+1} | \sigma(\xi_t)] = Z_t$. As $\tA$ is a bounded stopping time w.r.t. $\mathcal{F}$, the Optional Stopping Theorem yields $\E\left[Z_{\tA}\right] = 0$. Therefore,
		\begin{align*}
			\E\left[\sum_{t=1}^{{\tA}} \mu_t^\top \left(\lambda_t - b_t(x_t)\right) \right]
			= \E\left[\sum_{t=1}^{{\tA}} \mathbb E\left[ \mu_t^\top \left(\lambda_t - b_t(x_t)\right) | \sigma(\xi_{t-1}) \right] \right]\,.
		\end{align*}
		A similar argument yields
		\begin{align*}
			\E\left[\sum_{t=1}^{{\tA}} f_t(x_t) \right]
			= \E\left[\sum_{t=1}^{{\tA}} \mathbb E\left[ f_t(x_t) | \sigma(\xi_{t-1}) \right] \right]\,.
		\end{align*}
		Hence, summing over \eqref{eq:bound-one-period} and taking expectations, we get
		\begin{align}\label{eq:alg-lower-bound}
			\mathbb E\left[ \sum_{t=1}^{\tA} f_t(x_t)  \right]
			&=
			\mathbb E\left[ \sum_{t=1}^{\tA} \tilde D(\mu_t| \lambda_t, \PP) \right ]  - \mathbb E\left[ \sum_{t=1}^{\tA} \mu_t^\top \left(  \lambda_t - b_t(x_t) \right) \right] \nonumber\\
			&\geq
			\mathbb E\left[ \sum_{t=1}^{\tA} \tilde D\left(\mu_t| \min\{\lambda_t/\rho_T, 1 \}\cdot \rho_T, \PP \right) \right ]  - \mathbb E\left[ \sum_{t=1}^{\tA} \mu_t^\top \left(  \lambda_t - b_t(x_t) \right) \right] \nonumber\\
			&\geq
			\mathbb E\left[ \sum_{t=1}^{\tA} \min\{\lambda_t/\rho_T, 1 \}\cdot \tilde D\left(\mu_t|  \rho_T, \PP \right) \right ]  - \mathbb E\left[ \sum_{t=1}^{\tA} \mu_t^\top \left(  \lambda_t - b_t(x_t) \right) \right] \nonumber\\
			&=
			\mathbb E\left[ \sum_{t=1}^{\tA}\min\{\lambda_t/\rho_T, 1 \} \cdot \sum_{t=1}^{\tA} \frac{\min\{\lambda_t/\rho_T, 1 \}}{\sum_{s=1}^{\tA}\min\{\lambda_s/\rho_T, 1 \}} \cdot \tilde D\left(\mu_t|  \rho_T, \PP \right) \right ]  - \mathbb E\left[ \sum_{t=1}^{\tA} \mu_t^\top \left(  \lambda_t - b_t(x_t) \right) \right] \nonumber\\
			&\ge
			\mathbb E\left[ \sum_{t=1}^{\tA}\min\{\lambda_t/\rho_T, 1 \} \cdot \tilde D\left(\bar\mu_{\tA}|  \rho_T, \PP \right) \right ] - \mathbb E\left[ \sum_{t=1}^{\tA} \mu_t^\top \left( \lambda_t - b_t(x_t) \right) \right]\,
		\end{align}
	where
	\begin{align*}
		\bar\mu_{\tA} = \sum_{t=1}^{\tA} \frac{\min\{\lambda_t/\rho_T, 1 \} \cdot \mu_t}{\sum_{s=1}^{\tA}\min\{\lambda_s/\rho_T, 1 \}}.
	\end{align*}
	The first inequality follows from part (c) of Lemma~\ref{lemma:dual-prop}, the second inequality follows from part (b) of Lemma~\ref{lemma:dual-prop} and the third inequality follows from the convexity of the single-period dual function (Lemma~\ref{lemma:dual-prop}).
	
	\paragraph{Step 2: Complementary slackness.} Define $w_t(\mu) \coloneqq \mu^\top(\lambda_t - b_t(x_t))$. Then, Algorithm~\ref{alg:dual-descent} can be seen as running online mirror descent on the choice of the dual variables with linear losses $\{w_t(\cdot)\}_t$. The gradients of these loss functions are given by $\nabla w_t(\mu) = \lambda_t - b_t(x_t)$, which satisfy $\|\nabla w_t(\mu)\|_\infty \leq \|b_t(x_t)\|_\infty + \|\lambda(t)\|_\infty \leq \bar b + \bar \lambda$. Therefore, Proposition 5 of \citet{balseiro2020best} implies that for all $\mu \in \R_+^m$:
	\begin{align}\label{eq:regret-omd}
		\sum_{t=1}^{\tA} w_t(\mu_t) - w_t(\mu) \le E(\tA, \mu) \le E(T,\mu)\,,
	\end{align}
	where $E(t,\mu) = \frac 1 {2\sigma} (\bar{b} + \bar\lambda)^2 \eta \cdot t + \frac{1}{\eta} V_h(\mu,\mu_1)$ is the regret bound of the online mirror descent algorithm after $t$ iterations, and the second inequality follows because $\tA \le T$ and the error term $E(t,\mu)$ is increasing in $t$.
	
	\paragraph{Step 3: Upper bound on the optimal reward.} For every $\tA \in [1, T]$, we have
		\begin{align}\label{eq:opt-upper-bound}
			&\E_{\vec\gamma \sim \D^T} \left[ c(\vec\lambda, T) \cdot \opt(T, \vec\gamma) \right] \nonumber\\
			= &  \frac{\sum_{s=1}^{T}\min\{\lambda_s/\rho_T, 1 \}}{T} \cdot \E_{\vec \gamma \sim \PP^T} \left[ \opt(T, \vec\gamma) \right] &\left( \text{Defn. of } c(\vec\lambda, T) \right) \nonumber\\
			= &  \frac{\sum_{s=1}^{\tA}\min\{\lambda_s/\rho_T, 1 \}}{T} \cdot \E_{\vec \gamma \sim \PP^T} \left[ \opt(T, \vec\gamma) \right] + \frac{\sum_{s=\tA + 1}^{T}\min\{\lambda_s/\rho_T, 1 \}}{T} \cdot \E_{\vec \gamma \sim \PP^T} \left[ \opt(T, \vec\gamma) \right] \nonumber\\
			\leq &  \sum_{s=1}^{\tA}\min\{\lambda_s/\rho_T, 1 \} \cdot \frac{\E_{\vec \gamma \sim \PP^T} [ \opt(T, \vec\gamma) ]}{T} + \sum_{s=\tA + 1}^{T}\min\{\lambda_s/\rho_T, 1 \} \cdot \bar f &\left(\opt(T) \leq T \cdot \bar f \right) \nonumber\\
			\leq &   \sum_{s=1}^{\tA}\min\{\lambda_s/\rho_T, 1 \} \cdot \frac{\E_{\vec \gamma \sim \D^T} \left[ D(\bar\mu_{\tA}|T, \vec\gamma) \right]}{T}  +  \sum_{s=\tA + 1}^{T}\min\{\lambda_s/\rho_T, 1 \} \cdot \bar f  &\left(\text{weak duality} \right)  \nonumber\\
			= &  \sum_{t=1}^{\tA}\min\{\lambda_t/\rho_T, 1 \} \cdot \tilde D\left(\bar\mu_{\tA}|  \rho_T, \PP \right)  +  \sum_{s=\tA + 1}^{T}\min\{\lambda_s/\rho_T, 1 \} \cdot \bar f
		\end{align}
	where the last equality follows from part (a) of Lemma~\ref{lemma:dual-prop}.

	\paragraph{Step 4: Putting it all together.} Combining the results from steps 1-3 yields:
	\begin{align*}
		&\E_{\vec\gamma \sim \PP^T} \left[ c(\vec\lambda, T) \cdot \opt(T, \vec\gamma) - R(A|T, \vec\gamma)\right]\\
		\leq\ & \E_{\vec\gamma \sim \PP^T} \left[ c(\vec\lambda, T) \cdot \opt(T, \vec\gamma) - \sum_{t=1}^\tA f_t(x_t) \right]\\
		\leq\ &  \E_{\vec\gamma \sim \PP^T} \left[ c(\vec\lambda, T) \cdot \opt(T, \vec\gamma) - \sum_{t=1}^{\tA}\min\{\lambda_t/\rho_T, 1 \} \cdot \tilde D\left(\bar\mu_{\tA}|  \rho_T, \PP \right) + \sum_{t=1}^{\tA} \mu_t^\top \left( \lambda_t - b_t(x_t) \right)\right] & (\text{Equation \ref{eq:alg-lower-bound}})\\
		\leq\ &  \E_{\vec\gamma \sim \PP^T} \left[ \sum_{s=\tA + 1}^{T}\min\{\lambda_s/\rho_T, 1 \} \cdot \bar f + \sum_{t=1}^{\tA} w_t(\mu_t)\right] & (\text{Equation \ref{eq:opt-upper-bound}})\\
		\leq\ &  \underbrace{\E_{\vec\gamma \sim \PP^T} \left[ \sum_{s=\tA + 1}^{T}\min\{\lambda_s/\rho_T, 1 \} \cdot \bar f + \sum_{t=1}^{\tA} w_t(\mu) + E(T, \mu)\right]}_{\clubsuit} & (\text{Equation \ref{eq:regret-omd}})
	\end{align*}
	for all $\mu \in \R_+^m$. All that remains to complete the proof is choosing the right $\mu$. If $\tA = T$ (no resource was completely depleted), set $\mu = 0$. If $\tA < T$, then there exists a resource $j \in [m]$ that nearly got depleted, i.e., $\sum_{t=1}^{\tA} b_{t,j}(x_t) + \bar{b} \ge B_j$. Moreover, recall that the definition of a target consumption sequence implies $\sum_{t=1}^{T} \lambda_{t,j} \leq B_j$. Thus, $\sum_{t=1}^{\tA} b_{t,j}(x_t)) \geq - \bar b +  \sum_{t=1}^{T} \lambda_{t,j}$. Therefore, setting $\mu = (\bar f/ \underline \rho_T) e_j$, where $e_j \in \R_m$ is the $j$-th unit vector, yields:
	\begin{align*}
		\sum_{t=1}^\tA w_t(\mu) &= \sum_{t=1}^\tA \frac{\bar f}{\underline \rho_T} e_j^\top (\lambda_t - b_t(x_t))\\
		&=  \frac{\bar f}{\underline \rho_T} \cdot \left( \sum_{t=1}^\tA \lambda_{t,j} - \sum_{t=1}^\tA b_{t,j}(x_t) \right)\\
		&\leq \frac{\bar f}{\underline \rho_T}\cdot \left(\bar b - \left\{\sum_{t=1}^{T} \lambda_{t,j} - \sum_{t=1}^\tA \lambda_{t,j} \right\} \right)\\
		&= \frac{\bar f \bar b}{\underline \rho_T} - \bar f \cdot \sum_{t = \tA + 1}^T\frac{\lambda_{t,j}}{\underline \rho_T}\\
		&\leq \frac{\bar f \bar b}{\underline \rho_T} - \bar f \cdot \sum_{t = \tA + 1}^T \min\{\lambda_t/\rho_T, 1\}
	\end{align*}
Here we use that $\frac{\lambda_{t,j}}{\underline \rho_T} \ge \min\{\lambda_t/\rho_T, 1\}$. This follows because $\min_j a_j / \min_j b_j \ge \min_j a_j/b_j$.
	Finally, if we put everything together, we get
	\begin{align*}
		\clubsuit \leq \frac{\bar f \bar b}{\underline \rho_T} + E(T, \mu) \le \frac{\bar f \bar b}{\underline \rho_T} + \frac 1 {2\sigma} (\bar{b} + \bar\lambda)^2 \eta \cdot T + \frac{1}{\eta} V_h(\mu,\mu_1) \ ,	
	\end{align*}
	where we have used the definition of $E(T, \mu)$. The theorem follows from observing that all of our choices of $\mu$ in the above discussion lie in the set $\{0, (\bar f/\underline \rho_T) e_1,\ldots, (\bar f/\underline \rho_T) e_m\}$.
\end{proof}

\subsection{Proof of Proposition~\ref{prop:master}}

\begin{proof}[Proof of Proposition~\ref{prop:master}]
	Setting $\eta = \sqrt{C'_3/\{C_2 \tau_2\}}$ in Theorem~\ref{thm:regret} yields
	\begin{align*}
		\E_{\vec\gamma \sim \PP^T} \left[ c(\vec\lambda, T) \cdot \opt(T, \vec\gamma) - R(A|T, \vec\gamma)\right] &\le C^{(T)}_1 + C_2 T \sqrt{\frac{C_3'}{C_2 \tau_2}} + C^{(T)}_3 \sqrt{\frac{C_2 \tau_2}{C_3'}}\\
		&\leq C'_1 + \sqrt{C_2 C'_3 \tau_2} + \sqrt{C_2 C'_3 \tau_2} \,.
	\end{align*}
	Dividing both sides by $\E[\opt(T,\vec\gamma)]$ and using $\E[\opt(T, \vec\gamma)] \geq \kappa \cdot T$, we get
	\begin{align*}
		& \frac{\E_{\vec\gamma \sim \PP^T} \left[ c(\vec\lambda, T) \cdot \opt(T, \vec\gamma) - R(A|T, \vec\gamma)\right]}{\E_{\vec\gamma \sim \PP^T}[\opt(T,\vec\gamma)]} \leq \frac{C'_1 + 2\cdot \sqrt{C_2 C'_3 \tau_2}}{\E_{\vec\gamma \sim \PP^T}[\opt(T,\vec\gamma)]}\\
		\iff & c(\vec\lambda, T) - \frac{\E_{\vec\gamma \sim \PP^T} \left[R(A|T, \vec\gamma)\right]}{\E_{\vec\gamma \sim \PP^T}[\opt(T,\vec\gamma)]} \leq \frac{C'_1}{\kappa T} + \frac{2\cdot \sqrt{C_2 C'_3 \tau_2}}{\kappa T} \,.
	\end{align*}
	Therefore, rearranging terms and using $T \geq \tau_1$ we get
	\begin{align*}
		\frac{\E_{\vec\gamma \sim \PP^T} \left[R(A|T, \vec\gamma)\right]}{\E_{\vec\gamma \sim \PP^T}[\opt(T,\vec\gamma)]} &\geq c(\vec\lambda, T) - \left( \frac{C'_1}{\kappa T} + 2\cdot\frac{\sqrt{C_2 C'_3 \tau_2}}{\kappa T} \right)\\
		&\geq c(\vec\lambda, T) - \left( \frac{C'_1}{\kappa T} + 2 \cdot\frac{\sqrt{(\tau_2/\tau_1) C_2 C'_3 \tau_1}}{\kappa T} \right)\\
		&\geq c(\vec\lambda, T) - \left( \frac{C'_1}{\kappa \tau_1} + 2 \cdot\frac{\sqrt{(\tau_2/\tau_1) C_2 C'_3 \tau_1}}{\kappa \tau_1} \right)\\
		&\geq c(\vec\lambda, T) - \left( \frac{C'_1}{\kappa \tau_1} + 2 \cdot\frac{\sqrt{(\tau_2/\tau_1) C_2 C'_3}}{\kappa \sqrt{\tau_1}} \right)
	\end{align*}
	as required.
\end{proof}

\section{Proofs and Extensions for Section~\ref{sec:upper-bound}}

To prove Theorem~\ref{thm:upper-bound}, it suffices to prove the stronger statement which holds for online algorithms with prior knowledge of $(r, \PP_r)$ before time $t=1$. Consequently, we assume that online algorithms have this prior knowledge in the remainder of this section. Any algorithm without this knowledge can only do worse.

\subsection{Proof of Lemma~\ref{lemma:r-opt}}

\begin{proof}[Proof of Lemma~\ref{lemma:r-opt}]
	Fix $T \in [\tau_1, \tau_2]$. As the resource consumption function is given by $I(\cdot)$, we get that
	\begin{align*}
		\opt(T,r) = \max_x \sum_{t=1}^T f_r(x_t) \quad \text{subject to } \sum_{t=1}^T x_t \leq B	
	\end{align*}
	Let $x$ be a feasible solution of the above optimization problem. Then,
	\begin{align}\label{eqn:hind-opt}
		\sum_{t=1}^T f_r(x_t) = T \cdot \frac{\sum_{t=1}^T f_r(x_t)}{T} \leq T \cdot f_r\left( \frac{\sum_{t=1}^T x_t}{T} \right) \leq T \cdot f_r\left( \frac{B}{T} \right)
	\end{align}
	where the first inequality follows from the concavity of $f_r$ and the second inequality follows from the resource constraint $\sum_{t=1}^T x_t \leq B$. Hence, we get that $x_t^* = B/T$ for all $t \leq T$ is an optimal solution to the above optimization problem and as a consequence, $\opt(T,r) = T \cdot f_r(B/T) = B^r \cdot T^{1-r}$. Moreover, we have that $x_t^* = B/T$ is the unique optimal solution because $f_r$ is strictly concave and increasing for $r \in (0,1)$, and therefore (i) The first inequality in \eqref{eqn:hind-opt} is strict whenever $x_t \neq x_s$ for some $s,t \in [T]$; (ii) $f_r(\sum_{t=1}^T x_t/T) < f(B/T)$ whenever $\sum_{t=1}^T x_t < B$.
\end{proof}

\subsection{Proof of Lemma~\ref{lemma:alter-char}}

\begin{proof}[Proof of Lemma~\ref{lemma:alter-char}]
	We begin by noting that our use of $\max$ instead of $\sup$ is justified in the right-hand side of the equality in Lemma~\ref{lemma:alter-char} because $f_r^{-1}$ is continuous for all $r \in (0,1)$.  Now, fix $r \in (0,1)$ and $1 \leq \tau_1 \leq \tau_2$. Let
	\begin{align*}
		c \in \argmax \left\{ c'\ \biggr\lvert\ \tau_1 \cdot f_r^{-1} \left(c' \cdot \frac{\opt(\tau_1, r)}{\tau_1} \right) + \sum_{t= \tau_1 + 1}^{\tau_2} f_r^{-1} \left( c' \cdot \Delta \opt(t, r) \right) \leq B  \right\}\ ,	
	\end{align*}
 	and define
 	\begin{align}\label{eq:x-t}
 		x_t \coloneqq \begin{cases}
 			f_r^{-1} \left(c \cdot \frac{\opt(\tau_1, r)}{\tau_1} \right) &\text{if } t \leq \tau_1\\
 			f_r^{-1} \left( c \cdot \Delta \opt(t, r) \right) &\text{if } \tau_1 < t \leq \tau_2
 		\end{cases}	
 	\end{align}
	Then, by definition of $c$, we have $\sum_{t=1}^{\tau_2} I(x_t) = \sum_{t=1}^{\tau_2} x_t \leq B$. Moreover, observe that $\opt(\tau_1, r)/\tau_1 = (B/\tau_1)^r$ and
	\begin{align*}
		\Delta \opt(t,r) = \opt(t,r) - \opt(t-1,r) = B^r \cdot \left( t^{1-r} - (t-1)^{1-r} \right) \leq B^r \cdot \frac{1-r}{(t-1)^r}
	\end{align*}
	for all $t \geq \tau_1 +1$. To see why the second-last inequality holds, note that the Intermediate Value Theorem applied to the function $t \mapsto t^{1-r}$ between the points $t$ and $t-1$ yields the existence of an $s \in [t-1,t]$ such that $t^{1-r} - (t - 1)^{1-r} = (1 -r)/s^r$, which implies $t^{1-r} - (t - 1)^{1-r} \leq (1 -r)/(t-1)^r$.

As a consequence, we get
\begin{align}
 		x_t \leq \begin{cases}
 			c^{1/r} \{B/\tau_1\} &\text{if } t \leq \tau_1\\
 			c^{1/r} \{B/(t-1)\} &\text{if } \tau_1 < t \leq \tau_2
 		\end{cases}	
\end{align}
Combining the above inequalities with the definition of $c$ yields $c \leq 1$. Hence, we have $x_t \in \X$ for all $t \in [\tau_2]$.

Consider the algorithm that selects action $x_t$ at time $t$. Then, in $T \in [\tau_1, \tau_2]$ time steps, it receives a reward of
	\begin{align*}
		\sum_{t=1}^T f_r(x_t)
		 &= \sum_{t=1}^{\tau_1} f_r\left( f_r^{-1} \left(c \cdot \frac{\opt(\tau_1, r)}{\tau_1} \right) \right) + \sum_{t=\tau_1 +1}^T f_r\left( f_r^{-1} \left( c \cdot \Delta \opt(t, r) \right) \right)\\
		 &= c\cdot \opt(\tau_1,r) + \sum_{t = \tau_1 + 1}^T c \cdot \Delta \opt(t,r)\\
		 &= c \cdot \opt(T,r)\ .
	\end{align*}
	Therefore, we have shown that
	\begin{align*}
		\sup_A\min_{T \in [\tau_1, \tau_2]} \frac{R(A|T,r)}{\opt(T,r)} \geq 	\max \left\{ c \in [0,1]\ \biggr\lvert\ \tau_1 \cdot f_r^{-1} \left(c \cdot \frac{\opt(\tau_1, r)}{\tau_1} \right) + \sum_{t= \tau_1 + 1}^{\tau_2} f_r^{-1} \left( c \cdot \Delta \opt(t, r) \right) \leq B  \right\}
	\end{align*}

	Next, we prove the above inequality in the opposite direction. Consider an online algorithm $A$ such that
	\begin{align*}
		\frac{R(A|T,r)}{\opt(T,r)} \geq c \quad \forall\ T \in [\tau_1, \tau_2] ,
	\end{align*}
	for some constant $c > 0$. Let $x(A)_t$ represent the action taken by $A$ at time $t$. Since the action of an online algorithm cannot depend on future information, $x(A)_t$ represents the action taken by algorithm $A$ for all horizons $T \geq t$. Let $\{x(\tilde A)_t\}_t$ represent the sequence obtained by sorting $\{x(A)_t\}_t$ in decreasing order, and let $\tilde A$ represent the algorithm that takes action $x(\tilde A)_t$ at time $t$. Then, we have
	\begin{align*}
		\frac{\sum_{t=1}^T f_r(x(\tilde A)_t)}{\opt(T,r)} \geq \frac{\sum_{t=1}^T f_r(x(A)_t)}{\opt(T,r)} \geq c \quad \forall\ T \in [\tau_1, \tau_2]\ ,	
	\end{align*}
	which allows us to assume that $\{x(A)_t\}_t$ is sorted in decreasing order without loss of generality.

	Since $\sum_{t=1}^{\tau_2} x(A)_t \leq B$, to complete the proof it suffices to show that
	\begin{align*}
		\sum_{t=1}^{\tau_2} x_t = \tau_1 \cdot f_r^{-1} \left(c \cdot \frac{\opt(\tau_1, r)}{\tau_1} \right) + \sum_{t= \tau_1 + 1}^{\tau_2} f_r^{-1} \left( c \cdot \Delta \opt(t, r) \right) \leq \sum_{t=1}^{\tau_2} x(A)_t	
	\end{align*}
	where the equality follows from the definition of $x_t$ \eqref{eq:x-t}. We will prove this via induction by inductively proving the following statement for all $T \in [\tau_1, \tau_2]$:
	\begin{align*}
		\sum_{t=1}^{T} x_t = \tau_1 \cdot f_r^{-1} \left(c \cdot \frac{\opt(\tau_1, r)}{\tau_1} \right) + \sum_{t= \tau_1 + 1}^{T} f_r^{-1} \left( c \cdot \Delta \opt(t, r) \right) \leq \sum_{t=1}^{T} x(A)_t	
	\end{align*}
	
 	To do so, we will maintain variables $\{w(T)_t\}_{t\leq T}$ that we initialize to be 0 and update inductively. At a high level, they capture a water-filling procedure. Suppose there is a container corresponding to each time step $t$ with a capacity of $x(A)_t$. We assume that these containers can be connected to each other so that water always goes to the lowest level, which corresponds to the highest marginal reward since $f_r$ is concave.  Moreover, we will assume that container $T$ becomes available at time $T$ and is connected to containers $t < T$ at that point. Finally, we also freeze the newly-added water at the end of each time step to inductively use the properties of the water level from the previous time step. We would like to caution the reader that this water-filling interpretation is just a tool that guided our intuition, and the mathematical quantities defined below may not match it exactly.
 	
 	\paragraph{Base Case $T= \tau_1$:} Let $\{w(\tau_1)_t\}_{t=1}^{\tau_1}$ be a decreasing sequence that satisfies the following properties:
 	\begin{itemize}
 		\item[I.] $\sum_{t=1}^{\tau_1} f_r(w(\tau_1)_t) = c \cdot \opt(\tau_1,r)$.
 		\item[II.] $w(\tau_1)_t \leq x(A)_t$ for all $t \leq \tau_1$.
 		\item[III'.] $w(\tau_1)_t < w(\tau_1)_1 \implies w(\tau_1)_t = x(A)_t$.
 	\end{itemize}
 	Such a sequence is guaranteed to exist because $\{x(A)_t\}_{t=1}^{\tau_1}$ satisfies properties (II - III') trivially, and (I) can be satisfied as $\sum_{t=1}^{\tau_1} f_r(x(A)_t) \geq c \cdot \opt(\tau_1,r)$ and $f_r$ is a continuous increasing function. If $\{w(\tau_1)_t\}_{t=1}^{\tau_1}$ is a constant sequence, then property (I) implies
 	\begin{align*}
 		w(\tau_1)_t = f_r^{-1} \left(c \cdot \frac{\opt(\tau_1, r)}{\tau_1} \right) = x_{\tau_1} \quad \forall\ t \leq \tau_1
 	\end{align*}
	Suppose $\{w(\tau_1)_t\}_{t=1}^{\tau_1}$ is not a constant sequence. Then, the strict concavity of $f_r$ implies that
	\begin{align*}
		f_r \left( \frac{\sum_{t=1}^{\tau_1} w(\tau_1)_t}{\tau_1} \right) > \frac{\sum_{t=1}^{\tau_1} f_r(w(\tau_1)_t)}{\tau_1} = \frac{c \cdot \opt(\tau_1,r))}{\tau_1},	
	\end{align*}
	which implies
	\begin{align}\label{eqn:upper-bound-base-case}
		w(\tau_1)_1 \geq \frac{\sum_{t=1}^{\tau_1} w(\tau_1)_t}{\tau_1} > 	f_r^{-1} \left(c \cdot \frac{\opt(\tau_1, r)}{\tau_1} \right) = x_{\tau_1}.
	\end{align}
	Therefore, we have established:
	\begin{itemize}
		\item[III.] $w(\tau_1)_t < x_t \implies w(\tau_1)_t = x(A)_t$. 
		\item[IV.] $\sum_{t=1}^{\tau_1} w(\tau_1)_t \geq \sum_{t=1}^{\tau_1} x_t = \tau_1 \cdot x_{\tau_1}$. 
	\end{itemize}
	where (III) follows follows trivially when $\{w(\tau_1)_t\}_{t=1}^{\tau_1}$ is a constant sequence and follows from (III') and $w(\tau_1)_1 > x_t = x_{\tau_1}$ otherwise, and (IV) also follows trivially when $\{w(\tau_1)_t\}_{t=1}^{\tau_1}$ is a constant sequence and follows from \eqref{eqn:upper-bound-base-case} otherwise.
	
	\paragraph{Induction Hypothesis $\tau_1 \leq T < \tau_2$:} Suppose there exists a decreasing sequence $\{w(T)_t\}_{t=1}^{T}$ that satisfies the following properties:
 	\begin{itemize}
 		\item[I.] $\sum_{t=1}^T f_r\left( w(T)_t \right) = c \cdot \opt(T,r)$.
 		\item[II.] $w(T)_t \leq x(A)_t$ for all $t \leq T$.
 		\item[III.] $w(T)_t < x_t \implies w(T)_t = x(A)_t$.
 		\item[IV.] $\sum_{t=1}^T w(T)_t \geq \sum_{t=1}^T x_t$.
 	\end{itemize} 
 	
 	\paragraph{Induction Step $T+1$:} If $x(A)_{T+1} \geq x_{T+1}$, then set $w(T+1)_{T+1} = x_{T+1}$ and $w(T+1)_t = w(T)_t$ for all $t \leq T$. In this case, it is easy to see that conditions (I-IV) hold for $\{w(T+1)_t\}_t$. Next, assume $x(A)_{T+1} < x_{T+1}$. In this case, set $w(T+1)_{T+1} = x(A)_{T+1}$. Moreover, let $\{w(T+1)_t\}_{t=1}^{T}$ be a decreasing sequence that satisfies the following properties:
 	\begin{itemize}
 		\item[I.] $\sum_{t=1}^T f_r(w(T+1)_t) = c \cdot \opt(T,r) + f_r(x_{T+1}) - f_r(x(A)_{T+1})$.
 		\item[II.] $w(T+1)_t \leq x(A)_t$ for all $t \leq T$.
 		\item[III'.] $w(T)_t \leq w(T+1)_t$ for all $t \leq T$.
 	\end{itemize}
 	Such a sequence is guaranteed to exist because $\{x(A)_t\}_{t=1}^T$ satisfies property (II) trivially, (III') as a consequence of the inductive hypothesis, and (I) can be satisfied because $f_r$ is a continuous increasing function and
 	\begin{align*}
 		\sum_{t=1}^{T+1} f_r(x(A)_t) \geq c \cdot \opt(T+1, r) &\iff 	\sum_{t=1}^{T+1} f_r(x(A)_t) \geq c \cdot \opt(T, r) + c \cdot \Delta \opt(T+1,r)\\
 		&\iff \sum_{t=1}^{T} f_r(x(A)_t) \geq c \cdot \opt(T, r) + f_r(x_{T+1}) - f_r(x(A)_{T+1}).
 	\end{align*}
	Observe that (III') and $w(T+1)_{T+1} = x(A)_{T+1}$ implies
	\begin{itemize}
		\item[III.] $w(T+1)_t < x_t \implies w(T+1)_t = x(A)_t$
	\end{itemize}

	Now, only (IV) remains. First, note that the Intermediate Value Theorem applied to $t \mapsto t^{1-r}$ implies
	\begin{align*}
		\Delta \opt(T+1 , r) = 	B^r \cdot [(T+1)^{1-r} - T^{1-r}] \leq B^r \cdot \frac{1-r}{T^{1-r}} \leq B^r \cdot [T^{1-r} - (T-1)^{1-r}] = \Delta \opt(T , r),
	\end{align*}
	and as a consequence, we get $x_{T+1} \leq x_T$. This further implies
	\begin{align*}
		\left[ \sum_{t=1}^T w(T+1)_t - \sum_{t=1}^T w(T)_t \right] f_r'(x_{T+1}) &\geq \left[ \sum_{t=1}^T w(T+1)_t - \sum_{t=1}^T w(T)_t \right] f_r'(x_{T})\\
		 &\geq \sum_{t=1}^T\left[ w(T+1)_t - w(T)_t \right] f_r'(w(T)_t)\\
		&\geq \sum_{t=1}^T\left[ f_r(w(T+1)_t) - f_r(w(T)_t) \right]\\
		&= f(x_{T+1}) - f(x(A)_{T+1})\\
		&\geq \left[ x_{T+1} - x(A)_{T+1} \right] f_r'(x_{T+1})  
	\end{align*}
	where the first inequality follows from the concavity of $f_r$ and the fact that $x_{T+1} \leq x_T$; the second inequality follows from concavity of $f_r$ and the observation that the induction hypothesis and (III') imply $w(T)_t = w(T+1)_t = x(A)_t$ whenever $w(T)_t < x_T \leq x_t$, i.e., $w(T+1)_t - w(T)_t > 0$ implies $w(T)_t \geq x_T$; the third and the fourth inequalities follow from the Intermediate Value Theorem applied to $f_r$; and the equality follows from (I). Therefore, we get (IV) by using the inductive hypothesis and $w(T+1)_{T+1} = x(A)_{T+1}$:
	\begin{align*}
		\sum_{t=1}^{T+1} w(T+1) \geq \sum_{t=1}^T w(T)_t + x_{T+1} - x(A)_{T+1} + w(T+1)_{T+1} \geq \sum_{t=1}^{T+1} x_t	
	\end{align*}
	This concludes the induction step and thereby the proof, since (II) and (IV) together imply $\sum_{t=1}^T x_t \leq \sum_{t=1}^T x(A)_t$.
\end{proof}


\subsection{Proof of Theorem~\ref{thm:upper-bound}}

\begin{proof}[Proof of Theorem~\ref{thm:upper-bound}]
	Combining Lemma~\ref{lemma:alter-char} and Lemma~\ref{lemma:r-opt} yields
	\begin{align}\label{eq-simplify-start}
		\tau_1 \cdot \left(c^* \cdot \frac{B^r \tau_1^{1-r}}{\tau_1} \right)^{1/r} + \sum_{t= \tau_1 + 1}^{\tau_2} \left( c^* \cdot [B^r t^{1-r} - B^r (t-1)^{1-r}] \right)^{1/r} \leq B\ ,	
	\end{align}
	where $c^* = \sup_A\min_{T \in [\tau_1, \tau_2]} R(A|T,r)/\opt(T,r)$. First, note that the Intermediate Value Theorem applied to $t \mapsto t^{1-r}$ implies
	\begin{align*}
		t^{1-r} - (t-1)^{1-r} \geq \frac{1-r}{t^r} \quad \forall\ T \geq \tau_1+1,
	\end{align*}
	which allows us to derive the following inequality from \eqref{eq-simplify-start}:
	\begin{align*}
		&(c^*)^{1/r} + \sum_{t = \tau_1 + 1}^{\tau_2} \left(c^* \cdot \frac{1-r}{t^r} \right)^{1/r} \leq 1\\
		\iff & (c^*)^{1/r} \leq \frac{1}{1 + (1 - r)^{1/r} \sum_{t = \tau_1 +1}^{\tau_2} 1/t }
	\end{align*}
	Using $\sum_{t=\tau_1 +1}^{\tau_2} 1/t \geq \ln(\tau_2/(\tau_1+1)) = \ln(\tau_2/\tau_1) + \ln(\tau_1/ (\tau_1+1))$ and $(1 - r)^{1/r} < 1$, we get the first half of Theorem~\ref{thm:upper-bound}:
	\begin{align*}
	    c^* \leq \frac{1}{\left(1 + (1 - r)^{1/r} \cdot \ln(\tau_2/\tau_1) + (1 - r)^{1/r} \ln\left( \frac{\tau_1}{\tau_1 + 1} \right)  \right)^r} \leq \frac{1}{\left(1 + (1 - r)^{1/r} \cdot \ln(\tau_2/\tau_1) + \ln\left( \frac{\tau_1}{\tau_1 + 1} \right)  \right)^r}.
	\end{align*}
	
	To get the second half, note that $1 + \ln(\tau_1/(\tau_1 + 1)) \geq 0$, which allows us to write:
	\begin{align*}
		c^* \leq \frac{1}{(1 -r) \ln(\tau_2/\tau_1)^r}
	\end{align*}
	
	Finally, define $g: (0,1) \to \R$ as $g(r) = (1 - r) \ln(\tau_2/\tau_1)^r$. Then,
	\begin{align*}
		&\ln(g(r)) = \ln(1 - r) + r \cdot \ln\ln(\tau_2/\tau_1)\\
		\implies & \frac{g'(r)}{g(r)} = - \frac{1}{1 - r} + \ln\ln(\tau_2/\tau_1)\\
		\implies & g'(r) =  -\ln(\tau_2/\tau_1)^r + \ln\ln(\tau_2/\tau_1) g(r)\\
		\implies & g''(r) = - \ln\ln(\tau_2/\tau_1) \ln(\tau_2/\tau_1)^r + \ln\ln(\tau_2/\tau_1) f(r) = - \ln\ln(\tau_2/\tau_1) \cdot r \cdot \ln(\tau_2/\tau_1)^r
	\end{align*}

Hence, for $\tau_2/\tau_1 > e^e$, $g$ is concave and is maximized at $r = 1 - (1 / \ln\ln(\tau_2/\tau_1))$. Plugging in $r = 1 - (1 / \ln\ln(\tau_2/\tau_1))$ yields
\begin{align*}
	c^* \leq \frac{\ln\ln(\tau_2/\tau_1)}{\ln(\tau_2/\tau_1)^{1 - \frac{1}{\ln\ln(\tau_2/\tau_1)}}} = \frac{e \cdot \ln\ln(\tau_2/\tau_1)}{\ln(\tau_2/\tau_1)}
\end{align*}
which completes the proof.
\end{proof}

\subsection{Randomized Upper Bound with Linear Rewards and Consumptions}\label{appendix:randomized-linear-upper-bound}

The upper bound of Theorem~\ref{thm:upper-bound} can be extended to the popular special case of linear rewards and linear consumption. Fix $r \in (0,1)$ and $B = \tau_1$. We show that there exists a request distribution, with linear rewards and random linear resource consumption functions, that behaves like $(f_r, I, \X)$ in expectation. Define a consumption rate CDF $H$ as
\begin{align*}
  H(s) = \begin{cases}
   0 & \text{if } s \leq 0\\
   s^{\frac{r}{1-r}} & \text{if } 0 \leq s\leq 1\\
   1 & \text{if } s \geq 1
  \end{cases}\,.
\end{align*}

Consider the following request distribution: the reward of every request is $r^r$, the linear resource-consumption rate $s$ is drawn randomly from $H$, and the action set is $[0,1]$ (to represent the fraction of the request accepted), i.e., for request $(r^r,s,[0,1])$ and action $y \in [0,1]$, the reward is $r^r \cdot y$ and the amount of resource consumed is $s \cdot y$. It is relatively straightforward to see that the optimal action at each time-step is to set a threshold $x$ and accept a request $(r^r,s,[0,1])$ (set $y = 1$) if and only if the consumption rate $s$ is less than or equal to the threshold $x$. For such a threshold $x$, the expected cost is given by
\begin{align*}
  \E_{s \sim H}[s \cdot \mathds{1}(s \leq x)] = x H(x) - \int_0^x H(s) ds = x^{\frac{1}{1-r}}  - (1-r) \cdot x^{\frac{1}{1-r}} = r \cdot x^{\frac{1}{1-r}}
\end{align*}
and the expected reward is given by $r^r \cdot H(x) = r^r \cdot x^{\frac{r}{1-r}}$. Therefore, the expected reward is equal to the expected cost raised to the power $r$, and consequently this request distribution behaves like $(f_r, I, \X)$ in expectation.

\subsection{Upper Bound with Distributional Knowledge of Horizon}\label{appendix:dist-upper-bound}

In this appendix, we show that the upper bound of Theorem~\ref{thm:upper-bound} holds even when the horizon $T$ drawn from a distribution $\T$ supported on $[\tau_1, \tau_2]$ and this distribution is known to the decision maker. This is because the reward functions $f_r(\cdot)$ are concave and $I(\cdot)$ is linear, which makes the problem of maximizing the competitive ratio a convex problem that satisfies strong duality. As we note in the following theorem, the dual variables give rise to a distribution $\T_r$ of the horizon which leads to the same expected performance ratio as the competitive ratio.

\begin{proposition}\label{prop:dist-upper-bound}
  For every $r \in (0,1)$, there exists a distribution $\T_r$ of the horizon $T$ such that
  \begin{align*}
    \sup_A \E_{T \sim \T_r} \left[ \frac{R(A|T,r)}{\opt(T,r)} \right] = \sup_A \min_{T \in [\tau_1, \tau_2]} \frac{R(A|T,r)}{\opt(T,r)} \,,
  \end{align*}
  where $\sup_A$ denotes the supremum over all online algorithms $A$ with the knowledge of the horizon distribution $\T_r$ and the request distribution $\PP_r$.
\end{proposition}

\begin{proof}[Proof of Proposition~\ref{prop:dist-upper-bound}]
	Fix $r \in (0,1)$. We begin by showing that for all $r \in (0,1)$, we have
	\begin{align}\label{eqn:dist-upper-bound-1}
		\sup_A \min_{T \in [\tau_1, \tau_2]} \frac{R(A|T,r)}{\opt(T,r)} \quad = \quad \max_{z, x_t} \quad &z \nonumber\\
		\text{s.t.} \quad &z \leq \frac{\sum_{t=1}^T f_r(x_t)}{\opt(T,r)} &\forall T \in [\tau_1, \tau_2]\\
		&\sum_{t=1}^{\tau_2} x_t \leq B \nonumber\\
		&x_t \in \X &\forall t \in [\tau_2] \nonumber
	\end{align}
	The RHS is a convex program because $f_r$ is concave. Let $A$ be any online algorithm. Let $x(A)_t$ denote the action taken by $A$ at time $t$. Then, $x_t = x(A)_t$ and
	\begin{align*}
		z = \min_{T \in [\tau_1, \tau_2]} 	\frac{\sum_{t=1}^T f_r(x(A)_t)}{\opt(T,r)} = \min_{T \in [\tau_1, \tau_2]}  \frac{R(A|T,r)}{\opt(T,r)}\,
	\end{align*}
	is a feasible solution of the convex program, thereby establishing the `$\leq$' direction. The other direction is equally straightforward: Any feasible solution $(x_t, z)$ of the convex program naturally gives rise to an algorithm which takes action $x_t$ at time $t$ and achieves the desired competitive ratio.
	
	Note that $x_t = 1/2$ and $z = 0$ is a feasible solution of the convex program that satisfies
	\begin{align*}
		z < \frac{\sum_{t=1}^T f_r(x_t)}{\opt(T,r)} \quad \forall T \in [\tau_1, \tau_2]
	\end{align*}
	Therefore, by Slater's condition (see \citealt{bertsekas1998nonlinear}), we get that strong duality holds and there exists an optimal dual multiplier $\{p_T\}_{T=\tau_1}^{\tau_2}$ associated with the constraints in \eqref{eqn:dist-upper-bound-1} such that $p_T \geq 0$ for all $T \in [\tau_1, \tau_2]$ and
	\begin{align*}
		\sup_A \min_{T \in [\tau_1, \tau_2]} \frac{R(A|T,r)}{\opt(T,r)} \quad = \quad \max \quad & z + \sum_{T=\tau_1}^{\tau_2} p_T \left( \frac{\sum_{t=1}^T f_r(x_t)}{\opt(T,r)} - z\right) \\
		\text{s.t.} \quad &\sum_{t=1}^{\tau_2} x_t \leq B\\
		&x_t \in \X &\forall t \in [\tau_2]\\
		= \quad \max \quad & z\left( 1 - \sum_{T=\tau_1}^{\tau_2} p_T \right) + \sum_{T=\tau_1}^{\tau_2} p_T \cdot\frac{\sum_{t=1}^T f_r(x_t)}{\opt(T,r)} \\
		\text{s.t.} \quad &\sum_{t=1}^{\tau_2} x_t \leq B\\
		&x_t \in \X &\forall t \in [\tau_2]
	\end{align*}
	Observe that, since $z$ is an unrestricted variable, we need $\sum_{T=\tau_1}^{\tau_2} p_T =1$ to ensure that the restated (Lagrangian) optimization problem is bounded, which is necessary because the LHS is bounded. Hence, we get
	\begin{align*}
		\sup_A \min_{T \in [\tau_1, \tau_2]} \frac{R(A|T,r)}{\opt(T,r)} \quad
		= \quad \max \quad & \sum_{T=\tau_1}^{\tau_2} p_T \cdot\frac{\sum_{t=1}^T f_r(x_t)}{\opt(T,r)} \\
		\text{s.t.} \quad &\sum_{t=1}^{\tau_2} x_t \leq B\\
		&x_t \in \X &\forall t \in [\tau_2]
	\end{align*}
	Let $\T_r$ be the distribution which picks horizon $T$ with probability $p_T$ for all $T \in [\tau_1, \tau_2]$. Then, once again, we can use the equivalence between feasible solutions of the convex program and the actions of an online algorithm ($x_t = x(A)_t$) to get
	\begin{align*}
		\sup_A \min_{T \in [\tau_1, \tau_2]} \frac{R(A|T,r)}{\opt(T,r)} \quad
		= \quad \max \quad & \sum_{T=\tau_1}^{\tau_2} p_T \cdot\frac{\sum_{t=1}^T f_r(x_t)}{\opt(T,r)} \quad = \quad \sup_A \E_{T \sim \T_r} \left[ \frac{R(A|T,r)}{\opt(T,r)} \right]  \\
		\text{s.t.} \quad &\sum_{t=1}^{\tau_2} x_t \leq B\\
		&x_t \in \X \quad \forall t \in [\tau_2]
	\end{align*}
	as required.
\end{proof}

Combining Proposition~\ref{prop:dist-upper-bound} and Theorem~\ref{thm:upper-bound} immediately yields the following corollary.

\begin{corollary}\label{cor:dist-upper-bound}
    For all $r \in (0,1)$ and $1 \leq \tau_1 \leq \tau_2$, there exists a horizon distribution $\T_r$ such that every online algorithm $A$ with knowledge of $\T_r$ satisfies
    \begin{align*}
        \E_{T \sim \T_r} \left[ \frac{R(A|T,r)}{\opt(T,r)} \right] \leq \frac{1}{\left(1 + (1 - r)^{1/r} \cdot \ln(\tau_2/\tau_1) + \ln\left( \frac{\tau_1}{\tau_1 + 1} \right) \right)^r}\,.
    \end{align*}
  In particular, for $r = 1 - \{1/\ln\ln(\tau_2/\tau_1)\}$, $\tau_2/\tau_1 > e^e$ and $\tau_1 \geq 1$, we have
  \begin{align*}
    \E_{T \sim \T_r} \left[ \frac{R(A|T,r)}{\opt(T,r)} \right] \leq  \frac{e\cdot \ln\ln(\tau_2/\tau_1)}{\ln(\tau_2/\tau_1)}\,.
  \end{align*}
  The above bounds hold even for online algorithms that have prior knowledge of $\PP_r$ before time $t=1$.
\end{corollary}

\section{Proof of Proposition~\ref{prop:uncertainty-window-LP}}

\begin{proof}
	Consider a target consumption sequence $\lambda$. Set
	\begin{align*}
		y_{T,t} = \min \left\{ \min_{1 \leq j \leq m} \frac{\lambda_{t,j}}{\rho_{T,j}}, 1 \right\} \quad \text{ and } \quad z = \min_{T \in [\tau_1, \tau_2]} \frac{1}{T} \cdot \sum_{t=1}^T \min \left\{ \min_{1 \leq j \leq m} \frac{\lambda_{t,j}}{\rho_{T,j}}, 1 \right\}\,.
	\end{align*}
	Then, $(\lambda, z, y)$ is a feasible solution of the LP with objective value $\min_{T \in [\tau_1, \tau_2]} c(\vec\lambda, T)$. Hence, we get LHS  $\leq$ RHS. To prove LHS $\geq$ RHS, consider any feasible solution $(\lambda, z, y)$. Then, $\lambda$ is a target consumption sequence and
	\begin{align*}
		z \leq \min_{T \in [\tau_1, \tau_2]} \frac{1}{T} \cdot \sum_{t=1}^T y_{T,t} \leq \min_{T \in [\tau_1, \tau_2]} \frac{1}{T} \cdot \sum_{t=1}^T \min \left\{ \min_{1 \leq j \leq m} \frac{\lambda_{t,j}}{\rho_{T,j}}, 1 \right\}	 = \min_{T \in [\tau_1, \tau_2]} c(\vec\lambda, T)
	\end{align*}
	where the first inequality follows constraints
	\begin{align*}
		&z \leq \frac{1}{T} \sum_{t=1}^T y_{T,t} & \forall T \in [\tau_1, \tau_2] \,,
	\end{align*}
	 and the second inequality follows from constraints
	 \begin{align*}
	 	&y_{T,t} \leq \frac{\lambda_{t,j}}{\rho_{T,j}} &\forall j \in [m], T \in [\tau_1, \tau_2], t \in [T]\\
		&y_{T,t} \leq 1 &\forall T \in [\tau_1, \tau_2], t \in [T] \,.
	 \end{align*}
	 Therefore, we have LHS $\geq$ RHS, as required.
\end{proof}

\section{Proofs of Section~\ref{sec:advice}}

\subsection{Proof of Proposition~\ref{prop:advice-LP}}

\begin{proof}
	Consider a target consumption sequence $\lambda$ such that $\min_{T \in [\tau_1, \tau_2]} c(\vec\lambda, T) \geq \gamma$. Set
	\begin{align*}
		y_{T,t} = \min \left\{ \min_{1 \leq j \leq m} \frac{\lambda_{t,j}}{\rho_{T,j}}, 1 \right\}\,.
	\end{align*}
	Then, $(\lambda, y)$ is a feasible solution of the LP with objective value
	\begin{align*}
		\frac{1}{T_P} \sum_{t=1}^{T_P} y_{T_P, t} = \frac{1}{T_P} \sum_{t=1}^{T_P} \min \left\{ \min_{1 \leq j \leq m} \frac{\lambda_{t,j}}{\rho_{T_P,j}}, 1 \right\} = c(\vec\lambda, T_P)\,.
	\end{align*}
	Hence, we get LHS  $\leq$ RHS. 
	
	To prove LHS $\geq$ RHS, consider any feasible solution $(\lambda, y)$. Then, $\lambda$ is a target consumption sequence and we have
	\begin{align*}
		y_{T,t} \leq \min \left\{ \min_{1 \leq j \leq m} \frac{\lambda_{t,j}}{\rho_{T,j}}, 1 \right\}\,,
	\end{align*}
	for all $T \in [\tau_1, \tau_2], t \in [T]$, where the inequality follows from constraints
	 \begin{align*}
	 	&y_{T,t} \leq \frac{\lambda_{t,j}}{\rho_{T,j}} &\forall j \in [m], T \in [\tau_1, \tau_2], t \in [T]\\
		&y_{T,t} \leq 1 &\forall T \in [\tau_1, \tau_2], t \in [T]\,.
	 \end{align*}
	 Therefore, we get
	 \begin{align*}
	 	c(\vec\lambda, T_P) = \frac{1}{T_P} \sum_{t=1}^{T_P} \min \left\{ \min_{1 \leq j \leq m} \frac{\lambda_{t,j}}{\rho_{T_P,j}}, 1 \right\} = c(\vec\lambda, T_P) = \frac{1}{T_P} \sum_{t=1}^{T_P} y_{T_P, t}
	 \end{align*}
	 and
	 \begin{align*}
	 	c(\vec\lambda, T) = \frac{1}{T} \sum_{t=1}^{T} \min \left\{ \min_{1 \leq j \leq m} \frac{\lambda_{t,j}}{\rho_{T,j}}, 1 \right\} = c(\vec\lambda, T_P) \geq \frac{1}{T} \sum_{t=1}^{T} y_{T, t} \geq \gamma \,.
	 \end{align*}
	 Consequently, we have $\min_{T \in [\tau_1, \tau_2]} c(\vec\lambda, T) \geq \gamma$ and the objective of the LP is at most $c(\vec\lambda, T_P)$. Hence, LHS $\geq$ RHS as required.
\end{proof}

\subsection{Proof of Proposition~\ref{prop:simple-advice-seq}}

\begin{proof}[Proof of Proposition~\ref{prop:simple-advice-seq}]
	It is straightforward to see that $\vec\lambda$ satisfies the budget constraint:
\begin{align*}
  \sum_{t=1}^T \lambda_t &= \sum_{t=1}^{\tau_1} \frac{\alpha}{1 + \ln(\tau_2/\tau_1)} \cdot \frac{B}{\tau_1} + \sum_{t=\tau_1 + 1}^{\tau_2} \frac{\alpha}{1 + \ln(\tau_2/\tau_1)} \cdot \frac{B}{t} + \sum_{t=1}^{T_P} (1 - \alpha) \cdot \frac{B}{T_P}\\
  &= \frac{\alpha B}{1 + \ln(\tau_2/\tau_1)} \cdot \left(\frac{\tau_1}{\tau_1} + \sum_{t=\tau_1 + 1}^{\tau_2} \frac{1}{t} \right) + (1 - \alpha) \cdot B\\
  &\leq \frac{\alpha B}{1 + \ln(\tau_2/\tau_1)} \cdot (1 + \ln(\tau_2/\tau_1)) + (1 - \alpha) \cdot B\\
  &= B\,.
\end{align*}

Moreover, note that for any $T \in [\tau_1, \tau_2]$, we have
\begin{align*}
  c(\vec\lambda, T) = \frac{1}{T}\sum_{t=1}^T \min \left\{ \min_{1 \leq j \leq m} \frac{\lambda_{t,j}}{\rho_{T,j}}, 1 \right\} \geq  \frac{1}{T}\sum_{t=1}^T \frac{\alpha}{1 + \ln(\tau_2/\tau_1)} = \frac{\alpha}{1 + \ln(\tau_2/\tau_1)} \,,
\end{align*}
and
\begin{align*}
  c(\vec\lambda, T_P) = \frac{1}{T_P}\sum_{t=1}^{T_P} \min \left\{ \min_{1 \leq j \leq m} \frac{\lambda_{t,j}}{\rho_{T_P,j}}, 1 \right\} \geq \frac{1}{T_P}\sum_{t=1}^{T_P} 1 - \alpha + \frac{\alpha}{1 + \ln(\tau_2/\tau_1)} = 1 - \alpha + \frac{\alpha}{1 + \ln(\tau_2/\tau_1)} \,,
\end{align*}
where we have used the fact that $\rho_t \geq \rho_T$ for all $t \leq T$. Hence, we have shown that Algorithm~\ref{alg:dual-descent} with target sequence $\vec\lambda$ is $(\gamma - \epsilon)$-competitive, where $\gamma = \alpha \cdot (1 + \ln(\tau_2/\tau_1))^{-1}$, and $(1 - \alpha)$-consistent on prediction $T_P$. Since $\vec\lambda$ is one possible choice of the target consumption sequence, the proposition holds.
\end{proof}

\section{Proofs for Section~\ref{sec:faster-alg}}\label{appendix:faster-alg}

\subsection{Proof of Theorem~\ref{thm:faster-alg}}

\begin{proof}[Proof of Theorem~\ref{thm:faster-alg}]
	To simplify exposition, we define
	\begin{align*}
		a_T = \begin{cases}
			\beta &\text{if } T = T_P\\
			\gamma &\text{if } T \neq T_P
		\end{cases}
	\end{align*}
	Hence,
	\begin{align*}
		c(\vec\lambda', T_P) \geq \beta \quad \text{ and  } \min_{T \in [\tau_1, \tau_2]} c(\vec\lambda', T) \geq \gamma \quad \iff c(\vec\lambda', T) \geq a_T \quad \forall\ T \in [\tau_1, \tau_2]\,.
	\end{align*}
	
	We start by proving some important properties of Algorithm~\ref{alg:opt-target-sequence}. First, we show that
	\begin{align*}
	\frac{\lambda_{t,j}}{\rho_{T,j}} = \frac{\lambda_{t,k}}{\rho_{T,k}} \qquad \forall \quad j,k \in [m], t \in [T]
	\end{align*} 
	throughout the run of the algorithm. We do so via induction on each update of $\vec\lambda$ (see \eqref{eqn:opt-seq-update}). Initially, $\vec\lambda = 0$ so the statement holds trivially. Suppose it holds before the update. Observe that
	\begin{align*}
		\min\left\{\rho_{T,j} - \lambda_{t,j}, a_T \cdot B_j - \sum_{s=1}^T \lambda_{s,j} \right\} = \rho_{T,j} \min\left\{1 - \frac{\lambda_{t,j}}{\rho_{T,j}}, a_T \cdot T - \sum_{s=1}^T \frac{\lambda_{s,j}}{\rho_{T,j}} \right\} \,.
	\end{align*}
	Let $\vec\lambda'$ be the sequence after the update in \eqref{eqn:opt-seq-update}. Then, we have
	\begin{align*}
		\frac{\lambda'_{t,j}}{\rho_{T,j}} = \frac{\lambda_{t,j}}{\rho_{T,j}} +  \min\left\{1 - \frac{\lambda_{t,j}}{\rho_{T,j}}, a_T \cdot T - \sum_{s=1}^T \frac{\lambda_{s,j}}{\rho_{T,j}} \right\} \,.
	\end{align*}
	Since the RHS is the same for all $j \in [m]$ by the induction hypothesis, the statement holds after the update, thereby completing the induction step. Therefore, we have
	\begin{align}\label{eqn:inter-1-opt-seq}
		c(\vec\lambda^*, T) =  \frac{1}{T}\sum_{t=1}^T \min \left\{ \min_{1 \leq j \leq m} \frac{\lambda_{t,j}^*}{\rho_{T,j}}, 1 \right\} = \min_{1 \leq j \leq m} \frac{1}{T}\sum_{t=1}^T \min \left\{\frac{\lambda_{t,j}^*}{\rho_{T,j}}, 1 \right\} = \min_{1 \leq j \leq m} \frac{1}{B_j}\sum_{t=1}^T \min \left\{\lambda_{t,j}^*, \rho_{T,j} \right\}\,.
	\end{align}

	We prove an intermediate lemma that will prove useful later
	\begin{lemma}\label{lemma:inter-faster-alg}
		For each outer \textbf{For} loop counter $T$, the inner \textbf{For} loop always maintains $\lambda_{t,j} \leq \rho_{T,j}$ and one of the following holds at its termination:
		\begin{itemize}
		\item $\sum_{s=1}^T\lambda_{s,j} = a_T \cdot B_j$.
		\item $\sum_{s=1}^T\lambda_{s,j} \geq a_T \cdot B_j$ held before the first iteration and $\lambda_{t,j}$ was not modified during any of the iterations of the inner \textbf{For} loop.
	\end{itemize}
	\end{lemma}
	
	\begin{proof}
		This is because, for each resource $j \in [m]$, exactly one of the following cases holds after each iteration of the inner \textbf{For} loop in which $\lambda_{t,j}$ was modified:
	\begin{itemize}
		\item $\lambda_{t,j} = \rho_{T,j}$ and $\sum_{s=1}^T\lambda_{s,j} < a_T \cdot B_j$
		\item $\sum_{s=1}^T\lambda_{s,j} = a_T \cdot B_j$
	\end{itemize}
	Now, suppose $\sum_{s=1}^T\lambda_{s,j} < a_T \cdot B_j$ at termination of the inner \textbf{For} loop. Then, $\lambda_{t,j} = \rho_{T,j}$ for all $t \in [T]$ and $a_T \cdot B_j - \sum_{s=1}^T\lambda_{s,j} \leq B_j - T \cdot \rho_{T,j} \leq 0$, which contradicts $\sum_{s=1}^T\lambda_{s,j} < a_T \cdot B_j$. Hence, the lemma holds.
	\end{proof}
	In both cases, at termination we have
	\begin{align}\label{eqn:inter-2-opt-seq}
		\sum_{s=1}^T\lambda_{s,j} \geq a_T \cdot B_j \quad \forall j \in [m] \,.
	\end{align}

	As we only ever increase $\vec\lambda$ in \eqref{eqn:opt-seq-update}, we get
	\begin{align*}
		\sum_{t=1}^{T} \min\{\lambda^*_{t,j},\rho_{T,j}\} \geq a_T \cdot B_j  \quad \forall\ T \in [\tau_1, \tau_2]
	\end{align*}
	for all $j \in [m]$. Therefore, $c(\vec\lambda^*, T) \geq a_T$ for all $T \in [\tau_1,\tau_2]$ by \eqref{eqn:inter-1-opt-seq}. Part (1) of the theorem follows as a direct consequence and we focus on part (2) in the remainder

	We are now ready to prove the theorem. We begin with the ``only if" direction. Suppose $\sum_{t=1}^{\tau_2} \lambda^*_t \leq B$. Since we have shown that $c(\vec\lambda^*, T) \geq a_T$ for all $T \in [\tau_1,\tau_2]$, this makes $\vec\lambda' = \vec\lambda^*$ the required target consumption sequence.
	
	For the other direction, assume that there exists a target consumption sequence $\vec\lambda^o$ (with $\sum_{t=1}^{\tau_2} \lambda^o_t \leq B$) which satisfies $c(\vec\lambda^o, T)\geq a_T$ for all $T \in [\tau_1,\tau_2]$. Let $\vec\lambda'$ be the target consumption sequence which minimizes $\sum_{k=1}^m \sum_{t=1}^T t \cdot \lambda^o_{t,j}$ among all such sequences, i.e.,
	\begin{align*}
		\vec\lambda' \in \argmin_{\vec\lambda^o} \quad  &\sum_{k=1}^m \sum_{t=1}^T t \cdot \lambda^o_{t,j}\\
		\text{s.t.} \quad &c(\vec\lambda^o, T) \geq a_T &\forall\ T \in [\tau_1,\tau_2]\\
		&\sum_{t=1}^{\tau_2} \lambda^o_t \leq B
	\end{align*}
	
	By \eqref{eqn:comp-ratio-restate}, we get
	\begin{align*}
		\sum_{t=1}^{T} \min\{\lambda'_{t,j},\rho_{T,j}\} \geq a_T \cdot B_j \quad \forall\ T \in [\tau_1, \tau_2], j\in [m] \,.
	\end{align*}
	To prove $\sum_{t=1}^{\tau_2} \lambda^*_t \leq B$, it suffices to show that $ \lambda^*_{t,j} \leq \lambda'_{t,j}$ for all $t\in [t],j \in [m]$. For contradiction, suppose the latter does not hold. In what follows, we will use $\vec\lambda^{(T^*)}$ to denote the value of $\vec\lambda$ at the end of the $T$-th iteration of the outer \textbf{For} loop.
	
	Let $T = T^*$ and $t = t^*$ be the outer and inner \textbf{For} loop counters respectively for the update \eqref{eqn:opt-seq-update} at the end of which $\lambda_{t^*,k} > \lambda'_{t^*,k}$ for some resource $k \in [m]$ for the first time during the run of Algorithm~\ref{alg:opt-target-sequence}. Since $c(\vec\lambda', T^*) \geq a_T$, \eqref{eqn:comp-ratio-restate} implies that $\sum_{t=1}^{T^*} \min \{\lambda_{t,k}', \rho_{T,k} \} \geq a_T \cdot B_k$. Since $\lambda_{t^*, k}$ had to be modified to get $\lambda_{t^*,k} > \lambda'_{t^*,k}$ for the first time, Lemma~\ref{lemma:inter-faster-alg} implies that the inner \textbf{For} loop will terminate with $\sum_{t=1}^{T^*} \min \left\{\lambda^{(T^*)}_{t,k}, \rho_{T,k} \right\} = a_T \cdot B_k$. Therefore, there must exist a $t^* < s \leq T^*$ such that $\lambda'_{s,k} > \lambda^{(T^*)}_{s,k}$ after the $T^*$-th iteration of the outer \textbf{For} loop.

	Now, pick $\nu < \min\left\{ \lambda'_{s,k} -\lambda^{(T^*)}_{s,k}, \lambda^{(T^*)}_{t^*,k} - \lambda_{t^*, k}'\right\}$ and define a new target consumption sequence $\vec\lambda''$ which is exactly the same as $\vec\lambda'$ except $\lambda''_{t^*,k} = \lambda'_{t^*,k} + \nu$ and $\lambda''_{s,k} = \lambda'_{s,k} - \nu$. Since $t^* <s$, we get
	\begin{align*}
		\sum_{j=1}^m \sum_{t=1}^T t \cdot \lambda''_{t,j} < \sum_{j=1}^m \sum_{t=1}^T t \cdot \lambda'_{t,j} \,.
	\end{align*}
	If we can show that $c(\vec\lambda'',T) \geq a_T$ for all $T \in [\tau_1,\tau_2]$, we will contradict the minimality of $\vec\lambda'$. To see this, consider the following cases
	\begin{itemize}
		\item $T > T^*$: $\lambda'' \geq \lambda^{(T)}$ by definition of $T^*,t^*$ and $\nu$. Hence, \eqref{eqn:inter-2-opt-seq} implies
		\begin{align*}
			\sum_{t=1}^{T} \min\{\lambda''_{t,j},\rho_{T,j}\} \geq \sum_{t=1}^{T} \min\{\lambda^{(T)}_{t,j},\rho_{T,j}\} \geq a_T \cdot B_j \quad \forall\  j\in [m]\,,
		\end{align*}
		and consequently $c(\vec\lambda'', T) \geq a_T$.

		\item $T \leq T^*$: Observe that $\lambda''_{t^*,k} \leq \lambda^{(T)}_{t^*,k} \leq \rho_{T^*,k} \leq \rho_{T,k}$. Recall that $\lambda''_{t^*,k} = \lambda'_{t^*,k} + \nu$ and $\lambda''_{s,k} = \lambda'_{s,k} - \nu$ where $t^* < s$, and $\lambda''_{t,j} = \lambda'_{t,j}$ otherwise. Therefore,
		\begin{align*}
			\sum_{t=1}^{T} \min\{\lambda''_{t,j},\rho_{T,j}\} \geq \sum_{t=1}^{T} \min\{\lambda'_{t,j},\rho_{T,j}\} \geq a_T \cdot B_j \quad \forall\  j\in [m]\,,
		\end{align*}
		and consequently $c(\vec\lambda'', T) \geq a_T$.
	\end{itemize} 
	Thus we have established the required contradiction, thereby completing the proof.
\end{proof}

\subsection{Binary Search Procedure}

We explain how to use Algorithm~\ref{alg:opt-target-sequence} to find an $\varepsilon$-approximate solution to the LP in Proposition~\ref{prop:advice-LP}. A similar procedure can be used to compute an $\varepsilon$-approximate solution to the LP in Proposition~\ref{prop:uncertainty-window-LP}.

Consider a required level of competitiveness $\gamma \geq 0$. Then, we can run binary search to find the highest consistency that can be achieved by any target consumption sequence which is $\gamma$ competitive as follows:
\begin{itemize}
	\item Initialize $\ell = 0$ and $u = 1$
	\item Set $\beta = (u + \ell)/2$. Run Algorithm~\ref{alg:opt-target-sequence}. If it returns TRUE, set $\ell = \beta$, otherwise set $u = \beta$. Repeat this step till $u - \ell \leq \varepsilon$. 
\end{itemize}

Let $\vec\lambda'$ be the optimal solution of the LP in Proposition~\ref{prop:advice-LP}, i.e.,
\begin{align*}
	\vec\lambda' \in \argmax_{\vec\lambda \geq 0} c(\vec\lambda, T_P) \quad \text{ s.t. } \min_{T \in [\tau_1, \tau_2]} c(\vec\lambda, T) \geq \gamma \text{ and } \sum_{t=1}^{\tau_2} \lambda_t \leq B \,.
\end{align*}

Then, part (2) of Theorem~\ref{thm:faster-alg} implies that Algorithm~\ref{alg:opt-target-sequence} returns TRUE if and only if $\beta \leq c(\vec\lambda', T_P)$. Consequently, $\ell \leq c(\vec\lambda', T_P) \leq u$ at all times during the run of the binary search procedure, which further implies that $\ell \geq c(\vec\lambda', T_P) - \varepsilon$ at termination. Let $\vec\lambda^*$ be the sequence computed by Algorithm~\ref{alg:opt-target-sequence} when given $\beta = \ell$. Then, Theorem~\ref{thm:faster-alg} implies that 
\begin{align*}
	c(\vec\lambda^*, T_P) \geq c(\vec\lambda^*, T_P) - \varepsilon, \quad \min_{T \in [\tau_1, \tau_2]} c(\vec\lambda^*, T) \geq \gamma \quad \text{ and } \sum_{t=1}^{\tau_2} \lambda^*_t \leq B \,,
\end{align*}
as required.

\end{document}